\documentclass[12pt]{article} 
\usepackage{etruscan}
\pdfoutput=1
\usepackage[authoryear]{natbib}
\usepackage{graphicx,epstopdf}
\usepackage{setspace}
\usepackage[margin=2.5cm]{geometry}
\usepackage{array}
\usepackage{hyperref}
\usepackage{xcolor}
\usepackage{xfrac}

\usepackage{pdflscape}
\usepackage{graphicx}
\usepackage[utf8]{inputenc}
\usepackage[english]{babel}
\usepackage[normalem]{ulem}
\usepackage{amsmath, amsthm, amssymb, amsfonts}
\usepackage{url} 
\usepackage{comment}
\usepackage{booktabs}

\usepackage{graphicx}
\usepackage{subcaption}

\makeatletter

\theoremstyle{definition}
\newtheorem{defn}{\protect\definitionname}
\theoremstyle{plain}
\newtheorem{lem}{\protect\lemmaname}
\theoremstyle{plain}
\newtheorem{thm}{\protect\theoremname}
\theoremstyle{plain}
\newtheorem{prop}{\protect\propname}
\theoremstyle{plain}
\newtheorem{cor}{\protect\corollaryname}
\theoremstyle{plain}
\newtheorem{assumption}{\protect\assumptionname}

\usepackage{pdflscape}
\makeatother

\usepackage{babel}
\providecommand{\assumptionname}{Assumption}
\providecommand{\corollaryname}{Corollary}
\providecommand{\definitionname}{Definition}
\providecommand{\lemmaname}{Lemma}
\providecommand{\propname}{Proposition}
\providecommand{\theoremname}{Theorem}

\usepackage{thmtools}

\usepackage{float}
\floatstyle{ruled}
\newfloat{algorithm}{tbp}{loa}
\providecommand{\algorithmname}{Algorithm}
\floatname{algorithm}{\protect\algorithmname}

\usepackage{ifpdf}
\ifpdf
  \DeclareGraphicsExtensions{.pdf,.png,.jpg}
\else
  \DeclareGraphicsExtensions{.eps}
\fi

\usepackage{natbib}
\setcitestyle{authoryear,open={(},close={)}}

\usepackage{enumitem}
\usepackage[section]{placeins}
\usepackage{framed}

\makeatletter
\providecommand{\sf@counterlist}{}
\makeatother

\begin{document}

\title{Homophily and transitivity in dynamic network formation}
\author{Kevin Dano, Bryan S. Graham and Yassine Sbai Sassi\thanks{\underline{Dano}: Department of Economics, Princeton University, Washington Rd, Princeton, NJ 08544, e-mail: \url{kdano@princeton.edu}, web:  \href{https://kevindano.github.io}{https://kevindano.github.io}. \newline  \underline{Graham}: Department of Economics, University of California - Berkeley, 530 Evans Hall \#3380, Berkeley, CA 94720-3880 and National Bureau of Economic Research, e-mail: \url{bgraham@econ.berkeley.edu}, web:  \href{http://bryangraham.github.io/econometrics/}{http://bryangraham.github.io/econometrics/}. \newline \underline{Sbai Sassi}: Department of Economics, New York University 19 West 4th Street, 6th Floor New York, NY 10012, e-mail: \url{ys6799@nyu.edu}, web:  \href{https://yassinesbaisassi.github.io/}{https://yassinesbaisassi.github.io/}. \newline This paper substantially develops and extends ideas initially reported in the second author's working paper ``Homophily and transitivity in dynamic network formation". That paper, in turn, expanded upon comments prepared in response to Guido Imbens' \emph{Journal of Business and Economic Statistics Annual Lecture} delivered at the American Economic Association meetings in Chicago (January 7th, 2012). We are grateful to Guido, and then JBES editors Kei Hirano and Jonathan Wright, for providing the initial opportunity to consider the ideas explored in this paper. We thank current JBES editor, Michal Kolesár, for the opportunity to develop this material further. This paper additionally develops previously unreported findings discovered by the first author in the context of dissertation research, a counter-example discovered by the third author, as well as substantial new work done by all three authors collaboratively in preparation for the 2026 JBES lecture in Philadelphia. We are also grateful to Peter Bickel, Paul Goldsmith-Pinkham, Jim Heckman, Bo Honore, Noureddine El Karoui, Michael Leung as well as to participants at the Measuring and Interpreting Inequality (MIE) Inaugural Conference (February 18th, 2012), the CEME Conference on Inference in Nonstandard Problems (June 15th \& 16th, 2012), the invited Social Interactions session of the ESEM (August 27 - 31, 2012), the CEMMAP conference on Estimation of Complementarities in Matching and Social Networks (October 5 \& 6, 2012), the 9th Invitational Choice Symposium (June 12 -- 16, 2013), the Berkeley Statistics Department's NSF reading group on networks, the Third European Meeting on Networks (June 18-19th, 2015), the USC-INET Conference on Networks (November 20-21st, 2015), the 2026 AEA/ES Winter Meetings, Structural Econometrics and Models of Strategic Interactions Conference in Toulouse (May 28th \& 29th, 2026) and the UPENN Econometric Seminar for valuable comments. Financial support, for the second author, from NSF Grant SES \#1357499 is gratefully acknowledged. Generative AI was used to catch typos, notational inconsistencies and check proofs. We used Gemini 3.1 Pro and GPT-5/6 for these purposes. The authors report there are no competing interests to declare. All the usual disclaimers apply.}} 
	
\date{\today 	}
	
\maketitle
	\thispagestyle{empty}
\newpage
\begin{abstract}
\begin{singlespace}
\noindent{\footnotesize In social and economic networks linked agents
often share connections in common. There are two competing explanations
for this phenomenon. First, agents may have a structural taste for
transitive links -- the returns to linking may be higher if two agents
share a common connection. Second, agents may assortatively match on unobserved
attributes, a process called homophily. We study parameter identifiability
in a simple model of dynamic network formation with both effects.
Agents form, maintain, and dissolve links over time to maximize
utility. The return to linking may be higher if agents share connections
in common. A pair-specific utility component allows for arbitrary
homophily on time-invariant agent attributes. We derive conditions
under which it is possible to detect the presence of a taste for transitivity
in the presence of assortative matching on unobservables. We leave
the joint distribution of the initial network and the pair-specific
utility component, both high dimensional nuisance parameters, unrestricted. Our identification result is constructive, suggesting an analog estimator, whose finite and (single) large network properties we characterize. We show, via examples, the delicacy of information accumulation in the single (large) network setting.\smallskip{}
}{\footnotesize\par}

\noindent{\footnotesize\uline{JEL Codes:}}{\footnotesize{} C31, C33,
C35}{\footnotesize\par}

{\footnotesize\smallskip{}
}{\footnotesize\par}

\noindent{\footnotesize\uline{Keywords}}{\footnotesize : Strategic
Network Formation, Homophily, Transitivity, Conditional Likelihood, Triads, Assortative Matching, Initial Conditions Problem, Panel Data, Fixed Effects, Incidental Parameters}{\footnotesize\par}
\end{singlespace}
\end{abstract}
	
\newpage
	
\pagenumbering{arabic}
\onehalfspacing
\renewcommand\thmcontinues[1]{Continued}

\newpage

Let $i,j$ and $k$ index the agents of a \emph{triad}, drawn at random from a simple undirected network. Let $D_{ij}=1$ if agents $i$ and $j$ in this triad are connected and zero otherwise.\footnote{Since links are
undirected, $D_{ij}=D_{ji}$ for all pairs $ij$. Self-ties
are also ruled out, so that $D_{ii}=0$ for all $i$.} The triad $ijk$ must be wired, up to isomorphisms, in one of the four ways shown in Figure \ref{fig: Triad-types}. The two-star and triangle configurations interest us here. Consider the probability that edge $ij$ is present \emph{conditional} on the presence of edges $ik$ and $jk$; that is, the probability that agents $i$ and $j$ are connected given that they share $k$ as a mutual friend:
\begin{equation}
\rho_{\mathrm{CC}}=\Pr\left(\left.D_{ij}=1\right|D_{ik}=1,\thinspace D_{jk}=1\right).\label{eq: transitivity_index}
\end{equation}
The sample analog of (\ref{eq: transitivity_index}) is called the
transitivity index or global clustering coefficient in the networks
literature \citep{Harary_Kommel_JMS1979,Dekker_et_al_NS2019}. In real world social and economic
networks $\rho_{\mathrm{CC}}$ is generally higher, often substantially so, than $\rho_{\mathrm{D}}=\Pr\left(D_{ij}=1\right),$
the unconditional frequency at which agents link \citep{Jackson_NetBook08}.
Networks exhibit \emph{transitivity}: agents link more frequently if they have connections in common (``the friend of my friend is also my friend'').

\begin{figure}[tbh]
\caption{Triad configurations in undirected networks\label{fig: Triad-types}}

\centering{}\includegraphics{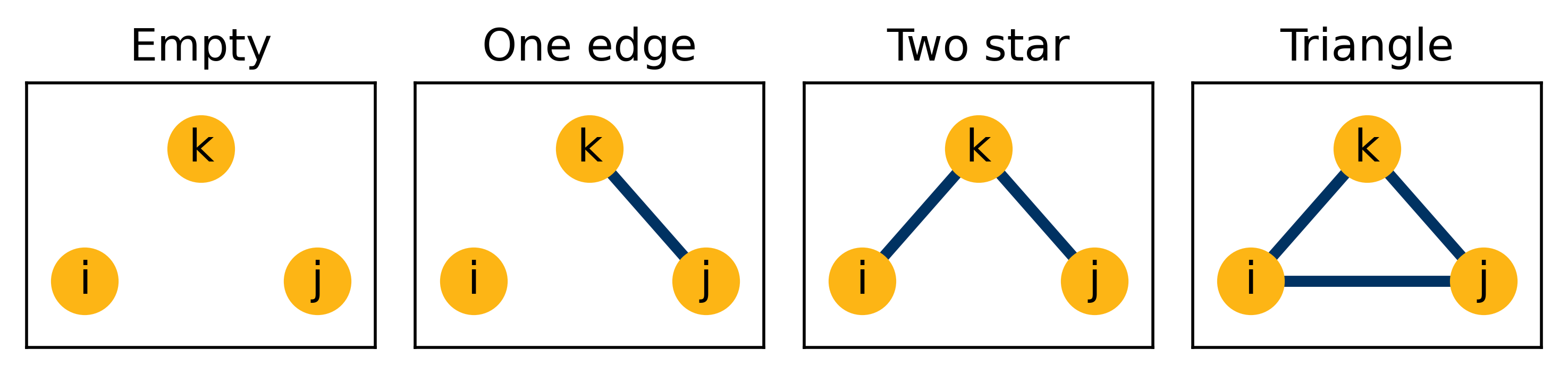}
\end{figure}    

Measured transitivity reflects two distinct phenomena. First, agents
may \emph{prefer} transitive ties. That is, the payoff associated with a link between any two agents may increase with the number of neighbors they share in common. \citet{Coleman_Book1990} argues that closed subnetworks, like triangles, facilitate the formation of social and human
capital.
\citet{Jackson_et_al_AER12} provide a game-theoretic foundation
for transitivity, arguing that common friends, by monitoring transactions
between agents, help to sustain cooperation. More mundanely socializing
may be easier and more enjoyable when individuals share friends in
common.

Second, transitivity may reflect assortative matching
on unobserved attributes. Link payoffs may be
greater across similar agents, leading to dense ties among them. The
tendency for individuals to assortatively match on gender, race, age
and other observed characteristics, called \emph{homophily} by network
researchers, is widely documented \citep{McPherson_et_al_ARS01}. 

Understanding how to detect a structural taste for transitive ties in the presence of unobserved homophily is the primary subject of this paper. Distinguishing between a taste for transitivity and homophily as drivers of clustering in networks is of considerable scientific
interest and policy relevance. When agents prefer transitive links, temporary outside interventions can induce long-run changes in network structure. Whereas, if clustering primarily reflects assortative matching on unobserved agent attributes, such interventions generally do not generate
durable changes in network structure \citep{Graham_HBE2020}. 

There is a useful analogy between our research problem, both in terms
of scientific motivation and technical content, and that
of discriminating between state dependence and unobserved heterogeneity
in single-agent dynamic binary choice analysis \citep*{Cox_JRSS58,Heckman_AI78,Heckman_SADD81a,Heckman_SADD81b,Heckman_SLM81,Chamberlain_LALMD85}.
The multi-agent aspect of our problem generates new challenges, nevertheless
we will utilize intuitions and proof strategies from research in this area \citep[e.g.,][]{Honore_Kyriazidou_EM00}.

Section \ref{sec: model}, which appears next, introduces a dynamic model of network formation. Our analysis presumes the researcher observes only a \emph{single} network, consisting of $N$ agents, for a small number of time periods, $T$.  The model features two high-dimensional nuisance parameters: (i) the $\binom{N}{2}$ dyad-specific unobserved `incidental' parameters $\mathbf{A}=\left[A_{ij}\right]_{1\leq i \neq j \leq N}$, which capture heterogeneity in the payoff to link-formation across pairs of agents,  and (ii) the probability mass function for the initial network (which assigns probabilities to the $2^{\binom{N}{2}}$ possible initializations of the observed sequence of networks).

Learning from a single network sequence is challenging. The presence of cross-dyad dependence in link formation, both for structural reasons and via (unmodeled) dependence across the elements of $\mathbf{A}$, complicates identification, estimation and inference. The likelihood for the observed network sequence does not naturally factor into a product of independent multiplicands; how information grows, for example, with the number of agents, $N$, is not readily apparent \citep[cf.,][]{Goldsmith-Pinkham_Imbens_JBSE13}. Although our results are specific to a particular data generating process, we expect that several of the methods introduced below could be adapted to analyze other models of dynamic network formation (in fixed $T$ settings).

We proceed as follows. Section \ref{sec: conditional_likelihood} derives a conditional likelihood for the observed network sequence that is invariant to both the dyad-specific heterogeneity and initial network condition. This likelihood is complicated; posing both computational and conceptual challenges. We therefore consider a more restrictive conditioning event. The associated conditional likelihood in this case \emph{is} computationally and analytically tractable. With some work, it can be factored into a product of conditionally independent multiplicands. This facilitates likelihood evaluation, as well as formal statistical analysis.

Section \ref{sec: inference}, building on the results of Section \ref{sec: conditional_likelihood}, introduces two approaches to conducting inference on the common structural parameters -- namely the coefficients indexing, respectively, the strength of structural state-dependence, $\alpha$, and agents' taste for transitivity, $\beta$. The first approach is \emph{exact}, as in \cite{Fisher_JRSS1922}. We propose a simple Metropolis-Hasting algorithm for enumeration of the elements of the conditioning set (and hence null distribution simulation). The second approach involves an asymptotic approximation as the number of agents in the \emph{single} observed network sequence grows large. Information accumulation in this setting is delicate. An upshot of our formal analysis is that it provides insight into what types of real-world settings our methods are \emph{a priori} mostly likely to be of value. 

We further explore whether and how information accumulates as the network grows via two examples. In the first example the $(t=0)$ initial network is an Erdős-Renyi random graph. The network evolves in subsequent periods ($t=1,2,3$) according to our dynamic model. We assume that $A_{ij}$ lies in a bounded interval for all dyads. In this setting, we show that the probability of our conditioning event shrinks to zero as $N$ grows large. Consistent estimation is not possible, although informative finite sample inference may be.

In a second example, we use a random geometric graph structure (see, for example,\cite{Penrose_Bk2003}) to impose a strong level of homophily on unobservables. In this example, any pair of agents that are ``far enough" apart never link. We show that information \emph{does} accumulate with $N$ in this setting, supporting consistent estimation.

While our identification argument imposes no restrictions on the distribution of unobserved heterogeneity, a researcher's ability to consistently estimate model parameters, as these examples illustrate, does vary with features of this distribution. This is an old, albeit underappreciated, observation in the context of single-agent panel data analysis (see, for example, \cite{Andersen_JRSS70}). In our more complicated multi-agent network setting, the heterogeneity distribution appears to more overtly limit statistical learning.

Section \ref{sec: Monte_Carlo_experiments} explores the finite sample properties of our methods via a series of Monte Carlo experiments. These experiments closely track the predictions of our asymptotic theory. In Section \ref{sec : empirical_application} we revisit the co-authorship network compiled by \cite{Ductoretal2014}. As a ``proof-of-concept" we use these data, and our methods, to explore the initiation and dissolution of co-authorship among researchers.

In follows we interchangeably use unit, node, vertex, agent and
individual to refer to the $i=1,\ldots,N$ vertices of the sampled network
or graph. We denote random variables by capital Roman letters, specific
realizations by lower case Roman letters, and their support by blackboard bold
Roman letters. That is $Y$, $y$ and $\mathbb{Y}$ respectively denote a generic
random draw of, a specific value of, and the support of, $Y$. In general, bold-facing, denotes a matrix-valued random variable, $\mathbf{D}$, or its sample realization, $\mathbf{d}$. We also use the shorthand $\Pr\left(\left.y\right|x\right)=\Pr\left(\left.Y=y\right|X=x\right)$ throughout. As is standard in the literature on U-Statistics, we use the shorthand $\sum_{i<j}$ for $\sum_{i=1}^{N-1}\sum_{j=i+1}^{N}$. 

\section{A dynamic network formation model} \label{sec: model}

Let $\mathcal{V} =\left\{1,\ldots,N\right\}$ be the set of vertices (agents) on which a graph of undirected edges $\mathcal{E}_t \subseteq \mathcal{V} \times \mathcal{V} \ \setminus \ \Delta(\mathcal{V})$ forms.\footnote{Here $\Delta(\mathcal{V})$ represents the diagonal set of disallowed self-loops}. We assume the graph
is observed in each of periods $t=0,1,\ldots,T.$ The adjacency matrix $\mathbf{D}_t \overset{def}{\equiv} \left[D_{ijt}\right]_{1\leq i,j \leq N}$ records the presence or absence of all possible edges in the period $t$ graph $G_t=\mathcal{G}\left(\mathcal{V},\mathcal{E}_t\right)$. This matrix is symmetric, with zeroes on its diagonal (since self-loops are not allowed). We let $\mathbb{D}_{N}$ denote the set of all $2^{ \binom{N}{2}}$ possible $N \times N$ adjacency matrices.

Our analysis leaves the initial condition of the network, $\mathbf{D}_0$, unmodeled. For periods $t=1,\ldots,T$ we assume that agents add, subtract, or maintain links in order to maximize a per-period payoff function. Let $\nu_{i}\thinspace:\thinspace\mathbb{D}_{N}\rightarrow\mathbb{R}$ denote agent $i$'s payoff function. This function maps adjacency matrices – equivalently networks – into the hedonimeter-measurable utils of Francis Ysidro Edgeworth, as described in his seminal book \emph{Mathematical Psychics},
\begin{equation} \label{eq: utility_per_period}
\nu_{i}\left(\left.\mathbf{d}_t\right|\mathbf{D}_{t-1},\mathbf{\tilde{U}}_t;\mathbf{\tilde{A}};\theta_{0}\right)=\sum_{j}d_{ijt}\left[\frac{\alpha_{0}}{2}D_{ijt-1}+\frac{\beta_{0}}{2}\left(\sum_{k}D_{ikt-1}D_{jkt-1}\right)+\tilde{A}_{ij}-\tilde{U}_{ijt}\right],
\end{equation}
with $\mathbf{\tilde{U}}_t \overset{def}{\equiv} \left[\tilde{U}_{ijt}\right]_{1\leq i \neq j \leq N}$ a matrix of \emph{time-varying} random utility shifters, $\mathbf{\tilde{A}} \overset{def}{\equiv} \left[\tilde{A}_{ij}\right]_{1\leq i \neq j \leq N}$ a matrix of \emph{time-invariant} dyad-specific heterogeneity, and $\theta=\left(\alpha,\beta\right)' \in \Theta$ the target parameter of interest. While agents observe $\mathbf{\tilde{A}}$ and $\mathbf{\tilde{U}}_t$, the econometrician does not. The goal is to learn the value of $\theta$ from the observed sequence of networks $\mathbf{D}_{0}, \mathbf{D}_{1}, \ldots, \mathbf{D}_{T}$.

Specification \eqref{eq: utility_per_period} implies that $i$ attaches greater utility to a link with $j$ in period $t$ if (i) they
were connected in the prior period and (ii) they shared friends in common in the prior period (the term $\sum_{k}D_{ikt-1}D_{jkt-1}$ equals a count of the number of agents $k$ which were connected to both $i$ and $j$ in period $t-1$).\footnote{For simplicity we assume that $\alpha>0$ and $\beta>0$ during the exposition, although this is not a restriction of the model.} 

We assume that the random utility \emph{pairs} $\left\{\tilde{U}_{ijt},\tilde{U}_{jit}\right\}_{1\leq i<j\leq N, t=1,\dots,T}$ are independently and identically distributed across dyads as well as over time. Below we make additional parametric assumptions on this distribution. In contrast we treat $\mathbf{\tilde{A}}$ as a random matrix whose distribution is left nonparametric.\footnote{Much of our analysis follows when $\mathbf{\tilde{A}}$ is instead viewed as a nonstochastic (and very high-dimensional) parameter. }$^,$\footnote{Although we place no formal restrictions on the distribution of $\mathbf{A}$, its realized value will influence the informativeness of our exact inference methods, and its underlying distribution the rate at which information accumulates in our asymptotic analysis. Similar issues arise in single agent binary choice panel data models. See, for example, \cite{Andersen_JRSS70} and \cite{Chamberlain_ReStud80}.} Consequentially, our setup accommodates complex forms of dependence across the elements of $\mathbf{\tilde{A}}$. It may be, for example, that agent $i$ is an extrovert and generically enjoys connecting with others, such that $\tilde{A}_{ij}$ and $\tilde{A}_{ik}$ positively co-vary. Alternatively, if $i,j$ and $k$ all share a latent attribute upon which agents homophilously sort, then $\tilde{A}_{ij}$, $\tilde{A}_{ik}$ and $\tilde{A}_{jk}$ will co-vary. Dyadic ``clustering" of this type arises frequently in network settings \citep{Fafchamp_Gubert_JDE07,Graham_HBE2020}. 

The marginal utility received by agent $i$ when edge $ij$ is added to, or subtracted from, $\mathbf{D}_t$ equals
\[
MU_{ij}\left(\mathbf{D}_t\right)=\left\{ \begin{array}{cc}
\nu_{i}\left(\mathbf{D}_t\right)-\nu_{i}\left(\mathbf{D}_t-ij\right) & \text{if}\thinspace D_{ijt}=1\\
\nu_{i}\left(\mathbf{D}_t+ij\right)-\nu_{i}\left(\mathbf{D}_t\right) & \text{if}\thinspace D_{ijt}=0
\end{array}\right.
\]
where $\mathbf{D}_t - ij$ is the adjacency matrix associated with the network obtained after deleting edge $ij$ and $\mathbf{D}_t + ij$ the one obtained via link addition (and we suppress the dependence of $\nu_i$ on terms other than $\mathbf{D}_t$)

We assume that utility is transferable across agents, such that $i$ and $j$ may freely share any surplus generated by an edge between them.\footnote{\cite{Goldsmith-Pinkham_Imbens_JBSE13} study a dynamic network formation model with non-transferrable utility.} This assumption allows us to invoke the pairwise stability with transfers equilibrium concept of \cite{Bloch_Jackson_IJGT06}.

\begin{defn}
    {\label{def: pairwise_stability} \textsc{(Pairwise stability with Transfers)} The network $G\left(\mathcal{V},\mathcal{E}_t\right)$ is pairwise stable with transfers if \\ 
    (i) $\forall ij \in\mathcal{E}_t\left(G\right),\thinspace MU_{ij}\left(\mathbf{D}_t\right)+MU_{ji}\left(\mathbf{D}_t\right)\geq0$ \\
    (ii) $\forall ij \notin\mathcal{E}_t\left(G\right),\thinspace   MU_{ij}\left(\mathbf{D}_t\right)+MU_{ji}\left(\mathbf{D}_t\right)<0$}    
\end{defn}
Pairwise stable networks are those where any dyad $ij$, on net, weakly benefits from edge $D_{ijt}$ when present and \emph{would not} benefit from adding it to the network when not present (otherwise the dyad could form the edge and split the surplus).

When the utility function is of the form given in \eqref{eq: utility_per_period} the marginal utility agent $i$ gets from a link with $j$ equals
\[
MU_{ij}\left(\mathbf{d}_t\right)=\frac{\alpha_{0}}{2}D_{ijt-1}+\frac{\beta_{0}}{2}\left(\sum_{k}D_{ikt-1}D_{jkt-1}\right)+\tilde{A}_{ij}-\tilde{U}_{ijt}. 
\]
Definition \ref{def: pairwise_stability} then implies that -- conditional on  $\mathbf{\tilde{A}}$, the realized values $\left\{\tilde{U}_{ijt},\tilde{U}_{jit}\right\}_{1\leq i<j\leq N}$, and the magnitude of $\theta_0$ -- the networks must each satisfy, for $i=1,\ldots,N-1$ and $j=i+1,\ldots,N$,
\begin{equation} \label{eq:ij_Link_Model}
    D_{ijt}=\mathbf{1}\left(\alpha_0 D_{ijt-1}+\beta_{0}\left(\sum_{k}D_{ikt-1}D_{jkt-1}\right) + A_{ij} \geq U_{ijt}\right),\ t=1,\ldots,T,
\end{equation}
where $A_{ij}\overset{def}{\equiv}\tilde{A}_{ij}+\tilde{A}_{ji}$ and $U_{ijt}\overset{def}{\equiv}\tilde{U}_{ijt}+\tilde{U}_{jit}$. 

Observe that since $\left(\tilde{U}_{ijt},\tilde{U}_{jit}\right)$ is distributed independently and identically across dyads and over time, so is $U_{ijt}$. Let $F_{U}\left(u\right)$ denote the distribution function of $U_{ijt}$, we have 
\begin{equation}
F\left(U_{121},\ldots,U_{12T},\ldots,U_{N-1N1},\ldots,U_{N-1NT}\right)=\prod_{i<j}\prod_{t=1}^{T}F_{U}\left(U_{ijt}\right).\label{eq:U_ijt_Independence}
\end{equation}
Throughout we assume that $F_{U}\left(u\right)$ is strictly increasing on $\mathrm{\mathbb{R}}^{1}$. 

Finally, since \eqref{eq:ij_Link_Model} only applies to
periods $t=1,\ldots,T$ we leave the joint distribution of $\left(\mathbf{D}_{0},\mbox{\textbf{A}}\right)$,
the initial condition of the heterogeneity, unrestricted:
\begin{equation}
\left(\mathbf{D}_{0},\mbox{\textbf{A}}\right)\sim\Pi_{0}.\label{eq:Initial_Condition}
\end{equation}
The mass-density function evaluated at $\mathbf{D}_{0}=\mathbf{d}_{0}$,
$\mathbf{A}=\mathbf{a}$ is denoted by $\pi_{0}\left(\mathbf{d}_{0},\mathbf{a}\right)$.

\subsection*{Some observations}
Before proceeding, it is helpful to make a few observations about the network formation process described by equations \eqref{eq:ij_Link_Model},  \eqref{eq:U_ijt_Independence} and \eqref{eq:Initial_Condition}. First, the nuisance parameter space is very ``large"; $\mathbf{A}$ consists of $\binom{N}{2}$ real-valued unknown heterogeneity terms, while $\Pi_0$ assigns probability mass to all possible $2^{\binom{N}{2}}$ initial network configurations. In contrast, the target parameter, $\theta$, consists of just two elements.

Second, the network evolves according to a heterogenous Markov process.\footnote{While this process is heterogeneous due to $\mathbf{A}$, it is time homogenous.}
\[
    \Pr\left(\left.\mathbf{D}_{t}=\mathbf{d}_{t}\right|\mathbf{D}_{t-1}=\mathbf{d}_{t-1},\ldots,\mathbf{D}_{0}=\mathbf{d}_{0},\mathbf{A}\right)=\Pr\left(\left.\mathbf{D}_{t}=\mathbf{d}_{t}\right|\mathbf{D}_{t-1}=\mathbf{d}_{t-1},\mathbf{A}\right)
\]
The process is Markov conditional on $\mathbf{A}$, but will generally not be unconditionally. Third, and relatedly, the network follows a type of nonlinear vector-autoregressive process; albeit one with substantial structure on how $\mathbf{D}_{t-1}$ influences $\mathbf{D}_{t}$.

Fourth, the model parsimoniously includes three forces
long-hypothesized by researchers as important drivers of link formation (cf.,
\citealp{Snijders_AR11}). First, there is \emph{state
dependence}: agents $i$ and $j$ are more likely to form a link in period $t$ if they were connected in period $t-1$ (parameterized by $\alpha$). Second, agents link more frequently when doing so generates ``triadic closure": links which ``close" two-stars left open in the prior period are more likely (parameterized by $\beta$). Third, links may form because of unobserved good `fundamentals' (i.e., $A_{ij}$ is high). 

One source of `good fundamentals' is that the pair $ij$ might be
similar in some salient unobserved dimension. The tendency for individuals
to assortatively match on various \emph{observed} characteristics
is well documented \citep{McPherson_et_al_ARS01}. Here $A_{ij}$
might reflect utility from assortative matching on \emph{unobserved}
attributes. To make this idea concrete let $\xi_{i}$ be a vector
of latent agent-specific characteristics and $g\left(\xi_{i},\xi_{j}\right)$
a measure of the distance between $i$ and $j$ (i.e., $g\left(\cdot,\cdot\right)$
is a distance function). If $A_{ij}=-g\left(\xi_{i},\xi_{j}\right)$,
then rule (\ref{eq:ij_Link_Model}) implies that a link between $i$
and $j$ is more likely if they are similar in terms of $\xi$. Note
$A_{ij}$ could reflect more than just homophily. For example, setting
\begin{equation}
A_{ij}=\lambda_{i}+\lambda_{j}-g\left(\xi_{i},\xi_{j}\right),\label{eq:KHRH}
\end{equation}
allows for the possibility that certain individuals may uniformly generate higher
friendship surplus. Put differently $\lambda_{i}$, and hence $A_{ij}$,
might be high because individual $i$ is a `good friend'. Effects
of this type give rise to \emph{degree heterogeneity} or variation
in the number of links maintained by different individuals \citep[e.g.,][]{Krivitsky_et_al_SN09, Graham_EM17}. Degree heterogeneity is also a feature of many real-world networks.

Both a \emph{structural} taste for transitivity in relationships,
here parameterized by $\beta_0$, and homophily on unobserved attributes, captured by $\mathbf{A}$, can generate high levels of clustering in networks. Heuristically, this is why identification and estimation are challenging in our setting.

Although we utilize an equilibrium notion drawn from game theory (Definition \ref{def: pairwise_stability} above), agents in our model -- while maximizing utility -- are not playing a game (at least a non-trivial one). In each period, agents condition their actions on the structure of the network in the \emph{prior} period and add, subtract, or maintain links myopically. Extending the work reported here to accommodate forward-looking agents, and the game-theoretic difficulties that doing so would raise, merits exploration.

As an alternative to specifying a dynamic game, one might ask what can be learned from a single network observation. \cite{Pelican_Graham_ReStud2026} and \cite{Graham_Pelican_BookCh2020} study static, complete information, games of network formation in the presence of agent-level heterogeneity. While they report positive results, doing so requires much stronger assumptions on the structure of unobserved heterogeneity, $\mathbf{A}$, than we maintain here. Without restrictions on $\mathbf{A}$, \emph{any} network configuration can be generated by an appropriate draw of $\mathbf{A}$. Our results circumvent this problem by exploiting the availability of multiple network observations over time and the assumed \emph{time-invariance} of $\mathbf{A}$.

\subsection*{Likelihood}

Equations \eqref{eq:ij_Link_Model}, \eqref{eq:U_ijt_Independence} and \eqref{eq:Initial_Condition} specify a dynamic model of network
formation. To ease the notation, let $R_{ijt}\overset{def}{\equiv}\sum_{k}D_{ikt}D_{jkt}$. The joint probability mass-density of a specific sequence of
network configurations $\mathbf{D}^{T}\overset{def}{\equiv}\left(\mathbf{D}_{0},\mathbf{D}_{1},\ldots,\mathbf{D}_{T}\right)$ and realization of $\mathbf{A}$ is
\begin{eqnarray}
    p\left(\mathbf{d}^{T},\mathbf{a};\theta,\pi\right) & = & \pi\left(\mathbf{d}_{0},\mathbf{a}\right)\nonumber \\
    &  & \times\prod_{i<j}\prod_{t=1}^{T}F\left(\alpha d_{ijt-1}+\beta r_{ijt-1}+a_{ij}\right)^{d_{ijt}}\nonumber \\
    &  & \times\left[1-F\left(\alpha d_{ijt-1}+\beta r_{ijt-1}+a_{ij}\right)\right]^{1-d_{ijt}}.\label{eq: likelihood}
\end{eqnarray}
One approach to estimating $\theta$ would involve (i) assuming $F_U$ is known (e.g., $F_U\left(u\right)=\Phi\left(u\right)$ with $\Phi\left(u\right)$ the CDF of a standard normal random variable) and (ii) additionally
parameterizing $\pi\left(\mathbf{d}_0,\mathbf{a}\right)$. This is the approach taken by \cite{Goldsmith-Pinkham_Imbens_JBSE13} in their own \emph{Journal of Business and Economic Statistics} lecture over 10 years ago (a lecture that directly inspired the beginning of the research reported in this paper). Let $\pi_{\left.\mathbf{D}_{0}\right|\mathbf{A}}\left(\left.\mathbf{d}_{0}\right|\mathbf{a};\eta_1\right)$ denote a parametric family, indexed by $\eta_1$, for $\Pr\left(\left.\mathbf{D}_{0}=\mathbf{d}_{0}\right|\mathbf{A}=\mathbf{a}\right)$, the distribution of the initial network given $\mathbf{A}=\mathbf{a}$. Let $\pi_{\mathbf{A}}\left(\mathbf{a}; \eta_2\right)$ denote a parametric family, indexed by $\eta_2$, for the marginal
density function for $\mathbf{A}$. The integrated likelihood for
the observed data is then
\begin{eqnarray}
p^{\mathrm{I}}\left(\mathbf{d}^{T};\theta,\eta\right) & =\displaystyle\int \ldots \int & \prod_{t=1}^{T}\prod_{i<j}\left\{ F_{U}\left(\alpha d_{ijt-1}+\beta r_{ijt-1}+a_{ij}\right)^{d_{ijt}}\right.\nonumber \\
 & & \left.\times\left[1-F_{U}\left(\alpha d_{ijt-1}+\beta r_{ijt-1}+a_{ij}\right)\right]^{1-d_{ijt}}\right\} \nonumber \\
 & & \times\pi_{\left.\mathbf{D}_{0}\right|\mathbf{A}}\left(\left.\mathbf{d}_{0}\right|\mathbf{a};\eta_1\right)\pi_{\mathbf{A}}\left(\mathbf{a};\eta_2\right)\mathrm{d}a_{21},\ldots,\mathrm{d}a_{NN-1}\label{eq:obs_likelihood}
\end{eqnarray}
with $\eta=\left(\eta_1',\eta_2'\right)'$. Although \eqref{eq:obs_likelihood} represents a parametric family for the observed data, using it as
a basis for estimation and inference raises at least three conceptual and computational issues.

First, as is familiar from prior work on dynamic discrete choice analysis,
rule \eqref{eq:ij_Link_Model} provides no guidance on how to specify
$\pi_{\left.\mathbf{D}_{0}\right|\mathbf{A}}\left(\left.\mathbf{d}_{0}\right|\mathbf{a},\eta_1\right)$
\citep{Heckman_SLM81,Honore_Tamer_EM06}. In analogy to its panel counterpart, we call this the
\emph{initial network problem}. Rule \eqref{eq:ij_Link_Model} does
suggest that the probability of the event $\mathbf{D}_{0}=\mathbf{d}_{0}$
should vary with the realized value of $\mathbf{A}$, but little else.
One approach, again inspired by single agent models, would involve
assuming that $\pi_{\left.\mathbf{D}_{0}\right|\mathbf{A}}\left(\left.\mathbf{d}_{0}\right|\mathbf{a},\eta_1\right)$
coincides with the steady state distribution implied by (\ref{eq:ij_Link_Model})
(e.g., \citealp*{Heckman_SADD81a,Card_Sullivan_EM88}). Even if this
is empirically plausible, operationalizing it would be non-trivial. 

A second problem is that, even if $\pi_{\left.\mathbf{D}_{0}\right|\mathbf{A}}\left(\left.\mathbf{d}_{0}\right|\mathbf{a},\eta_1\right)$
were correctly specified, evaluating (\ref{eq:obs_likelihood}) requires
computing a very high-dimensional integral. Since it seems reasonable,
at the very minimum, to choose a specification for $\pi_{\mathbf{A}}\left(\mathbf{a};\eta_2\right)$
which allows $A_{ij}$ and $A_{ik}$ to covary, there is no obvious
way to factor (\ref{eq:obs_likelihood}) to reduce the dimensionality
of the required integration. \citet*{Goldsmith-Pinkham_Imbens_JBSE13}
assume that $A_{ij}=\alpha_{\xi}\left|\xi_{i}-\xi_{j}\right|$ with
$\xi_{i}$ binary, independent of $\xi_{j}$, and $\Pr\left(\xi_{i}=1\right)=p$.
Directly evaluating the integrated likelihood in this case involves
a weighted sum of the likelihood given $\mathbf{A}$ over its $2^{N}$
possible realizations.\footnote{\citet{vanDuijn_et_al_SN04} introduce MCMC methods that also might
be adapted in order to evaluate \eqref{eq:obs_likelihood} in special
cases.}

A third problem, related to the second, and emphasized by \citet*{Goldsmith-Pinkham_Imbens_JBSE13}, is that even if the maximum likelihood estimate could be computed,
it would not be clear how to conduct large sample inference using
$\left(\hat{\alpha}_{ML},\hat{\beta}_{ML}\right)$; at least with
data drawn from a single network. This motivates their recourse to
Bayesian methods, which are attractive for computational reasons and
also for providing a principled approach to inference. The main approach
to inference developed below, in contrast, is frequentist (albeit one that involves a delicate asymptotic argument).

\section{Conditional likelihood analysis} \label{sec: conditional_likelihood}
While integrated likelihood analysis merits further exploration, our focus will instead lie with methods which remain valid irrespective of the precise form of $\pi$ and/or realized configuration of $\mathbf{A}$: so called \emph{fixed effects} methods.

Hereon we maintain a logistic assumption on the $\left\{U
_{ijt}\right\}_{i<j,t=1,\ldots,T}$.\footnote{That is, the CDF of $U_{ijt}$ equals $\Pr\left(U_{ijt} \leq u\right)=F_U(u)=\frac{\exp\left(u\right)}{1+\exp\left(u\right)}$.} For simplicity we also focus on the $T=3$ case.\footnote{By aggregating across (possibly overlapping) sub-panels of length four, we can accommodate longer network sequences. However, fully exploiting the identifying content of longer panels would require substantively extending the results which follow.} Under the logistic assumption, reminiscent of \cite{Chamberlain_LALMD85}, we can write the likelihood of the event $\mathbf{D}^{3}=\mathbf{d}^{3}$ as
\begin{multline} \label{eq: likelihood_chamberlain}
    \Pr\left(\mathbf{D}^{3} =\mathbf{d}^{3};\theta,\mathbf{A},\pi\right) =  \pi\left(\mathbf{d}_{0};\mathbf{A}\right) \\    \times\prod_{i<j}\prod_{t=1}^{3}\left[\frac{\exp\left(\alpha d_{ijt-1}+\beta r_{ijt-1}+A_{ij}\right)}{1+\exp\left(\alpha d_{ijt-1}+\beta r_{ijt-1}+A_{ij}\right)}\right]^{d_{ijt}}\\ 
     \times\left[\frac{1}{1+\exp\left(\alpha d_{ijt-1}+\beta r_{ijt-1}+A_{ij}\right)}\right]^{1-d_{ijt}} \\ 
     = \frac{\pi\left(\mathbf{d}_{0};\mathbf{A}\right)\times\prod_{i<j}\exp\left(\alpha\sum_{t=1}^{3}d_{ijt-1}d_{ijt}+\beta\sum_{t=1}^{3}r_{ijt-1}d_{ijt}+A_{ij}\sum_{t=1}^{3}d_{ijt}\right)}{\prod_{i<j}\prod_{\delta=0,1}\prod_{\rho=0}^{N-2}\left[1+\exp\left(\alpha\delta+\beta\rho+A_{ij}\right)\right]^{m_{ij}\left(\delta,\rho; \mathbf{d}^3\right)}},
\end{multline}
with
\[
m_{ij}\left(\delta,\rho;\mathbf{d}^{3}\right)=\left|\left\{ t=1,2,3\ :\ d_{ij,t-1}=\delta,r_{ij,t-1}=\rho\right\} \right|
\]
equal to the within-$ij$-dyad count of the number of link decisions $d_{ijt}$
preceded by different combinations of feasible values for $d_{ijt-1}\in\left\{ 0,1\right\} $
and $r_{ijt-1}\in\left\{ 0,1,\ldots,N-2\right\} $. For example, $m_{ij}\left(0,3\right)$ counts the number of times dyad $ij$ considers forming a link when, in the period prior, they were (i) unlinked and (ii) shared $3$ friends in common.

Recall that $\mathbb{D}_{N}$ denotes the set of $N\times N$ undirected binary adjacency
matrices. Let $\mathbb{D}_{N}^{3}\overset{def}{\equiv}\left(\mathbb{D}_{N},\mathbb{D}_{N},\mathbb{D}_{N},\mathbb{D}_{N}\right)$ denote the set of all possible length $T=3$ sequences of such networks, and consider the restricted set of network sequences:
\begin{align}
    \mathbb{B}^{*}\left(\mathbf{d}^{3}\right) = & \left\{ \mathbf{b}^{3}=\left(\mathbf{b}_{0},\mathbf{b}_{1},\mathbf{b}_{2},\mathbf{b}_{3}\right) \ : \ \mathbf{b}_{0}=\mathbf{d}_{0}, \right. \nonumber \\
    & \forall\ i=1,\ldots,N-1,\ j=i+1,\ldots,N,\nonumber \\ 
    & \&\ \forall\ \delta=0,1\ \&\ \rho=0,1,\ldots,N-2,\nonumber \\ 
    & \sum_{t=1}^{3}b_{ijt}=\sum_{t=1}^{3}d_{ijt},\nonumber \\
    & \&\ m_{ij}\left(\delta,\rho;\mathbf{b}^{3}\right)=m_{ij}\left(\delta,\rho;\mathbf{d}^{3}\right),\nonumber \\ 
    & \left. \&\ b_{ijt}=b_{jit}\in\left\{ 0,1\right\} \mbox{ for }t=0,1,2,3 \right\}.\label{eq: reference_set_v1}
\end{align}
The denominator of  $\Pr\left(\mathbf{D}^{3}=\mathbf{b}^{3};\theta,\mathbf{A},\pi\right)$ coincides with that appearing in \eqref{eq: likelihood_chamberlain} for all $\mathbf{b}^{3} \in \mathbb{B}^{*}\left(\mathbf{d}^{3}\right)$. The initial condition, $\pi$, as well as the contributions of the $A_{ij}$ appearing in the numerator of \eqref{eq: likelihood_chamberlain}, additionally coincide. These observations deliver the conditional probability:
\begin{align}
    \Pr\left(\left.\mathbf{d}^{3}\right|\mathbf{d}^{3}\in\mathbb{B}^{*};\theta,\mathbf{A},\pi\right) & =\frac{\prod_{i<j}\exp\left(\alpha\sum_{t=1}^{3}d_{ijt-1}d_{ijt}+\beta\sum_{t=1}^{3}r_{ijt-1}d_{ijt}\right)}{\sum_{\mathbf{b}\in\mathbb{B}^{*}\left(\mathbf{d}^{3}\right)}\prod_{i<j}\exp\left(\alpha\sum_{t=1}^{3}b_{ijt-1}b_{ijt}+\beta\sum_{t=1}^{3}s_{ijt-1}b_{ijt}\right)},\label{eq: conditional_likelihood_v1}
\end{align}
where we let $s_{ijt}\overset{def}{\equiv}\sum_{k}b_{ikt}b_{jkt}$. The right-hand-side of \eqref{eq: conditional_likelihood_v1} does not vary with $\mathbf{A}$ or $\pi\left(\mathbf{d}_{0};\mathbf{A}\right)$; we have found an implication of the model, indexed by $\theta$, that is free of the nuisance parameters.

The set $\mathbb{B}^{*}\left(\mathbf{d}^{3}\right)$ appears rather complicated. Indeed, a superficial glance does not indicate whether its cardinality is large, and hence using \eqref{eq: conditional_likelihood_v1} for inference on $\alpha$ and $\beta$ likely informative, or whether ``conditioning away" the high-dimensional nuisance parameters $\mathbf{A}$ and $\pi$ leaves effectively no identifying variation to work with. Furthermore, \eqref{eq: conditional_likelihood_v2} suggest no obvious independence structure indicating how, or even if, information might accumulate as the number of agents, $N$, in the network grows large.

Our first result, Proposition \ref{prop: conditioning_set} immediately below, provides an alternative characterization of $\mathbb{B}^*\left(\mathbf{d}^{3}\right)$; a characterization that facilitates the derivation of several results presented below. Defining (see below for commentary)
\begin{align}
\mathbb{B}\left(\mathbf{d}^{3}\right) = & \left\{ \mathbf{b}^{3}=\left(\mathbf{b}_{0},\mathbf{b}_{1},\mathbf{b}_{2},\mathbf{b}_{3}\right)\ :\ \mathbf{b}_{0}=\mathbf{d}_{0},\ \mathbf{b}_{3}=\mathbf{d}_{3},\right.\nonumber \\
 & \mathbf{b}_{1}+\mathbf{b}_{2}=\mathbf{d}_{1}+\mathbf{d}_{2},\ \forall\ i=1,\ldots,N-1,j=i+1,\nonumber \\
 & \left(s_{ij1},s_{ij2}\right)=\left(r_{ij1},r_{ij2}\right)\ \text{if}\ \left(b_{ij1},b_{ij2}\right)=\left(d_{ij1},d_{ij2}\right)\nonumber \\
 & \&\ \left(s_{ij1},s_{ij2}\right)=\left(r_{ij2},r_{ij1}\right)\ \text{if}\ \left(b_{ij1},b_{ij2}\right)=\left(d_{ij2},d_{ij1}\right)\nonumber, \\
 & \left.\&\ b_{ijt}=b_{jit}\in\left\{ 0,1\right\} \mbox{ for } t=0,1,2,3\right\} \label{eq: reference_set_v2},
\end{align}
allows us to state:
\begin{prop}
\label{prop: conditioning_set}\textsc{(Sufficient Conditioning Set)}
\\
(i) $\mathbf{d}_{0}$, $\sum_{t=1}^{3}d_{ijt}$, and $m_{ij}\left(\delta,\rho;\mathbf{d}^{3}\right)$
for $\forall\ i=1,\ldots,N-1,j=i+1,\ldots,N$ are sufficient statistics for $\mathbf{A}$ and $\pi$;\\
 (ii) $\mathbb{B}^{*}\left(\mathbf{d}^{3}\right)=\mathbb{B}\left(\mathbf{d}^{3}\right)$.
\end{prop}
\proof{See Appendix \ref{app: main-results}.}

Proposition \ref{prop: conditioning_set} gives a set of sufficient statistics for $\mathbf{A}$ and $\pi$. It also provides a convenient alternative characterization of the set of all network sequences where these sufficient statistics coincide with those found in the observed network sequence, $\mathbf{d}^3$.
 
Using Proposition \ref{prop: conditioning_set} -- see also \cite{dano2025binary} -- we can rewrite the conditional likelihood \eqref{eq: conditional_likelihood_v1} as:
\begin{multline}
\Pr\left(\left.\mathbf{d}^{3}\right|\mathbf{d}^{3}\in\mathbb{B}\left(\mathbf{d}^{3}\right);\theta,\mathbf{A},\pi\right) =\Pr\left(\left.\mathbf{d}^{3}\right|\mathbf{d}^{3}\in\mathbb{B}\left(\mathbf{d}^{3}\right);\theta\right)  \\    =\frac{\exp\left(\alpha\sum_{i<j}\sum_{t=1}^{3}d_{ijt-1}d_{ijt}+\beta\sum_{i<j}\sum_{t=1}^{3}r_{ijt-1}d_{ijt}\right)}{\sum_{\mathbf{b}\in\mathbb{B}\left(\mathbf{d}^{3}\right)}\exp\left(\alpha\sum_{i<j}\sum_{t=1}^{3}b_{ijt-1}b_{ijt}+\beta\sum_{i<j}\sum_{t=1}^{3}s_{ijt-1}b_{ijt}\right)}\label{eq: conditional_likelihood_v2}.
\end{multline}
The likelihood \eqref{eq: conditional_likelihood_v2}, by conditioning on sufficient statistics for $\mathbf{A}$ and $\pi$, generates a probability law for $\left.\mathbf{D}^{3} \ \right| \ \mathbf{D}^{3}\in\mathbb{B}\left(\mathbf{d}^{3}\right)$ indexed only by the finite-dimensional parameter $\theta$. For a given population value of $\theta=\theta_0$, equation \eqref{eq: conditional_likelihood_v2} assigns different \emph{ex ante} probabilities to network sequences in $\mathbb{B}\left(\mathbf{d}^{3}\right)$. The (conditional) maximum likelihood principle suggests choosing  $\hat\theta$ to make this \emph{ex ante} probability maximal for the realized network sequence, $\mathbf{d}^3$.

Although free of nuisance parameters, using \eqref{eq: conditional_likelihood_v2} as a basis for estimation of, and inference on, $\theta$ is not straightforward. Evaluating the denominator of \eqref{eq: conditional_likelihood_v2} requires enumerating the elements of $\mathbb{B}\left(\mathbf{d}^{3}\right)$, a set with non-trivial geometry and potentially intractably large cardinality. 

Although exact evaluation of \eqref{eq: conditional_likelihood_v2} appears challenging, it may be possible to construct a simulation-based estimate of it. Doing so requires a method of sampling from $\mathbb{B}\left(\mathbf{d}^{3}\right)$ uniformly at random. The observed network sequence $\mathbf{D}^3 = \mathbf{d}^3$ is an element of $\mathbb{B}\left(\mathbf{d}^{3}\right)$ by construction. To find a new element in $\mathbb{B}\left(\mathbf{d}^{3}\right)$ we can randomly rewire $\mathbf{d}^3$.

We call dyads where $\left(d_{ij1},d_{ij2}\right)=\left(0,1\right)$ or $\left(d_{ij1},d_{ij2}\right)=\left(1,0\right)$ \emph{movers}, since they change their link status between periods $1$ and $2$. We say mover $ij$'s edges are ``swapped" in $\mathbf{b}^3$ relative to $\mathbf{d}^3$ if $\left(b_{ij1},b_{ij2}\right)=\left(d_{ij2},d_{ij1}\right)$. Since the initial and terminal conditions of $\mathbf{b}^3$ are constant across all $\mathbf{b}^3 \in \mathbb{B}\left(\mathbf{d}^{3}\right)$ and, furthermore, $b_{ij1}+b_{ij2}=d_{ij1}+d_{ij2}$ for all dyads $ij$, a move from $\mathbf{d}^3$ to \emph{any} $\mathbf{b}^3 \in \mathbb{B}\left(\mathbf{d}^{3}\right)$ can be engineered by swapping edges for a subset of \emph{mover} dyads.

Implementing swaps of this type, however, may induce violations of the constraint that $\left(s_{ij1},s_{ij2}\right)=\left(r_{ij1},r_{ij2}\right)$ if $\left(b_{ij1},b_{ij2}\right)=\left(d_{ij1},d_{ij2}\right)$ and
$\left(s_{ij1},s_{ij2}\right)=\left(r_{ij2},r_{ij1}\right)$ if $\left(b_{ij1},b_{ij2}\right)=\left(d_{ij2},d_{ij1}\right)$. To understand why, observe that $d_{ijt}$ is a component of $r_{ikt}$ for $k \neq i,j$ and similarly for $r_{jkt}$ (for $k \neq i,j$). Hence, if $b_{ij1}$ and $b_{ij2}$ are a swap of the observed $d_{ij1}$ and $d_{ij2}$, then up to $2 \times2(N-2)$ of the ``friends-in-common" regressor values in $\mathbf{b}^3$ may differ from their corresponding values in $\mathbf{d}^3$. The set $\mathbb{B}\left(\mathbf{d}^{3}\right)$ allows for such differences, but only of a very particular form. Specifically, it requires that if $b_{ij1}$ and $b_{ij2}$ are a swap of the observed $d_{ij1}$ and $d_{ij2}$, then it must be the case that $s_{ij1}$ and $s_{ij2}$ are also a swap of the observed $r_{ij1}$ and $r_{ij2}$.

Consequently, successfully moving from one element in $\mathbb{B}\left(\mathbf{d}^{3}\right)$ to another may require the simultaneous swapping of \emph{many} movers' edges to ensure that any constraint violations induced by one group of edge swaps is appropriately offset by another.

A Markov Chain Monte Carlo (MCMC) algorithm for constrained graph simulation introduced by \cite{Pelican_Graham_ReStud2026} provides a template for how one might approach this problem. We do not pursue this approach. Instead, we will characterize a subset of $\mathbb{B}\left(\mathbf{d}^{3}\right)$, that we \emph{can} explicitly enumerate and develop corresponding conditional estimators. Our approach, since it conditions on only a subset of those network sequences with common sufficient statistics for $\mathbf{A}$ and $\pi\left(\mathbf{d}_{0};\mathbf{A}\right)$, involves some loss of information. Its virtues are analytical and computational tractability.

\subsection*{Star systems}

In this subsection we define a \emph{subset} of $\mathbb{B}\left(\mathbf{d}^{3}\right)$ whose elements can be enumerated by simultaneously swapping all period $1$ and period $2$ edges among a set of dyads which we will call \emph{identifying star systems}.  An identifying star system is a set of dyads such that, if we swap all its period $1$ edges with their period $2$ counterparts, then we move to a new network sequence that is \emph{guaranteed} to remain in $\mathbb{B}\left(\mathbf{d}^{3}\right)$. Identifying star systems provide a simple way to move through a (subset) of the elements of $\mathbb{B}\left(\mathbf{d}^{3}\right)$. Star systems vary in the number of dyads constituting them. We show how to utilize all those systems composed of $p=2 \ldots K$ agents.

Let $\mathbb{B}_{K}\left(\mathbf{d}^{3}\right) \subset \mathbb{B}\left(\mathbf{d}^{3}\right)$ denote the set of all network sequences reachable by swapping the edges of identifying star systems of size $K$ or smaller (defined formally below). If the observed network sequence, $\mathbf{d}^3$, contains $\mathbf{m}_{2:K,N}$ identifying star systems of size $K$ or smaller, then $\left|\mathbb{B}_{K}\left(\mathbf{d}^{3}\right)\right|=2^{\mathbf{m}_{2:K,N}}$, corresponding to all the possible ways in which the $\mathbf{m}_{2:K,N}$ identifying K-star systems may have their edges swapped in $\mathbf{b}^3 \in \mathbb{B}_{K}\left(\mathbf{d}^{3}\right)$ relative to their configuration in $\mathbf{d}^3$ or not. We say a star system's edges are ``swapped" in $\mathbf{b}^3$ if they are permuted relative to their configuration in the observed network sequence $\mathbf{d}^3$.

In what follows we explore estimators which condition on the event $\mathbf{D}^3 \in \mathbb{B}_{K}\left(\mathbf{d}^{3}\right)$. Considering the conditional probability of observing the network sequence in hand, given that it belongs to a set of network sequences that can be enumerated by swapping the edges of identifying star systems, provides an analytically tractable conditioning event and an associated criterion function amenable to asymptotic analysis.

To describe our estimator and characterize its sampling properties we need to first carefully describe $\mathbb{B}_{K}\left(\mathbf{d}^{3}\right)$.

\begin{defn}
\label{def: k-star-system}(\textsc{p-Star System}) The dyad set $\mathcal{S}_{p}=\left\{ i_{1}i_{2},i_{1}i_{3},\ldots,,i_{1}i_{p}\right\} $
corresponds to a $p$-\emph{star system} with central vertex $i_{1}$. Let $v\left(\mathcal{S}_{p}\right)=\left\{ i_{1},\ldots,i_{p}\right\}$ denote the agents belonging to $\mathcal{S}_{p}.$ 
\end{defn}

\begin{defn}
\label{def: neighborhood}(\textsc{Neighborhood}) The period $t$ \emph{neighborhood} of $\mathcal{S}_{p}$ includes all agents not in $v\left(\mathcal{S}_{p}\right)$ but connected to an agent in this set during period $t$:
\begin{equation}
n_{t}\left(\mathcal{S}_{p}\right)=\left\{j\ :\ d_{jit}=1\ \text{for a least one }  i \in v\left(\mathcal{S}_{p}\right) \right\} \setminus\left\{ i_{1},\ldots,i_{p}\right\} .\label{eq: neighborhood}
\end{equation}
\end{defn}
\begin{defn}
\label{def: stable-neighborhood}(\textsc{Stable Neighborhood}) The
neighborhood of $\mathcal{S}_{p}=\left\{ i_{1}i_{2},i_{1}i_{3},\ldots,,i_{1}i_{p}\right\}$ is \textit{stable} if\\
(i) $\forall i\in v\left(\mathcal{S}_{p}\right)$ and $j \in \left\{1,\ldots,N\right\} \setminus v\left(\mathcal{S}_{p}\right)$: $d_{ij1}=d_{ij2}$;\\
(ii) $\forall j\in n_{1}\left(\mathcal{S}_{p}\right)$ and $l \in \left\{1,\ldots,N\right\}$: $d_{jl1}=d_{jl2}$; \\
(iii) $\forall i_{k}i_{l}$ with $2\leq k,l\leq p$ and $k\neq l$
$d_{i_{k}i_{l}1}=d_{i_{k}i_{l}2}$. \\
\end{defn}
A p-star system is embedded in a stable neighborhood if all agents not in $v\left(\mathcal{S}_{p}\right)$, but connected to someone in $v\left(\mathcal{S}_{p}\right)$, maintain a fixed link structure in periods $1$ and $2$. This includes edges into $v\left(\mathcal{S}_{p}\right)$, as well as edges to other agents in the network (Parts (i) and (ii) of Definition \ref{def: stable-neighborhood}). We also require that any edges among agents in $v\left(\mathcal{S}_{p}\right)$ that are not among the dyads in $\mathcal{S}_{p}$ remain fixed across periods $1$ and $2$ (Part (iii) of Definition \ref{def: stable-neighborhood}). The joint impact of these three conditions is that, other than any edges among the dyads in $\mathcal{S}_{p}$, the network structure up to two degrees away from any agent in $v\left(\mathcal{S}_{p}\right)$ is stable across periods $1$ and $2$.

While we require \emph{stability} of those links surrounding $\mathcal{S}_{p}$, we require \emph{movement} of the links within it.

\begin{defn}
\label{def: identifying-star-system}(\textsc{Identifying p-Star System}) $\mathcal{S}_{p}=\left\{ i_{1}i_{2},i_{1}i_{3},\ldots,,i_{1}i_{p}\right\}$ is
\foreignlanguage{estonian}{\textit{identifying}} if \\
(i) $\mathcal{S}_{p}$ is embedded in a stable neighborhood,\\
(ii) $\forall i_{1}i_{k}$ with $2\leq k\leq p$ $d_{i_{1}i_{k}1}\neq d_{i_{1}i_{k}2}$.
\end{defn}

Observe that if we remove one agent from the identifying $p$-star system $\mathcal{S}_{p}$, the remaining $p-1$ agents do not constitute an identifying $(p-1)$-star system (condition (i) of Definition \ref{def: stable-neighborhood} is violated). Conversely, if $\mathcal{S}_{p}$ is \emph{not} identifying, there may be subsets of agents in $v\left(\mathcal{S}_{p}\right)$ which \emph{do} form an identifying system. 

Lemma \ref{lem: separation-nodes-lemma} and a corollary of it, as well as Lemma \ref{lem: invariance-lemma}, which immediately follows, play important roles in helping us analyze the conditional likelihood  $\Pr\left(\left.\mathbf{D}^{3}=\mathbf{d}^{3}\right|\mathbf{D}^{3}\in\mathbb{B}_{K}\left(\mathbf{d}^{3}\right);\theta\right)$.

\begin{lem}(\textsc{Agent Separation})   \label{lem: separation-nodes-lemma}
    Any node $i\in \mathcal{V}$ belongs to at most one identifying star system.
\end{lem}
\begin{proof}
    Towards a contradiction, suppose that $i \in v(\mathcal{S}_{p}) \cap v(\mathcal{S}_{l})$ where $\mathcal{S}_{p}$ and $\mathcal{S}_{l}$ are two distinct identifying star systems. \\
    \textbf{Case 1:} If $i$ is the central node of $\mathcal{S}_{p}$ and $\mathcal{S}_{l}$. Since $\mathcal{S}_{p}\neq \mathcal{S}_{l}$, either $\exists k\in v(\mathcal{S}_{p}) \setminus v(\mathcal{S}_{l})$ or $\exists k\in v(\mathcal{S}_{l}) \setminus v(\mathcal{S}_{p})$  and part (ii) of Definition  \ref{def: identifying-star-system} implies $d_{ik1}\neq d_{ik2}$. This change in links status of dyad $ik$ violates the assumption that  $\mathcal{S}_{l}$ or  $\mathcal{S}_{p}$  is embedded in a stable neighborhood. \\
    \textbf{Case 2:} If $i$ is the central node of $\mathcal{S}_{p}$ but not $\mathcal{S}_{l}$. Let $j$ denote the center of $\mathcal{S}_{l}$. If $j\notin v(\mathcal{S}_{p})$, then part (i) of Definition \ref{def: stable-neighborhood} implies that $d_{ij1}=d_{ij2}$. However, since $j$ is the center node of $\mathcal{S}_{l}$, this contradicts part (ii) of Definition  \ref{def: identifying-star-system} which requires $d_{ij1}\neq d_{ij2}$. Hence, $j\in v(\mathcal{S}_{p})$ which means that $j$ must be a leaf node of $\mathcal{S}_{p}$. There are two cases. Either $\mathcal{S}_{p}=\mathcal{S}_{l}=\{ij\}$ contradicting the premise that $\mathcal{S}_{p}$ and $\mathcal{S}_{l}$ are two distinct identifying star systems. Or, without loss of generality, there exists $k\in v(\mathcal{S}_{p})\setminus \{i,j\}$ and $d_{ik1}\neq d_{ik2}$ by part (ii) of Definition  \ref{def: identifying-star-system}. If $k \notin v(\mathcal{S}_{l})$, this contradicts part (i) of Definition \ref{def: stable-neighborhood}. If $k \in v(\mathcal{S}_{l})$, then by part (ii) of Definition  \ref{def: identifying-star-system}, $d_{jk1}\neq d_{jk2}$. However, $j$ and $k$ are both leaves of $\mathcal{S}_{p}$, so part (iii) of Definition \ref{def: stable-neighborhood} requires $d_{jk1}=d_{jk2}$, a contradiction. \\
    \textbf{Case 3}: If $i$ is the central node of $\mathcal{S}_{l}$ but not $\mathcal{S}_{p}$, we are back to Case 2 after swapping the roles of  $\mathcal{S}_{l}$ and $\mathcal{S}_{p}$.\\
    \textbf{Case 4:} If $i$ is not the central node of either $\mathcal{S}_{p}$ or $\mathcal{S}_{l}$. Let $j$ denote the central node of $\mathcal{S}_{p}$. Then, we know by part (ii) of Definition \ref{def: identifying-star-system} that $d_{ji1}\neq d_{ji2}$. Hence, if $j$ is not the central node of $\mathcal{S}_{l}$, we again violate the neighborhood stability of $\mathcal{S}_{l}$. Finally if $j$ is also the central node of $\mathcal{S}_{l}$, we are back to Case 1 with $j$ substituting $i$.
\end{proof}

\begin{cor}(\textsc{Dyad Separation})\label{cor: separation-dyads-corro}
    Any dyad $ij$ belongs to at most one identifying star system, hence there can be at most $\frac{N}{2}$ identifying star systems.
\end{cor}

Corollary \ref{cor: separation-dyads-corro}, as we shall see, implies limits on the rate-of-convergence of our preferred estimate of $\theta$. 

Lemma \ref{lem: invariance-lemma} demonstrates that direct connections between agents in two identifying star systems are impossible.

\begin{lem}
\label{lem: invariance-lemma}\textsc{(Invariance of links across identifying star systems)} Let $\mathcal{S}_{p}$ and $\mathcal{S}_{l}$ denote two distinct identifying star systems. Then, for all nodes $i\in v(\mathcal{S}_{p})$ and $j\in v(\mathcal{S}_{l})$, we have
\begin{itemize}
    \item[(A)] $D_{ij1}=D_{ij2}=0$
    \item[(B)] $R_{ij1}=R_{ij2}$.
\end{itemize}
\end{lem}

\begin{proof}
    We begin with claim (A). Towards a contradiction, assume $D_{ij1}=1$ or $D_{ij2}=1$. In either case, it follows from parts \textcolor{black}{(i)} and \textcolor{black}{(ii)} of Definition \ref{def: stable-neighborhood} that $j\in n_{1}\left(\mathcal{S}_{p}\right)$ and $D_{jl1}=D_{jl2}$ for all $l=1,\ldots,N$. This violates Definition \ref{def: identifying-star-system}.(ii). \\
    Next, for claim (B) we can write the decomposition:
\begin{align*}
R_{ij2}-R_{ij1} =&  \sum_{k=1}^{N}\left(D_{ik2}D_{jk2}-D_{ik1}D_{jk1}\right) \\
=&  \sum_{k\in v\left(\mathcal{S}_{p}\right)}\left(D_{ik2}D_{jk2}-D_{ik1}D_{jk1}\right)+\sum_{k \in v\left(\mathcal{S}_{l}\right)}\left(D_{ik2}D_{jk2}-D_{ik1}D_{jk1}\right) \\
&+\sum_{k \notin v\left(\mathcal{S}_{p}\right)\cup v\left(\mathcal{S}_{l}\right)}\left(D_{ik2}D_{jk2}-D_{ik1}D_{jk1}\right)  \\
=& \sum_{k\in v\left(\mathcal{S}_{p}\right)}\left(D_{ik2}D_{jk2}-D_{ik1}D_{jk1}\right)+\sum_{k \in v\left(\mathcal{S}_{l}\right)}\left(D_{ik2}D_{jk2}-D_{ik1}D_{jk1}\right)+0 
\end{align*}
where the second equality follows from Lemma \ref{lem: separation-nodes-lemma} and the third equality follows from the neighborhood stability of $\mathcal{S}_{p}$ and $\mathcal{S}_{l}$. By claim (A), $\forall k\in v\left(\mathcal{S}_{p}\right)$, $D_{jk1}=D_{jk2}=0$ and $\forall k\in v\left(\mathcal{S}_{l}\right)$, $D_{ik1}=D_{ik2}=0$. Consequently, $R_{ij1}=R_{ij2}$ as claimed.
\end{proof}

Lemma \ref{lem: invariance-lemma} has important implications for the type of observed network sequences that are likely to contain information about $\theta$. Such networks will involve at least some separated clusters of links. They will not be ``hairball" graphs, although they may include giant components.

\subsection*{Permutation Lemmas}
As noted above, setting $\left(b_{ij1},b_{ij2}\right)=\left(d_{ij2},d_{ij1}\right)$ need not result in a network sequence in $\mathbb{B}\left(\mathbf{d}^3\right)$: such swaps may induce violations of the constraint that $\left(s_{ij1},s_{ij2}\right)=\left(r_{ij1},r_{ij2}\right)$, if $\left(b_{ij1},b_{ij2}\right)=\left(d_{ij1},d_{ij2}\right)$, and
$\left(s_{ij1},s_{ij2}\right)=\left(r_{ij2},r_{ij1}\right)$, if $\left(b_{ij1},b_{ij2}\right)=\left(d_{ij2},d_{ij1}\right)$.

In this section we show that swapping the period 1 and 2 edges of an identifying star system \emph{does} result in a new network sequence which is \emph{guaranteed to} lie in $ \mathbb{B}\left(\mathbf{d}^{3}\right)$. We define the set $\mathbb{B}_{K}\left(\mathbf{d}^{3}\right) \subset \mathbb{B}\left(\mathbf{d}^{3}\right)$ to be the set of all network sequence obtainable by ``swapping the edges" (or not) of identifying star systems in $\mathbf{d}^3$ of size $K$ or less.

Lemma \ref{lem: permutation-lemma}, stated and proved below, is the key tool used to establish this result. To state and prove Lemma \ref{lem: permutation-lemma} requires some additional notation and definitions. In what follows we use $\land$ to denote ``and'', $\lor$ for ``or'', and $\veebar$ for the ``exclusive or''. We also use $\exists!$ to denote ``there exists exactly one".

Let $Z_{ij,p}$ denote an indicator function equal to one if dyad $ij$ belongs to an identifying star system of size $p\geq 2$ and define the following dyad \emph{and} vertex/agent sets.

(i) Dyads contained in an identifying p-star system: $\mathcal{D}_{p}=\{ij\ | \ Z_{ij,p}=1\}$. \\
(ii) Agents belonging to an identifying p-star system: $\mathcal{V}_{p}=\{i\in \mathcal{V}\ | \ \exists j\in \mathcal{V}: Z_{ij,p}=1\}$.  \\
(iii) Dyads composed of agents in a common p-star system, but where the dyad itself is not part of the system:
\[
\mathcal{D}_{p}^{fd}=\left\{ ij\ | \ ij\notin\mathcal{D}_{p}\land i\in \mathcal{V}_{p} \land j\in \mathcal{V}_{p} \land \left(\exists! m\in \mathcal{V}_{p}: (im,jm) \in \mathcal{D}_{p}^2 \right)\right\}.
\]
\textcolor{black}{Observe that for $p=2$, $\mathcal{D}_{p}^{fd}=\emptyset$.}

(iv) Agents belonging to some identifying star system of maximal size $K$: 
$\mathcal{V}_{2:K}=\bigcup\limits_{p=2}^{K} \mathcal{V}_{p}$. 

(v) Agents \emph{not} belonging to an identifying star system of maximal size $K$: $\mathcal{V}^{o}=\mathcal{V}\setminus \mathcal{V}_{2:K}$. 

(vi) Dyads composed of agents belonging to two distinct identifying star-systems: 
\begin{align*}
    \mathcal{D}_{k,l}^{pd}&=\left\{ ij \ | \ \left(i\in \mathcal{V}_{k} \land j\in \mathcal{V}_{l}\right)\veebar \left(i\in \mathcal{V}_{l} \land j\in \mathcal{V}_{k}\right)\right\}, \quad \text{ if } l<k \\
    \mathcal{D}_{k,k}^{pd}&=\left\{ ij \ | \ ij\notin\mathcal{D}_{k}\land i\in \mathcal{V}_{k} \land j\in \mathcal{V}_{k} \land \left(\not\exists \ m \in \mathcal{V}_{k}: (im,jm) \in \mathcal{D}_{k}^2 \right)\right\}.
\end{align*}
(vii) Dyads with one agent in $\mathcal{V}_{p}$ and
the other $\mathcal{V}^{o}$: 
\[
\mathcal{D}_{p,o}^{pd}=\left\{ ij \ | \ \left(i\in \mathcal{V}_{p} \land j\in \mathcal{V}^{o}\right)\veebar \left(i\in \mathcal{V}^{o} \land j\in \mathcal{V}_{p}\right)\right\}.
\]
(viii) Dyads with neither agent belonging to an identifying star system of size less than $K$
\[
\mathcal{D}_{K}^{od}=\left\{ ij \ | \ i \in \mathcal{V}^{o} \land j\in \mathcal{V}^{o}\right\}.
\]

An implication of Lemma \ref{lem: separation-nodes-lemma} and Corollary \ref{cor: separation-dyads-corro},  specifically the fact that an agent and a dyad can belong to at most one identifying star system, means we can partition the set of all $\binom{N}{2}$ possible dyads $\mathcal{D}$ as follows:
\begin{align} \label{eq: dyad_partitioning}
    \mathcal{D} = \left(\bigcup\limits_{p=2}^{K} \mathcal{D}_p \cup \mathcal{D}_{p}^{fd} \cup \mathcal{D}_{p,o}^{pd}\right)\cup \left(\bigcup\limits_{p=2}^{K} \bigcup\limits_{l=2}^{K} \mathcal{D}_{p,l}^{pd} \right) \cup \mathcal{D}_{K}^{od}. 
\end{align}
Partition \eqref{eq: dyad_partitioning}, by grouping dyads into analytically convenient subsets, simplifies likelihood analysis.

Finally, we formally define set $\mathbb{B}_{K}\left(\mathbf{d}^{3}\right)$ to consist of all network sequences obtained via permutations (i.e., ``swapping the edges") of the period $t=1$ and $t=2$ link decisions of identifying star systems. Permutations are applied simultaneously to all dyads in a given $\mathcal{S}_K$. All other link decisions in any $\mathbf{b}^{3} \in \mathbb{B}_{K}\!\left(\mathbf{d}^{3}\right)$ are held fixed at their observed values.

With the above definitions established we can state and prove the following Lemma.

\begin{lem}
\label{lem: permutation-lemma}\textsc{(Permutations)} Define $S_{ijt}=\sum_{k}B_{ikt}B_{jkt}$ for $\mathbf{B}^{3}\in\mathbb{B}_{K}\left(\mathbf{D}^{3}\right).$
Let $\mathcal{S}_{p}$ be an identifying $p$-star system in $\mathbf{D}^{3}$
with center $i$. If $\left(B_{ik1},B_{ik2}\right)=\left(D_{ik1},D_{ik2}\right)$
for all $k \neq i$ with $k \in v\left(\mathcal{S}_{p}\right)$, then for each
\begin{itemize}
  \item[(A)] $ij\in\mathcal{S}_{p}$, $\left(S_{ij1},S_{ij2}\right)=\left(R_{ij1},R_{ij2}\right)$;
  \item[(B)] $(j,k) \in v(\mathcal{S}_{p})\times v(\mathcal{S}_{p})$ but with $jk\in \mathcal{D}_{p}^{fd}$, $\left(S_{jk1},S_{jk2}\right)=\left(R_{jk1},R_{jk2}\right);$
  \item[(C)]$j\in v\left(\mathcal{S}_{p}\right)$ and $l\notin v\left(\mathcal{S}_{p}\right)$, $\left(S_{jl1},S_{jl2}\right)=\left(R_{jl1},R_{jl2}\right)$.
\end{itemize}
If, instead, $\left(B_{ik1},B_{ik2}\right)=\left(D_{ik2},D_{ik1}\right)$
for all $k \neq i$ with $k\in v\left(\mathcal{S}_{p}\right)$, then (A), (B) and (C) continue to hold after
replacing $\left(R_{\cdot\cdot1},R_{\cdot\cdot2}\right)$ with $\left(R_{\cdot\cdot2},R_{\cdot\cdot1}\right)$.
\end{lem}

Lemma \ref{lem: permutation-lemma} states that if, relative to their configuration in $\mathbf{d}^{3}$, the periods $1$ and $2$ edges of all dyads constituting an identifying p-star system, $\mathcal{S}_{p}$, are swapped, then the number of ``friends-in-common" terms are also swapped. Furthermore, this swapping of ``friends-in-common" extends to all dyads where at least one vertex is part of $v\left(\mathcal{S}_{p}\right)$. A consequence of this observation is that the likelihoods for the events $\mathbf{D}^{3}=\mathbf{b}^{3}$ for any $\mathbf{b}^{3} \in \mathbb{B}_{K}\left(\mathbf{d}^{3}\right)$ will, after permutation, include many multiplicands-in-common. This, as we shall see, simplifies our conditional likelihood analysis.


The proof of Lemma \ref{lem: permutation-lemma} is straightforward and is provided in Appendix \ref{app: preliminary-results-main-text}. The idea is to repeatedly decompose the common-friends term into contributions from agents within the identifying $p$-star system $\mathcal{S}_{p}$ and from agents outside it. The implications of Definition \ref{def: stable-neighborhood} are then repeatedly applied to establish the stated results.

\subsection*{Conditional star system likelihood}

We are now ready to define the criterion function on which our identification and estimation results rely. We begin with a formal defintion of $\mathbb{B}_{K}\!\left(\mathbf{d}^{3}\right)$:
\begin{align}
\mathbb{B}_{K}\!\left(\mathbf{d}^{3}\right)
= & \Bigl\{ \left(\mathbf{b}_{0},\mathbf{b}_{1},\mathbf{b}_{2},\mathbf{b}_{3}\right)\in\mathbb{D}_{N}^{3} : 
\ \mathbf{b}_{0}=\mathbf{d}_{0},\ \mathbf{b}_{3}=\mathbf{d}_{3}, \label{eq:conditioning_set} \\
& \quad \text{for each identifying $p$-star system $\mathcal{S}_{p}$ in $\mathbf{d}^{3}$, } p=1,\ldots,K, \notag \\
& \quad  b_{ij1}=d_{ij1}\land b_{ij2}=d_{ij2} \ \forall\, ij\in\mathcal{S}_{p} \notag \\
& \quad \text{or }  b_{ij1}=d_{ij2}\land b_{ij2}=d_{ij1} \ \forall\, ij\in\mathcal{S}_{p}, \notag \\
& \quad \text{otherwise } b_{ij1}=d_{ij1}\land b_{ij2}=d_{ij2}\ \text{ for all other dyads } ij \Bigr\}. \nonumber
\end{align}
Next consider the conditional likelihood of observing $\mathbf{D}^{3}=\mathbf{d}^{3}$
given that $\mathbf{D}^{3}\in\mathbb{B}_{K}\left(\mathbf{d}^{3}\right)$:
\begin{equation}
\ell^{css}\left(\mathbf{d}^{3};\theta,\mathbf{A},\pi\right)=\frac{\Pr\left(\mathbf{D}^{3}=\mathbf{d}^{3};\theta,\mathbf{A},\pi\right)}{\sum_{\mathbf{b}^{3}\in\mathbb{B}_{K}\left(\mathbf{d}^{3}\right)}\Pr\left(\mathbf{D}^{3}=\mathbf{b}^{3};\theta,\mathbf{A},\pi\right)}.\label{eq: conditional_likelihood_ss_v1}
\end{equation}
The `css' superscript denotes ``conditional star-system'' (likelihood). Like, \eqref{eq: conditional_likelihood_v2} introduced earlier, \eqref{eq: conditional_likelihood_ss_v1} is invariant to the population values of $\mathbf{A}$ and $\pi$ (see below). Hence we will often just write $\ell^{css}\left(\mathbf{d}^{3};\theta\right)$ = $\ell^{css}\left(\mathbf{d}^{3};\theta,\mathbf{A},\pi\right)$ in order to economize on notation.

Define:
\begin{align}
    b(q_{ij};\theta)\overset{def}{\equiv}&\exp\Bigl(-\alpha \left(d_{ij2}-d_{ij1}\right)\left(d_{ij3}-d_{ij0}\right)\notag \\
    &-\beta \Bigl(\left[d_{ij3}-d_{ij1}\right](r_{ij2}-r_{ij0})-\left[d_{ij3}-d_{ij2}\right](r_{ij1}-r_{ij0})\Bigr)\Bigr),\label{eq: b_function_def}
\end{align}
with $q_{ij}\overset{def}{\equiv}\left(d_{ij0},d_{ij1},d_{ij2},d_{ij3},r_{ij0},r_{ij1},r_{ij2}\right)'.$

Let $\mathcal{G}_{p}=\left\{ \mathcal{S}_{p\mathbf{1}},\ldots,\mathcal{S}_{p\mathbf{m}_{p,N}}\right\} $
denote the set of $\mathbf{m}_{p,N}$ identifying $p$-star systems
contained in $\mathbf{D}^{3}=\mathbf{d}^{3}$ for $p=2,\ldots,K$. We use the bold subscripts $\mathbf{i}=\mathbf{1},\ldots,\mathbf{m}_{p,N}$ to index the identifying star systems in $\mathcal{G}_{p}$ in order to emphasize that they're composed of \emph{multiple} ($p\geq2$) agents.

Let $\mathcal{S}_{p}$
denote a generic identifying system in $\mathcal{G}_{p}$. Next let $\mathbf{m}_{2:K,N} \overset{def}{\equiv} \sum_{p=2}^K \mathbf{m}_{p,N}$ denote the \emph{total} number of identifying star systems of size $K$ or less. Finally, for a $p$-star system $S_{p\mathbf{i}}\in \mathcal{G}_{p}$, define the dyad set $\bar{\mathcal{S}}_{p\mathbf{i}}=\mathcal{S}_{p\mathbf{i}}\cup \mathcal{S}_{p\mathbf{i}}^{fd} \cup \mathcal{S}_{p\mathbf{i},o}^{pd}$ with $\mathcal{S}_{p\mathbf{i}}^{fd}=\{ij \in \mathcal{D}_{p}^{fd}|  (i,j)\in v(\mathcal{S}_{p\mathbf{i}})\times v(\mathcal{S}_{p\mathbf{i}})\}$ and $\mathcal{S}_{p\mathbf{i},o}^{pd} = \{ij \in \mathcal{D}_{p,o}^{pd}|  (i,j)\in v(\mathcal{S}_{p\mathbf{i}})\times \mathcal{V}^{o}\}$. 

Our main result about $\ell^{css}\left(\mathbf{d}^{3};\theta,\mathbf{A},\pi\right)$ is stated in the following theorem.
\begin{thm}
\label{thm: condition-star-system-likelihood}\textsc{(Conditional Star System Likelihood)} Under the dynamic network formation process defined by equations \eqref{eq:ij_Link_Model}, \eqref{eq:U_ijt_Independence}, \eqref{eq:Initial_Condition} and logistic random utility shocks, the conditional likelihood of the event $\mathbf{D}^3=\mathbf{d}^3$ given $\mathbf{D}^{3}\in\mathbb{B}_{K}\left(\mathbf{d}^{3}\right)$, equation \eqref{eq: conditional_likelihood_ss_v1} above, exhibits the equivalent representation: 
\begin{align}
\ell^{css}\left(\mathbf{d}^{3};\theta,\mathbf{A},\pi\right)&=\prod_{p=2}^K \prod_{\mathbf{i}=\mathbf{1}}^{\mathbf{m}_{p,N}} \left[1+\prod_{ij\in \bar{\mathcal{S}}_{p\mathbf{i}}} b(q_{ij};\theta)\right]^{-1} \notag \\
&=\prod_{p=2}^K \prod_{\mathbf{i}=\mathbf{1}}^{\mathbf{m}_{p,N}} \frac{\exp(t_{p\mathbf{i}}'\theta)}{1+\exp(t_{p\mathbf{i}}'\theta)} \label{eq: conditional_likelihood_ss_v2}
\end{align}
where $t_{p\mathbf{i}}=\left(t_{1,p\mathbf{i}}, t_{2,p\mathbf{i}}\right)'$, with
\begin{align*}
t_{1,p\mathbf{i}}&=   \sum_{ij\in \bar{\mathcal{S}}_{p\mathbf{i}}} \left(d_{ij2}-d_{ij1}\right)\left(d_{ij3}-d_{ij0}\right) \\
t_{2,p\mathbf{i}}&= \sum_{ij\in \bar{\mathcal{S}}_{p\mathbf{i}}} \Bigl(\left[d_{ij3}-d_{ij1}\right](r_{ij2}-r_{ij0})-\left[d_{ij3}-d_{ij2}\right](r_{ij1}-r_{ij0})\Bigr).
\end{align*}
\end{thm}
\proof{See Appendix \ref{app: main-results}.}

As in analogous results for single-agent dynamic binary choice models (e.g., \citealp*{Cox_JRSS58,Chamberlain_LALMD85,Honore_Kyriazidou_EM00}),
Theorem \ref{thm: condition-star-system-likelihood} codifies a model implication invariant to the value of $\mathbf{A}$; allowing for fixed effects identification of $\theta_0=\left(\alpha_0,\beta_0\right)'$. The result has a number of practical and theoretical implications.

First, direct evaluation of \eqref{eq: conditional_likelihood_ss_v1} requires enumeration of the set $\mathbb{B}_{K}\left(\mathbf{d}^{3}\right)$ (with cardinality $2^{\sum_{p=2}^K \mathbf{m}_{p,N}}$); the so called ``intractable denominator problem". Equation \eqref{eq: conditional_likelihood_ss_v2} shows that we can instead evaluate the log-likelihood by computing a sum of just $\sum_{p=2}^K \mathbf{m}_{p,N}$ elements. This is a substantial simplification.

Second, the product structure of \eqref{eq: conditional_likelihood_ss_v2} hints at the \emph{possibility} of some form of (conditional) independence across its multiplicands. Of course, such independence is not inevitable,\footnote{Indeed composite likelihoods with dependence across multiplicands, feature widely in the econometrics literature on networks; see, for example, \cite{Graham_HBE2020}.} but we show in Section \ref{sec: inference} that the multiplicands in \eqref{eq: conditional_likelihood_ss_v2} are indeed conditionally independent of one another. We utilize this structure to derive single network, many agent, sampling distribution approximations.

Third, observe that three types of dyads make nontrivial contributions to \eqref{eq: conditional_likelihood_ss_v2}: (i) those dyads which directly make up identifying star systems, $\bigcup\limits_{p=2}^{K} \mathcal{D}_p$; (ii) those dyads where both agents belong to a star system, but the dyad itself does not, $\bigcup\limits_{p=2}^{K} \mathcal{D}_{p}^{fd}$; and, finally, (iii) all those dyads with one agent belonging to a star system and one agent not, $\bigcup\limits_{p=2}^{K} \mathcal{D}_{p,o}^{pd}$. This highlights how the entire architecture of the neighborhood surrounding a star-system contributes to identifying $\theta$.

A proof of Theorem \ref{thm: condition-star-system-likelihood} is available in Appendix \ref{app: main-results}. It involves the application of Lemmas \ref{lem: separation-nodes-lemma} to \ref{lem: permutation-lemma} to navigate from \eqref{eq: conditional_likelihood_ss_v1} to \eqref{eq: conditional_likelihood_ss_v2}. The proof provides insight into the nature of both our problem and proposed solution. It also highlights connections between our analysis and prior work on dynamic discrete choice panel data. For example, we use several ``tricks" introduced by \cite{Honore_Kyriazidou_EM00} in our proof.

In the single-agent setting, \cite{Cox_JRSS58} compares the \emph{relative} frequencies of the link sequences $d_{ij0}01d_{ij3}$ and $d_{ij0}10d_{ij3}$ to infer the strength of state-dependence in the presence of unobserved heterogeneity.\footnote{The notation ``$d_{ij0}01d_{ij3}$" is for the link sequence $d_{ij0}$, $d_{ij1}=0$,  $d_{ij2}=1$ and $d_{ij3}$.} \emph{Movers} -- agents who revise their choice between periods $t=1$ and $t=2$ -- contribute to identifying state-dependence.

In the multi-agent network setting, the\emph{ relative }frequencies of the link sequences $d_{ij0}01d_{ij3}$ and $d_{ij0}10d_{ij3}$ similarly provides
information about the signs and magnitudes of $\alpha_{0}$ and $\beta_{0}$.
However, we must confine analysis to sets of dyads -- \emph{star systems} in our nomenclature -- that, in addition to changing their link status across periods $t=1$ and $t=2$, are also embedded in stable neighborhoods (see Definition \ref{def: stable-neighborhood}). We separately learn about $\alpha_{0}$ versus $\beta_{0}$ from star systems embedded in
neighborhoods with varying types of local network architecture.
The need to condition on neighborhood stability arises because of
interdependencies in link decisions across dyads. 

Consider two network sequences,\emph{ identical every respect}, except
that in the first one dyad $ij$'s link history is $d_{ij0}01d_{ij3}$,
while in the second it is $d_{ij0}10d_{ij3}$. The second network
sequence can be derived by permuting the period $t=1$ and $t=2$
link decisions of \emph{just a single dyad}. If linking decisions
were conditionally independent across dyads, then the likelihoods
associated with these two network sequences would differ by only a
single term (corresponding to the direct likelihood contribution of
dyad $ij$). When linking decisions are interdependent, however,
these two likelihoods generally have many terms different, even though the
two network sequences are nearly identical. That two nearly identical
network sequences may have very different likelihoods is a consequence of the interdependence in linking decisions
across dyads induced by a structural taste for transitivity.

To see this, consider the effect, on the form of the likelihood, of
changing $d_{ij1}$ from zero to one. The effect of such a small change
on the structure of the likelihood is complicated. First, due to state
dependence, this change alters the incentive for $i$ and $j$ to
form a link in period $t=2$. Second, the period $t=2$ incentives
for other agents to link with either $i$ or $j$ may change. This
occurs if $r_{il1}$ changes when $d_{ij1}$ does, as would
occur if $l$ and $j$ are linked in period $t=1$. In that case, the
presence of a period $t=1$ link between $i$ and $j$ creates an
opportunity for $i$ and $l$ to engineer triadic closure in period
$t=2$ by linking. Introducing a $ij$ link in period $t=1$ therefore
increases the incentives for certain other links to form in period $t=2$.\footnote{The change in $d_{ij1}$ does not affect period two link incentives for pairs that do not include either $i$ or $j$. This is because a change in $d_{ij1}$ does not alter $r_{kl1}$ for such pairs.}

An implication of Theorem \ref{thm: condition-star-system-likelihood} is that we can control these cascading effects on the likelihood by restricting how the neighborhood surrounding a star system, say $\mathcal{S}_{p\mathbf{i}}$, evolves. A simultaneous swapping of $d_{ij1}$ and $d_{ij2}$ for all dyads $ij \in \mathcal{S}_{p\mathbf{i}}$ leaves the \emph{net} likelihood contribution of all dyads outside of the set $\bar{\mathcal{S}}_{p\mathbf{i}}$ unchanged. Only dyads' likelihood contributions where at least one agent is in $v\left(\mathcal{S}_{p\mathbf{i}}\right)$ are affected by swapping the edges in $\mathcal{S}_{p\mathbf{i}}$. Specifically, while the likelihood contributions of dyads outside the set $\bar{\mathcal{S}}_{p\mathbf{i}}$ may nominally differ, it turns out that, after permuting terms, they can be shown to be identical (this is shown carefully in the Proof). In contrast, the contributions of those dyads in $\bar{\mathcal{S}}_{p\mathbf{i}}$ substantively \emph{do} differ after swapping $d_{ij1}$ and $d_{ij2}$. Crucially, however, the effect of $\mathbf{A}$ on the likelihood does \emph{not} vary across such swaps. This is
why $\mathbf{A}$ does not appear in \eqref{eq: conditional_likelihood_ss_v2}.

\subsection*{Examples of identifying star systems}

To understand the implications of Theorem \ref{thm: condition-star-system-likelihood} it is helpful to consider some explicit examples of identifying star systems. These examples provide insight into how ``swapping the edges" of star systems generate a reference set of networks which conditions away the influence of $\mathbf{A}$, while retaining information about $\theta$. 

\begin{figure}[ht!]
\caption{Identifying a taste for transitive ties}
\begin{center}
    \includegraphics[scale=0.85]{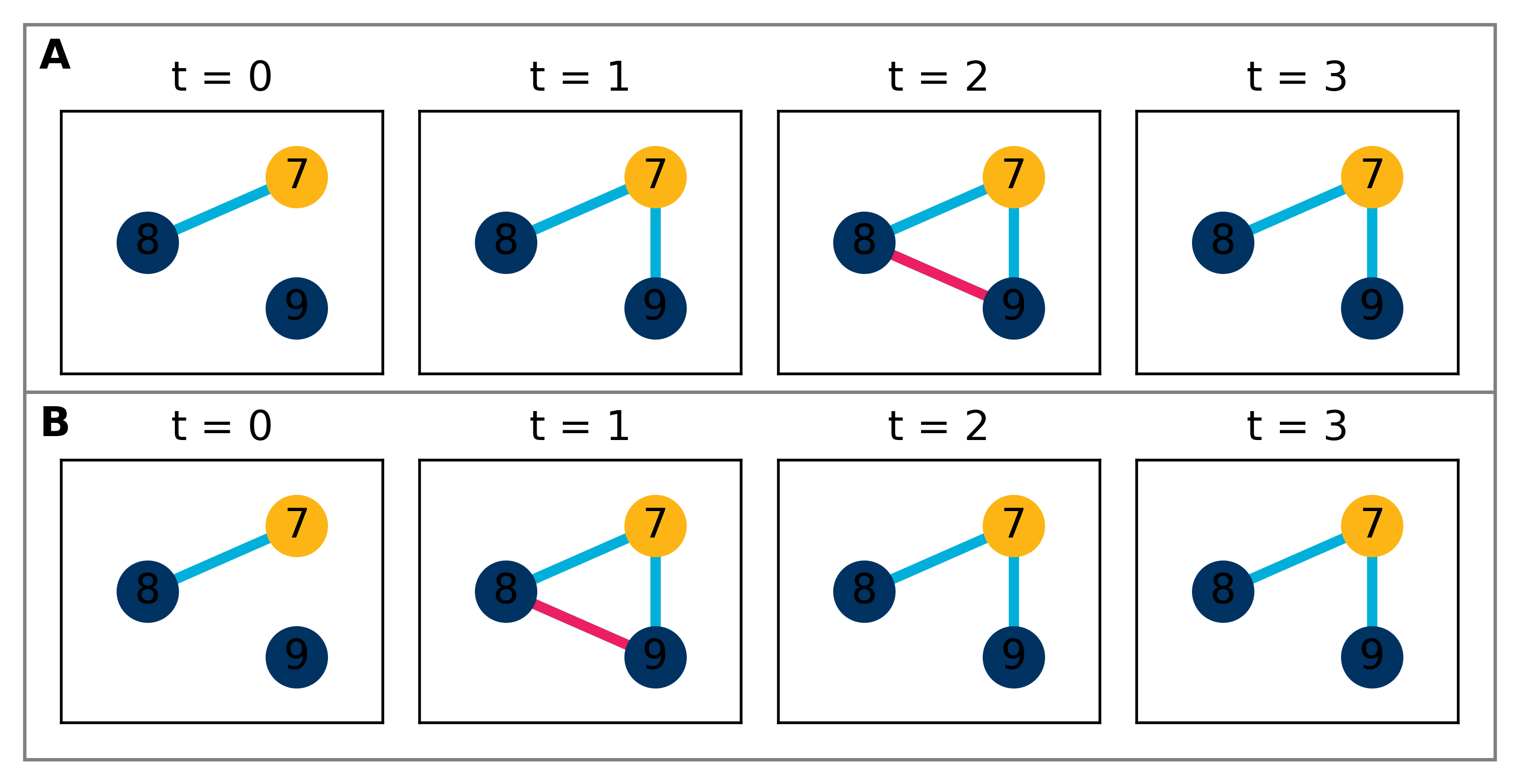}
\end{center}
\caption*{\underline{Source:} Authors' calculations.\hfill\break
\underline{Notes:} The top (Panel A) and bottom (Panel B) rows depict two network sequences. In the top row agents 8 and 9 link in period 2, but not in period 1 (Rose Garden colored edge). In the bottom row they link in period 1, but not in period 2. Observe (i) agents 8 and 9 constitute a stable two star system (since $d_{781}=d_{782}$ and $d_{791}=d_{792}$) and (ii) while they share a link in common (with agent $7$) in period $t=1$, they do not in period $t=0$. Consequently the payoff $8$ and $9$ get from forming a link is higher in period $t=2$ than in period $t=1$. In period $2$ the link generates utility from engineering `triadic closure', no such utility gain is generated by a period $1$ link. Therefore, the top network sequence arises more \emph{relatively} frequently than the bottom in the presence of a structural taste for transitivity in links.\hfill\break }
\label{fig: fig_css_likelihood_example_1}
\end{figure}

Our examples are shown in Figures \ref{fig: fig_css_likelihood_example_1} through \ref{fig: fig_css_likelihood_example_4}. In the figures those nodes which belong to an identifying star system are colored ``Berkeley Blue", while neighboring nodes are colored ``California Gold". Star system edges are colored ``Rose Garden". Each row, or panel, of the figures shows a different star system edge swap.

Our first example is shown in Figure \ref{fig: fig_css_likelihood_example_1}. Agents $8$ and $9$ (the two Berkeley Blue nodes) constitute an identifying $2$ star system. Further observe that, while these two agents share agent $7$ (the California Gold node) as a common friend in period $t=1$, they do not in period $t=0$. Therefore, their returns to linking in period $t=2$, where they can reap the benefits of triadic closure, are higher than the corresponding returns in period $t=1$. In the presence of a structural taste for transitivity, $\beta>0$, we will observe the top sequence \emph{relatively} more frequently than the bottom sequence. 

Figure \ref{fig: fig_css_likelihood_example_2} develops an example of how the relative frequency of two different network sequences provides information about the state dependence parameter $\alpha$. Here the intuition parallels that familiar from the single agent binary choice case (e.g., \citealp*{Cox_JRSS58,Heckman_AI78,Chamberlain_LALMD85}). Namely `runs' in a binary choice sequence, holding the overall number of $1$'s fixed, suggests true state dependence. In our setting, identification additionally requires holding the payoffs from triadic closure fixed. See the notes to the figure for more details.

\begin{figure}[ht!]
\caption{Identifying a structural state dependence}
\begin{center}
    \includegraphics[scale=0.85]{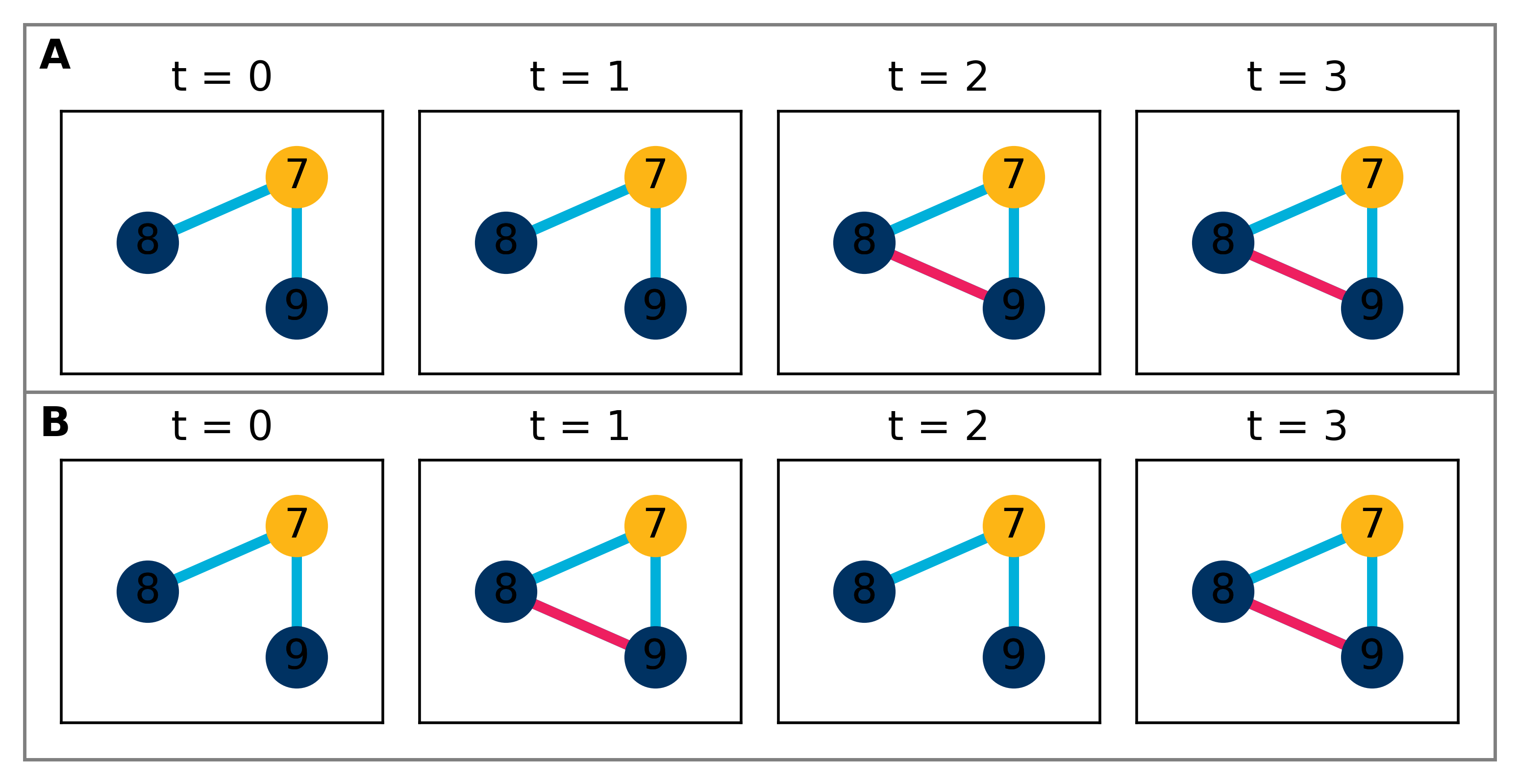}
\end{center}
\caption*{\underline{Source:} Authors' calculations.\hfill\break
\underline{Notes:} See the notes to Figure \ref{fig: fig_css_likelihood_example_1}. In this example $r_{890}=r_{891}=1$, so agents $8$ and $9$ accrue the same transitivity payoff whether linking in periods $1$ or $2$. However,
$d_{890}=0$ and $d_{893}=1$, suggesting that, in the presence of
state dependence ($\alpha>0$), the top sequence (Panel A) will occur \emph{relatively} more
frequently than the bottom (Panel B). The intuition in this case is very similar
to that underlying the results of \citet{Cox_JRSS58}, \citet{Heckman_AI78} and \citet*{Chamberlain_LALMD85}.\hfill\break }
\label{fig: fig_css_likelihood_example_2}
\end{figure}

Figure \ref{fig: fig_css_likelihood_example_3} presents a more complicated example. In this figure dyads $12$ and $13$ form an identifying $3$ star system with $1$ as the center agent. The neighborhood of this system includes agents $4$, $5$, $6$ and $7$ (colored in California Gold). Although tedious, it is straightforward to check that the stability requirements of Definition \ref{def: stable-neighborhood} are satisfied. Inspecting the top sequence, we see that agents $1$ and $2$ share three links in common in period $t=0$ (with agents $4$, $5$ and $6$). Whereas they share no links in common in subsequent periods. In contrast, agents $1$ and $3$ share no connections in common in period $t=0$, but they do share two connections (with agents $4$ and $6$) in periods $t=1$ and $t=2$. This neighborhood architecture implies that the payoff from transitive closure will be larger for agents $1$ and $2$ if they connect in period $t=1$ versus period $t=2$ (since they share many friends in common in period $t=0$). In contrast, the transitivity-driven gains for agents $1$ and $3$ are greater in period $t=2$ versus period $t=1$ (since these two agents shared several connections in common in period $t=1$, but not in period $t=0$). Therefore, in the presence of a structural taste for transitivity, we will observe the top sequence \emph{relatively} more frequently than the bottom one. The reader can verify that the top sequence is also more likely in the presence of structural state dependence.

\begin{figure}[ht!]
\caption{A more complicated example}
\begin{center}
    \includegraphics[scale=0.85]{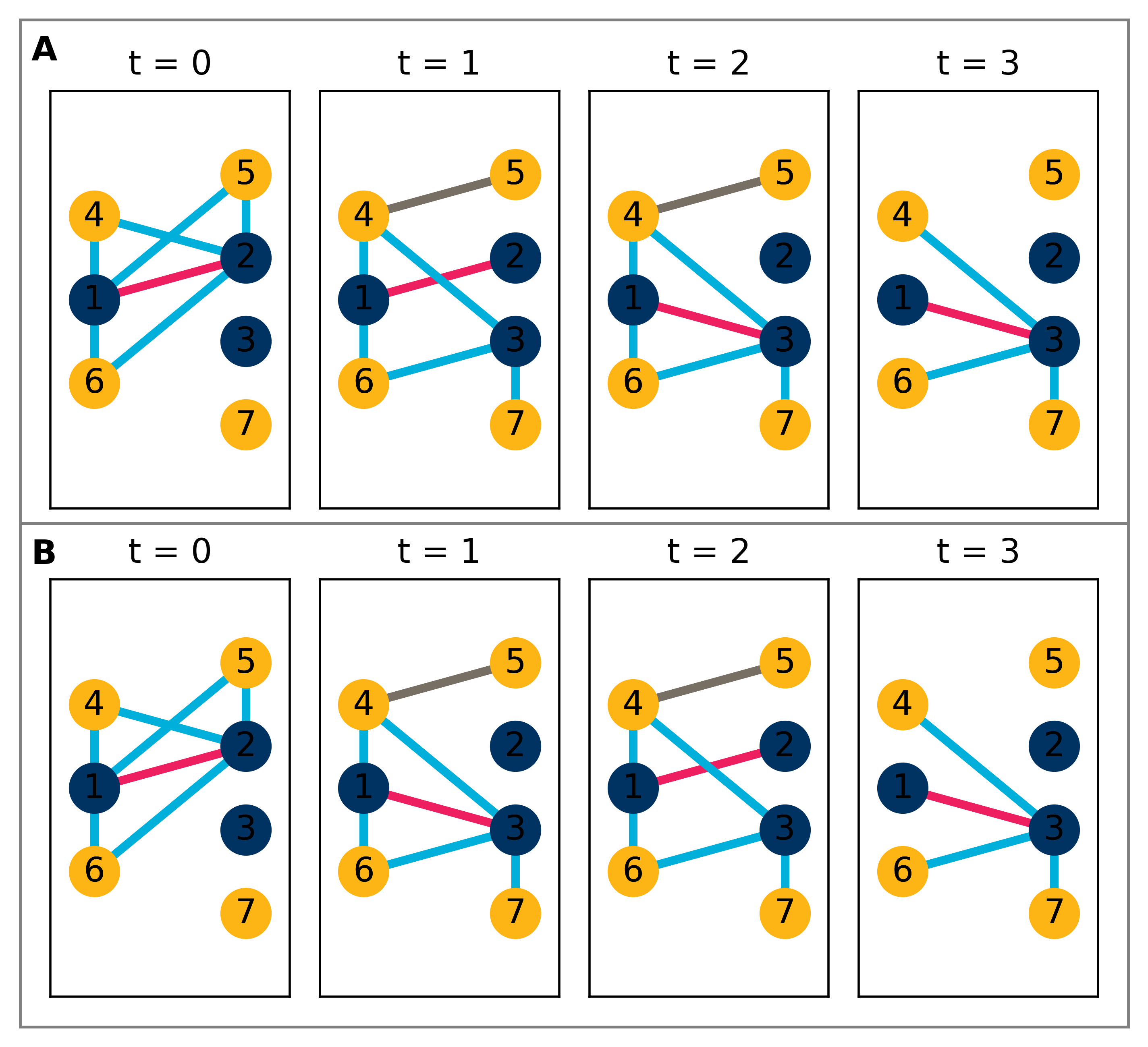}
\end{center}
\caption*{\underline{Source:} Authors' calculations.\hfill\break
\underline{Notes:} The top row (Panel A) depicts an identifying $3$ star, $\mathcal{S}_{3}=\left\{12,13\right\}$. In $t=0$ dyad $12$ shares friends in common with agents $4$, $5$ and $6$. They do not share any friends in common in later periods. In contrast, dyad $13$ has no friends in common in $t=0$, but \emph{does} share connections in common with agents $4$ and $6$ in $t=1$. If $\beta>0$, we would expect to observe the top sequence relatively more frequently than the bottom one (Panel B). This follows since it maximizes the payoffs from transitive closure. \hfill\break }
\label{fig: fig_css_likelihood_example_3}
\end{figure}

Our final example, shown in Figure \ref{fig: fig_css_likelihood_example_4}, combines the star systems shown in Figures \ref{fig: fig_css_likelihood_example_1} and \ref{fig: fig_css_likelihood_example_3} into a single network. Note that the agents forming the two different star systems -- $\mathcal{S}_{3}=\{12,13\}$ and $\mathcal{S}_{2}=\{89\}$ -- although not directly connected, have overlapping neighborhoods (both star systems include agents connected to agent $7$). See Lemmas \ref{lem: separation-nodes-lemma} and \ref{lem: invariance-lemma} above.

Because there are two identifying star systems, the set $\mathbb{B}_{K}\!\left(\mathbf{d}^{3}\right)$ includes four elements (based on whether we swap the edges, or not, in each of the two star systems). For a given value of $\theta$, we can evaluate $\eqref{eq: conditional_likelihood_ss_v2}$. Values of $\theta$ making the observed network sequence, say the one depicted in Panel A, relatively more likely than the other sequences (those in Panels B, C and D) are more data consonant.

Each of the four networks depicted in Figure \ref{fig: fig_css_likelihood_example_4} share common values for those statistics sufficient for $\mathbf{A}$ and $\pi$ (see Proposition \ref{prop: conditioning_set} above). However, the networks \emph{do} vary in their numbers of consecutive links, $\sum_{i<j}\sum_{t=1}^{3}d_{ijt-1}d_{ijt}$, and attempts to engineer triadic closure, $\sum_{i<j}\sum_{t=1}^{3}r_{ijt-1}d_{ijt}$. If, among all networks in the reference set (Panels A, B, C, \& D), the observed network (Panel A) has \emph{relatively} more consecutive links and attempts to close triangles, then we can conclude that $\alpha$ and $\beta$ are larger. The maximum likelihood estimate (MLE) chooses $\hat{\theta}=\left(\hat{\alpha},\hat{\beta}\right)'$ to make likelihood of the network shown in Panel A as large as possible, \emph{relative}, to that of Panels B, C or D. In the next section, we present conditions under which the MLE consistently estimates $\theta_0$.

\section{Inference on $\alpha$ and $\beta$} \label{sec: inference}

In this section we develop two approaches to conditional inference on $\theta=\left(\alpha,\beta\right)'$. The first method is exact (and similar) \citep[cf.,][]{Pelican_Graham_ReStud2026}.\footnote{By `similar' we mean that the size of the test is invariant to the precise instance of the conditioning sufficient statistics. See \cite{Ferguson_Book1967} for a formal definition and/or the related discussion in \cite{Pelican_Graham_ReStud2026}.} This method exploits the fact that the conditional distribution of $\mathbf{D}^3$ given $\mathbf{D}^{3}\in\mathbb{B}_{K}\left(\mathbf{d}^{3}\right)$ is a parametric family indexed by $\theta$ alone. 

The second method involves constructing an estimate of $\theta$ by maximizing \eqref{eq: conditional_likelihood_ss_v2} and characterizing its sampling properties as the network grows large, $N\rightarrow \infty$. This asymptotic argument is delicate: we observe only a single network for $T=3$ periods (plus the initial condition). Consistency and asymptotic normality requires the number of identifying star-systems within the network to grow with $N$; a condition which places subtle restrictions on the form of population network generating process.

\begin{landscape}
\begin{figure}
\caption{Conditional star system (css) likelihood}
\begin{centering}
\includegraphics[scale=0.75]{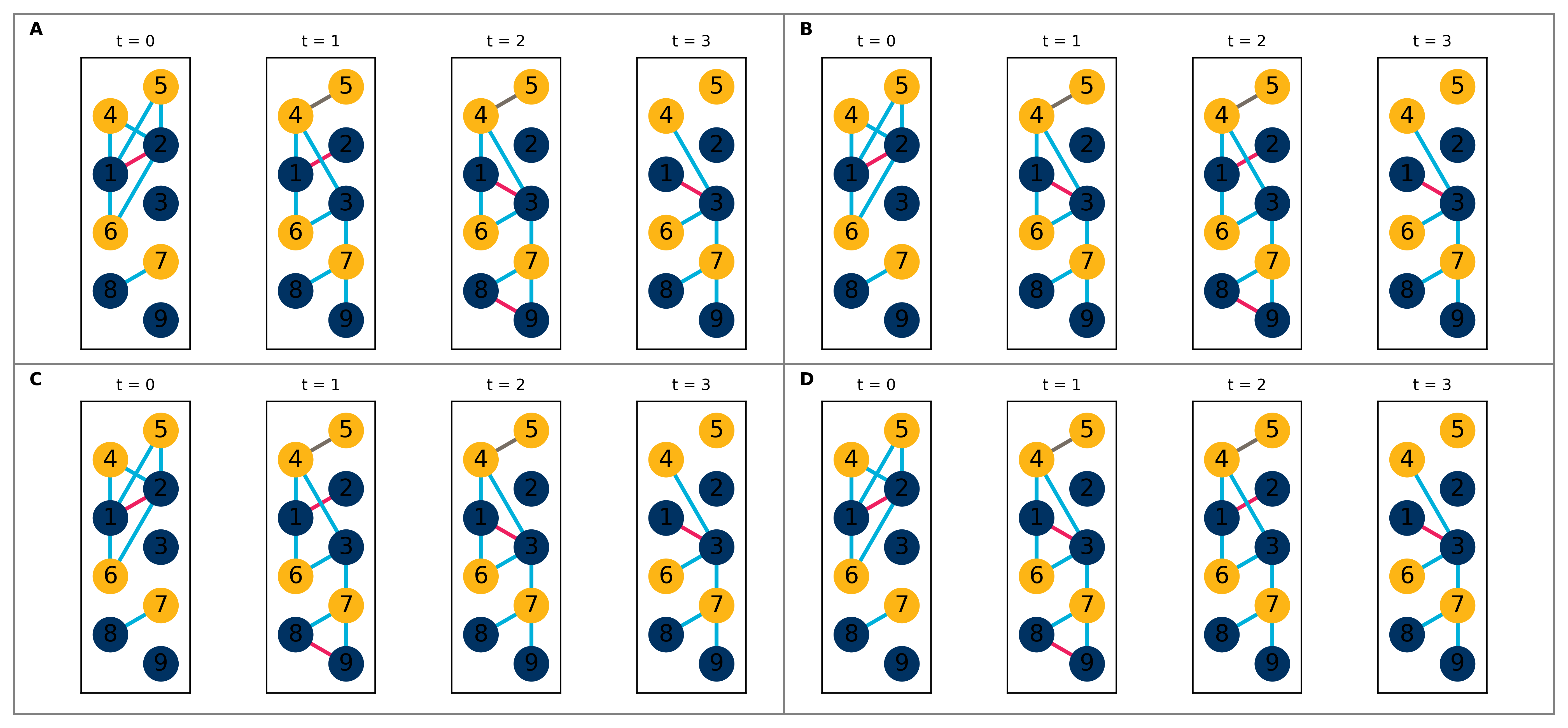}
\end{centering}
\caption*{\underline{Source:} Authors' calculations.\hfill\break
\underline{Notes:} Panel A depicts a hypothetical observed network sequence containing two identifying star systems: the 3-star $\{12,13\}$ and the 2-star $\{89\}$. In this sequence, there are 18 instances of links persisting across two adjacent periods, $\sum_{i<j}\sum_{t=1}^{3}d_{ij,t-1}d_{ijt}=18$, and 16 attempts to engineer triadic closure, $\sum_{i<j}\sum_{t=1}^{3}r_{ij,t-1}d_{ijt}=16$. In Panel B, the period-1 and period-2 edges of the 3-star system are swapped, reducing the persistence total to 16 and the transitivity total to 13. In Panel C, only the period-1 and period-2 edges of the 2-star system are swapped. This leaves the persistence total unchanged at 18 but reduces the transitivity total to 15. In Panel D, both star systems' period-1 and period-2 edges are swapped, yielding a persistence total of 16 and a transitivity total of 12.\hfill }
\label{fig: fig_css_likelihood_example_4}
\end{figure}
\end{landscape}

\subsection*{Exact inference on $\alpha$ and $\beta$}

This subsection describes an \emph{exact} method of inference on $\theta$. The basic approach is due to \cite{Fisher_JRSS1922}. The main difficulty, in the present context, involves enumerating the set $\mathbb{B}_{K}\left(\mathbf{d}^{3}\right)$; a problem which we solve by simulation (as do \cite{McDonald_et_al_SN07} and \cite{Pelican_Graham_ReStud2026} in related contexts). 

\subsubsection*{A Metropolis-Hastings algorithm}
Our main tool is Algorithm \ref{alg: markov_draw}, which describes a Metropolis-Hastings procedure for constructing a random draw from the conditional distribution of $\mathbf{D}^{3}$ given that $\mathbf{D}^{3}\in\mathbb{B}_{K}\left(\mathbf{d}^{3}\right)$. This distribution, for a null structural parameter of $\theta=\theta_{0}$, places a probability weight of $\ell^{css}\left(\mathbf{b}^3;\theta_{0}\right)$ on each $\mathbf{b}^3 \in \mathbb{B}_{K}\left(\mathbf{d}^{3}\right)$ (see equation \eqref{eq: conditional_likelihood_ss_v2} of Theorem \ref{thm: condition-star-system-likelihood} above). 

\begin{algorithm} 
\caption{\textsc{Markov Draw Algorithm}}\label{alg: markov_draw}
\textbf{\underline{Inputs:}} A sequence of adjacency
matrices $\left(\mathbf{b}_{0},\mathbf{b}_{1},\mathbf{b}_{2},\mathbf{b}_{3}\right)=\mathbf{b}^{3}\in\mathbb{B}_{K}\left(\mathbf{d}^{3}\right)$ with a set $\mathcal{G}$ of $\mathbf{m}_{2:K,N}$ stable star systems of size $p=2,\dots,K$; a null value for the common parameter, $\theta_0$;
a mixing time $\sigma.$ \\
\textbf{\underline{Procedure:}}
\begin{enumerate}
\item Set $s=0$.
\item With probability $1-\rho$ go to step 3, with probability $\rho$
go to step $4$.
\item Draw a stable star system $\mathcal{S}$ uniformly at random from $\mathcal{G}$
and permute its period $1$ and $2$ edges (as currently configured
in $\mathbf{b}^{3}$) generating $\mathbf{\tilde{b}}^{3}\in\mathcal{N}\left(\mathbf{b}^{3}\right)\subseteq\mathbb{B}_{K}\left(\mathbf{d}^{3}\right)$.
\begin{enumerate}
\item Draw $U\sim\mathrm{Uniform}\left[0,1\right]$.
\item \textbf{If} $\min\left\{ 1,\frac{\ell^{css}\left(\tilde{\mathbf{b}}^{3};\theta_{0}\right)}{\ell^{css}\left(\mathbf{b}^{3};\theta_{0}\right)}\right\} \geq U$.
\begin{enumerate}
\item Set $\mathbf{b}^{3}=\mathbf{\tilde{b}}^{3}$
\item Go to step 4.
\end{enumerate}
\item \textbf{Else} go to step 4.
\end{enumerate}
\item Set $s=s+1.$
\begin{enumerate}
\item \textbf{If} $s=\sigma$ return $\mathbf{b}^{3}.$
\item \textbf{Else} go to step 2.
\end{enumerate}
\end{enumerate}
\textbf{\uline{Output:}} A random draw, $\mathbf{b}^{3}$, from $\ell^{css}\left(\cdot;\theta_{0}\right).$
\end{algorithm}

Algorithm \ref{alg: markov_draw} stochastically moves between the elements of $\mathbb{B}_{K}\left(\mathbf{d}^{3}\right)$ by (at each iteration $s=1,\ldots,\sigma$) randomly selecting a stable star system and ``swapping its edges" (i.e., permuting its period $1$ and $2$ edges). More precisely, at each iteration we (i) stay at the current state with probability $0<\rho<1$ or (ii), with probability $1-\rho$, pick a stable star system at random and swap its edges. In case (ii) the swap causes the algorithm to move from its current state, say $\mathbf{B}^{3,\left(s\right)}$, to a \emph{proposed} new state -- also in $\mathbb{B}_{K}\left(\mathbf{d}^{3}\right)$ by construction -- of $\tilde{\mathbf{B}}^{3,\left(s+1\right)}$. This proposal is then accepted according to the rule
\[
\mathbf{B}^{3,\left(s+1\right)}=\left\{ \begin{array}{ll}
\tilde{\mathbf{B}}^{3,\left(s+1\right)} & \text{if }-\theta_{0}'w\left(\mathbf{B}^{3,\left(s\right)},\tilde{\mathbf{B}}^{3,\left(s+1\right)}\right)\leq V^{\left(s\right)}\\
\mathbf{B}^{3,\left(s\right)} & \text{otherwise}
\end{array}\right.,
\]
with $V^{\left(s\right)}\sim\mathrm{Exponential}\left(1\right)$\footnote{Recall
that if $U^{\left(s\right)}\sim\mathcal{U}\left[0,1\right]$, then
$-\ln U^{\left(s\right)}\sim\mathrm{Exponential}\left(1\right)$.} and
\[
w\left(\mathbf{B}^{3,\left(s\right)},\tilde{\mathbf{B}}^{3,\left(s+1\right)}\right)=\left(\begin{array}{c}
\sum_{i<j}\sum_{t=1}^{3}\left[\tilde{B}_{ijt-1}^{\left(s+1\right)}\tilde{B}_{ijt}^{\left(s+1\right)}-B_{ijt-1}^{\left(s\right)}B_{ijt}^{\left(s\right)}\right]\\
\sum_{i<j}\sum_{t=1}^{3}\left[\tilde{S}_{ijt-1}^{\left(s+1\right)}\tilde{B}_{ijt}^{\left(s+1\right)}-S_{ijt-1}^{\left(s\right)}B_{ijt}^{\left(s\right)}\right]
\end{array}\right).
\]
Recall that $S_{ijt}=\sum_{k}B_{ikt}B_{jkt}$ and define $\tilde{S}_{ijt-1}$ analogously. Note that $\theta_{0}'w\left(\mathbf{B}^{3,\left(s\right)},\tilde{\mathbf{B}}^{3,\left(s+1\right)}\right)$ equals the logarithm of the likelihood ratio $\sfrac{\ell^{css}\left(\tilde{\mathbf{B}}^{3,\left(s+1\right)};\theta_{0}\right)}{\ell^{css}\left(\mathbf{B}^{3,\left(s\right)};\theta_{0}\right)}$. Inspection of the expressions above reveal, assuming (for clarity of exposition) that the elements of $\theta_0$ are positive, that Algorithm \ref{alg: markov_draw} accepts any move to a network sequence with (sufficiently) ``more" link persistence and/or transitivity. However, moves with reduced link persistence and/or transitivity also occur with positive probability. Because the conditional star system likelihood assigns higher probability to network sequences with more link persistences and/or transitivity, the Markov chain defined by Algorithm \ref{alg: markov_draw} spends more time in regions in $\mathbb{B}_{K}\left(\mathbf{d}^{3}\right)$ assigned high probability by the null model; but it also visits all network sequences in this reference set with positive probability.

Let $\mathbf{B}^{3,\left(0\right)}=\mathbf{d}^{3}$ --
the observed network sequence -- initialize the chain and consider the
simulated chain $\mathbf{B}^{3,\left(1\right)},\mathbf{B}^{3,\left(2\right)},\ldots,\mathbf{B}^{3,\left(\sigma\right)}$.
For $\sigma$ large enough, it is straightforward to verify\footnote{The Markov chain is finite, irreducible and aperiodic and hence ergodic.
The claimed stationary distribution can be derived using, for
example, Theorem 7.10 of \cite{Mitzenmacher_Upfal_PC05}} that
\[
\mathbf{B}^{3,\left(S\right)}\sim\ell^{css}\left(\cdot;\theta_{0}\right).
\]
That is, the ergodic distribution of the Markov process described
above is such that it visits each network in $\mathbb{B}_{K}\left(\mathbf{d}^{3}\right)$
with a long-run frequency of 
\[
\Pr\left(\left.\mathbf{B}^{3,\left(s\right)}=\mathbf{b}^{3}\right|\mathbf{B}^{3,\left(s\right)}\in\mathbb{B}_{K}\left(\mathbf{d}^{3}\right);\theta_{0}\right)=\ell^{css}\left(\mathbf{b}^{3};\theta_{0}\right)
\]
for $\mathbf{b}^{3}\in\mathbb{B}_{K}\left(\mathbf{d}^{3}\right)$
and with zero frequency for $\mathbf{b}^{3}\notin\mathbb{B}_{K}\left(\mathbf{d}^{3}\right).$

\begin{thm}
\label{thm: markov_draw}\textsc{(Null Sampler)} The output of Algorithm \ref{alg: markov_draw}, for $\sigma$ large enough, is a random draw from $\ell^{css}\left(\cdot;\theta_{0}\right).$
\end{thm}
A proof of Theorem \ref{thm: markov_draw} can be found in Appendix \ref{app: markov_draw_proof}. The `large enough' proviso is a limitation. However, based upon related work, we expect the mixing properties of this MCMC chain to be quite good. A rule-of-thumb is to run the chain until the edges of each identifying star system are swapped at least two times, viewing further draws as from the target stationary distribution.

Let $\left(\mathbf{A},\pi\right)\in\triangle$ denote the nuisance
parameter space and $\gamma=\left(\theta,\mathbf{A},\pi\right)$ the
full model parameter. Define the set $\Gamma_{0}=\left\{ \left(\theta,\mathbf{A},\pi\right)\ :\ \theta=\theta_{0},\left(\mathbf{A},\pi\right)\in\triangle\right\} $.
Our goal is to construct a confidence set by inverting a two-sided
test of the composite null
\begin{equation}
H_{0}\ :\ \gamma\in\Gamma_{0}.\label{eq: null_hypothesis}
\end{equation}
Under \eqref{eq: null_hypothesis} the structural parameter is fixed
at $\theta=\theta_{0}$, while the nuisance parameter $\left(\mathbf{A},\pi\right)$
is allowed to range freely over $\triangle$; hence the `composite null'
nomenclature. See \cite{Pelican_Graham_ReStud2026} for more discussion in a related context.

Let $v\left(\mathbf{d}^{3}\right)$ be some test statistic. A starting point might be a count of the number of links persisting for two periods \emph{plus} the number of two-stars `closing' into triangles the following period:
\[v\left(\mathbf{d}^{3}\right)=\sum_{i<j}\sum_{t=1}^{3}d_{ijt-1}d_{ijt}+\sum_{i<j}\sum_{t=1}^{3}r_{ijt-1}d_{ijt}.
\]
Under $H_{0}$ the conditional distribution of $v\left(\mathbf{D}^{3}\right)$
given $\mathbf{D}^{3}\in\mathbb{B}_{K}\left(\mathbf{d}^{3}\right)$
equals
\[
F_{\left.v\left(\mathbf{D}^{3}\right)\right|\mathbf{D}^{3}\in\mathbb{B}_{K}\left(\mathbf{d}^{3}\right)}\left(\left.c\right|H_{0}\right)=\sum_{\mathbf{b}^{3}\in\mathbb{B}_{K}\left(\mathbf{d}^{3}\right)}\mathbf{1}\left(v\left(\mathbf{b}^{3}\right)\leq c\right)\cdot\ell^{css}\left(\mathbf{b}^{3};\theta_{0},\mathbf{A},\pi\right),
\]
which can be approximated by
\begin{equation}
\hat{F}_{\left.v\left(\mathbf{D}^{3}\right)\right|\mathbf{D}^{3}\in\mathbb{B}_{K}\left(\mathbf{d}^{3}\right)}\left(\left.c\right|H_{0}\right)=\frac{1}{S}\sum_{s=1}^{S}\mathbf{1}\left(v\left(\mathbf{B}^{3,\left(s\right)}\right)\leq c\right),\label{eq: test_statistic_null distribution}
\end{equation}
where $s=1,\ldots,S$ indexes draws from the ergodic distribution
of the Markov chain defined by Algorithm \ref{alg: markov_draw}. If, for a given $\theta=\theta_{0}$,
$v\left(\mathbf{D}^{3}\right)$ lies between the $\alpha/2$ and $1-\alpha/2$
quantiles of distribution \eqref{eq: test_statistic_null distribution},
then we can conclude that $\theta_{0}$ lies within a $1-\alpha$
confidence set for $\theta$. By repeating this process for an appropriately
chosen grid of $\theta_{0}$ values we can compute the boundaries
of this confidence set.

Observe that the coverage of this confidence set \textit{exactly equals}  $1-\alpha$
(modulo simulation/numerical approximation error which can be minimized
with sufficient compute).\footnote{There is also a small amount of coverage distortion induced by the
fact that $v\left(\mathbf{D}^{3}\right)$ is a discrete random variable
with finite support. This can be formally resolved by inverting a
randomized test. However, this issue is practically irrelevant when
$N$ is modestly large.} Invariance of the conditional distribution, to where $\left(\mathbf{A},\pi\right)$ in $\triangle$ lies, ensures that the confidence set is \textit{similar} (in the sense that its coverage is exactly $1-\alpha$ irrespective of the population value of the nuisance parameter). We leave a more careful exploration of exact inference on $\theta_0$ to future work.

\subsection*{Asymptotic inference on $\alpha$ and $\beta$}

In this section we study the properties of the maximizer, $\hat\theta$, of the conditional star system log-likelihood (see Theorem \ref{thm: condition-star-system-likelihood} above). Specifically, we consider the approximate repeated sampling properties of $\hat\theta$. The econometrician observes a single network sequence, $\mathbf{D}^3=\mathbf{d}^3$. We consider the behavior of $\hat\theta$ across repeated draws of $\mathbf{D}^3$. Our asymptotic approximation is with respect to the number of agents in these single network sequence draws (i.e., $N \rightarrow \infty)$.

\subsubsection*{Notation and setup}
In what follows, it will be convenient to iterate over all star systems, whether identifying or not. We use $\mathbf{i}$ to index a generic star system (of order $p=2,\ldots,K$) and, also, in an abuse of notation the dyad set which defines it. So that $\mathbf{i}$ can denote both a counting number and a set of the form $\mathbf{i}=\left\{i_1 i_2,i_1 i_3,\ldots,i_1 i_p\right\}$. We further use $v\left(\mathbf{i}\right)$ to denote the vertex set $\left\{i_1, \ldots i_p \right\}$.

Observe that there are $\mathbf{n}_p\overset{def}{\equiv}\binom{N}{p} \times p$ star systems of size $p$. To keep track of whether a star system is identifying, we define the binary indicator variable $Z_{p\mathbf{i}}$, which equals $1$ if the $\mathbf{i}^{th}$ star system of size $p$ is identifying and $0$ otherwise. Define $\mathbf{Z}:=\left(Z_{p\mathbf{i}}\right)_{2\leq p \leq K, \mathbf{i}=\mathbf{1},\ldots,\mathbf{n}_p}$.

Recall that $Q_{ij}=\left(D_{ij0},D_{ij1},D_{ij2},D_{ij3},R_{ij0},R_{ij1},R_{ij2}\right)'$.
Next define the tuples
\[Z_{p\mathbf{i}}\mathbf{Q}_{p\mathbf{i}}:=Z_{p\mathbf{i}} \cdot \left(Q_{ij}\ :\  i \in v\left(\mathbf{i}\right), j \in \mathcal{V} - \cup_{q=2}^{K}\cup_{\mathbf{k}\in\left\{ \mathbf{k}=\mathbf{1},\ldots,\mathbf{n}_{p}\ :\ Z_{q\mathbf{k}}=1\right\} }v\left(\mathbf{k}\right)+v\left(\mathbf{i}\right)-i\right).
\]
Here we multiply by $Z_{p\mathbf{i}}$ to denote the conditional selection (or filtering) of those $Q_{ij}$ associated with stable star system $\mathbf{i}$.  The vertex set $\mathcal{V} - \left\{\cup_{q=2}^{K}\cup_{\mathbf{k}\in\left\{ \mathbf{k}=\mathbf{1},\ldots,\mathbf{n}_{p}\ :\ Z_{q\mathbf{k}}=1\right\} }v\left(\mathbf{k}\right)+v\left(\mathbf{i}\right)-i\right\}$ includes all agents $j$ not in the vertex set of stable star systems \emph{other than} $\mathbf{i}$. Hence $Z_{p\mathbf{i}}\mathbf{Q}_{p\mathbf{i}}$ returns all $Q_{ij}$ where at least one of $i$ or $j$ is in $v\left(\mathbf{i}\right)$, $\mathbf{i}$ is identifying, and every $Q_{ij}$ appears at most once in $Z_{p\mathbf{i}}\mathbf{Q}_{p\mathbf{i}}$. 

At the heart of our large sample theory is the following conditional independence result:

\begin{lem}
\label{lem: conditional-independence}\textsc{(Conditional Independence)} The tuples $Z_{p\mathbf{i}}\mathbf{Q}_{p\mathbf{i}}$  are independent  across indices $\mathbf{i}\in \{i_1i_2,\dots, i_1i_q: q\leq K \}$ conditional on  $\mathbf{D}_0,\mathbf{D}_3,\mathbf{Z}:=(\mathbf Z_{q,\mathbf{i}})_{\mathbf{i}, q\leq K}$ and $\mathbf{A}$.
\end{lem}
A proof is in Appendix \ref{app: conditional-independence_proof}.

Define the normalized version of the conditional star system log-likelihood:
\begin{align}
\hat{Q}_N\left(\theta\right)&=\frac{1}{\mathbb{E}\left[\sum_{p,\mathbf{i}}Z_{p\mathbf{i}}\right]}\ln\left[\prod_{p=2}^K \prod_{\mathbf{i}=\mathbf{1}}^{\mathbf{m}_{p,N}} \left[1+\prod_{ij\in \bar{\mathcal{S}}_{p\mathbf{i}}} b(Q_{ij};\theta)\right]^{-1}\right] \notag \\
&=\frac{1}{\mathbb{E}\left[\sum_{p,\mathbf{i}}Z_{p\mathbf{i}}\right]}\sum_{p=2}^K \sum_{\mathbf{i}=\left\{i_1 i_2,\ldots,i_1 i_p\right\}} Z_{p\mathbf{i}} \cdot \ln\left[\frac{\exp(T_{p\mathbf{i}}'\theta)}{1+\exp(T_{p\mathbf{i}}'\theta)}\right], \label{eq: conditional_likelihood_ss_v3}
\end{align}
where the second equality follows from the notation defined immediately above. Our estimate coincides with $\hat\theta = \arg \underset{\theta}{\max} \ \hat Q_N\left(\theta\right)$. By Lemma \ref{lem: conditional-independence} the summands of $\hat Q_N\left(\theta\right)$ are conditionally independent of one another. However the number of these summands which are non-zero is a random variable; equal to the number of identifying star systems in $\mathbf{D}^3$. As long as the number of such systems, $\sum_{p,\mathbf{i}}Z_{p\mathbf{i}}$, grows as $N \rightarrow \infty$, the conditional independence structure established by Lemma \ref{lem: conditional-independence} is sufficient, under some additional regularity conditions, to formally establish consistency and asymptotic normality of $\hat\theta$.

Our consistency result is based on the following four assumptions.
\begin{assumption}\label{ass: parameter_space}\textsc{(Parameter Space)} $\theta_0 \in \mathrm{int}\left(\Theta\right)$, a closed ball in $\mathbb{R}^2$.
\end{assumption}

\begin{assumption}\label{ass: rank}\textsc{(Rank)} There exists a constant $C_1>0$ such that for $N$ large enough:
\[\lambda_{\min}\left(\frac{1}{\mathbb{E}\left[ \sum_{p,\mathbf{i}}  Z_{p\mathbf{i}} \right]} \mathbb{E}\left[\left.\sum_{p,\mathbf{i}}  Z_{p\mathbf{i}} T_{p\mathbf{i}}T_{p\mathbf{i}}' \right| \mathbf{D}_0,\mathbf{D}_3,\mathbf{Z}\right ]\right)\geq C_1.
\]
\end{assumption}

\begin{assumption}\label{ass: boundedness}\textsc{(Boundedness)} There exists   $C_2>0$ such that for $N$ large enough: 
\[\frac{1}{\mathbb{E}\left[\sum_{p,\mathbf{i}}  Z_{p\mathbf{i}}  \right]}\sum_{p,\mathbf{i}}  Z_{p\mathbf{i}} \cdot \mathbb{E}\left[\left. \| T_{p\mathbf{i}}\|^4\right|\mathbf{D}_0,\mathbf{D}_3, \mathbf{Z}\right]\leq C_2.
\]
\end{assumption}

\begin{assumption}\label{ass: identifying_events}\textsc{(Identifying Events)} The sequence $\mathbb{E}\left[\sum_{p,\mathbf{i}}Z_{p\mathbf{i}}\right]$ diverges to $+\infty$ as $N \rightarrow \infty$.
\end{assumption}

Assumption \ref{ass: identifying_events}, which is key, ensures that our criterion function includes a growing number of conditionally independent summands as $N \rightarrow \infty$. Assumption \ref{ass: parameter_space} is standard. Assumption \ref{ass: rank} is a rank condition, while Assumption \ref{ass: boundedness} restricts the first four moments of $T_{p\mathbf{i}}$. 

Observe that, since $T_{p\mathbf{i}}$ is a direct function of $\mathbf{D}^3$, Assumptions \ref{ass: rank} to Assumption \ref{ass: identifying_events} involve implicit restrictions on the distribution of $\mathbf{A}$ and the initial network condition $\pi$. For this reason, they are not straightforward to verify. Below we develop examples of data-generating processes where Assumptions \ref{ass: identifying_events} provably holds, and provably fails to hold. These examples gives some insight into what types of networks are likely to include a sufficient number of identifying star systems to support estimation and inference. 

Of course, for any particular network under actual consideration, the researcher necessarily counts the number of identifying star systems available and directly assesses whether they arise in sufficient numbers to justify estimation (and large sample inference). Since $\theta$ only includes two elements, and the conditional star system log-likelihood is globally concave, we expect our large sample theory for $\hat{\theta}$ to be quite accurate even if settings with relatively modest numbers of identifying star systems. This accords with the limited Monte Carlo evidence summarized below.

\begin{thm}{\label{thm: consistency}}\textsc{(Consistency)}
    Under Assumptions \ref{ass: parameter_space} to \ref{ass: identifying_events},
    $ \hat \theta \overset{p}{\rightarrow} \theta_0$.
\end{thm}
A proof is available in Appendix \ref{app: main-results}.

Let $\Lambda\left(x\right)\overset{def}{\equiv}\frac{\exp(x)}{1+\exp(x)}$. Our demonstration of asymptotic normality requires the additional assumption:

\begin{assumption}\label{ass: hessian}\textsc{(Hessian)} 
The normalized (negative) Hessian $$\frac{1}{\sum_{p,\mathbf{i}} \mathbb{E}\left[Z_{p\mathbf{i}}\right]} \mathbb{E}\left[\sum_{p, \mathbf{i}} Z_{p\mathbf{i}}\cdot\Lambda(T_{p\mathbf{i}} '\theta_0)(1-\Lambda(T_{p\mathbf{i}} '\theta_0)) T_{p\mathbf{i}}T_{p\mathbf{i}}'\right]$$ converges to some positive definite  matrix $\mathcal{I}\left(\theta_0\right)$.
\end{assumption}

\begin{thm}{\label{thm: asymptotic-normality}}\textsc{(Asymptotic Normality)}
    Under Assumptions \ref{ass: parameter_space} to \ref{ass: hessian}:
    $$\sqrt{\mathbb{E}\left[\sum_{p,\mathbf{i}}Z_{p\mathbf{i}}\right]}\cdot \left(\hat{\theta}-\theta_0\right) \overset{D}{\rightarrow} \mathcal{N}\left(0, \mathcal{I}^{-1}\left(\theta_0\right)\right).$$
\end{thm}
A complete proof is available in Appendix \ref{app: main-results}. As expected, the rate-of-convergence depends on the (expected) number of identifying star systems; corresponding to the ``effective sample size" in our setting. The examples we develop in the next subsection indicate that the behavior of $\mathbb{E}\left[\sum_{p,\mathbf{i}}Z_{p\mathbf{i}}\right]$ as $N \rightarrow \infty$  varies with the distribution of $\mathbf{A}$ as well as the initial network condition, $\pi$. Hence, while our main identification result, Theorem \ref{thm: condition-star-system-likelihood} above, holds for \emph{any} $\mathbf{A}$ and $\pi$, the precision with which we can learn $\theta$ is not invariant to these nuisance parameters. Indeed, for some distributions of heterogeneity, consistent estimation is not possible. By Corollary \ref{cor: separation-dyads-corro}, the number of identifying star systems is, at most, of order $N$, suggesting a maximal rate-of-convergence of $\sqrt{N}$.

\subsection*{Single network estimation: examples}

In this subsection, by means of two examples, we explore conditions under which the number of identifying systems diverges as $N$ grows large (Assumption \ref{ass: identifying_events}). Such divergence is required for consistency and asymptotic normality. Whether Assumption \ref{ass: identifying_events} holds depends on the model's parameters, the distribution of the fixed effects, and the initial network. 

We study two models. In the first, no identifying star systems are observed (on average) when $N$ is large. In the second, the expected number of identifying star systems grows linearly with $N$ (the maximal rate). These two examples provide some indication of the positive possibilities associated with our method, as well as its limits.

\textbf{Example 1.} \textsc{(Dense Initial Graph and Compactly Supported Fixed Effects)} Suppose:
\begin{itemize}
    \item  the fixed effects $A_{ij}$ have bounded support, say $|A_{ij}|\leq \Bar{A}$ for all $i,j$ almost surely.
    \item $D_0$ is an Erdos-Renyi random graph with edge probability $p\in (0,1)$.
\end{itemize}

Under this DGP, for any $i$ and $j$, $R_{ij0} \sim (N-2)p^2$. When $\beta_0 > 0$, this implies that with high probability, $-|\alpha_0| + \beta_0 R_{ij0} - \bar{A} \gg U_{ij1}$, such that $\mathbf{D}_1$ approaches a complete graph. Consequently, $\mathbf{D}_2$ and $\mathbf{D}_3$ also become ``nearly complete''. In such dense graph sequences, very few edges transition between periods 1 and 2, making condition (ii) of Definition \ref{def: identifying-star-system} difficult to satisfy.\\

Conversely, when $\beta_0 < 0$, we show that with high probability, $\mathbf{D}_1=0$. Therefore, conditional on the event that $\mathbf{D}_1=0$,  for a given star system to be  identifying it must be that: \textit{i)} the nodes involved in the star become connected with the central node in the second period, \textit{ii)} these nodes remain disconnected with the rest of the network and \textit{iii)} the non central nodes remain disconnected from each other. We show that \textit{i)}, \textit{ii)} and \textit{iii)} only happen with a vanishing probability, which shows that identifying stars form with small probability.  Specifically, Proposition \ref{prop:ER} demonstrates that $\sum_K \mathbf{m}_{K,N}$, the total number of identifying stars, vanishes as $N$ grows. The complete proof is provided in Appendix \ref{proof:ER}.

\begin{prop}\label{prop:ER}
    Under the DGP of Example 1:
 $$ E\left(\sum_K \mathbf{m}_{K,N}\right)\rightarrow_{N\rightarrow \infty} 0.$$
\end{prop}

\textbf{Example 2.}\textsc{ (Random Geometric Graph)}
Let  $z_1,\dots,z_N$ be i.i.d.\ uniform on the square $[0,\sqrt{N}]^2$. Fix a finite real number $a$ and define
\[
A_{ij} :=
\begin{cases}
a, & \|z_i-z_j\|\le r,\\
-\infty, & \|z_i-z_j\|>r.
\end{cases}
\]
For $
D_{ij0}:=\mathbf{1}\left(A_{ij}-U_{ij0}\geq0\right)$ for all $i,j$, we can show that -- on average -- the number of stable systems of any fixed size grows linearly in the sample size $N$.\footnote{In fact, any $\bold{D}_0 $ such that $\mathbb{P}(D_{ij,0}=0|A)$ is bounded away from 0 works.}
\begin{prop}\label{prop:RGG}
     Let $\bold{m}_{K,N}$ be the number of identifying $K-$systems, 
     then\footnote{For any sequences of real  numbers $u_n$ and $v_n$, we write $u_n\asymp v_n$ when $u_n=O(v_n)$ and $v_n=O(u_n)$.}
\begin{enumerate}
\item  $
    \mathbb E\!\left[\sum_{K=2}^{ N-1} \bold{m}_{K,N}\right]\asymp N$
\item $
\mathbb E\!\left[\sum_{K=2}^{ \bar K} \bold{m}_{K,N}\right]\asymp N$ for any fixed $ \bar K\leq N-1$.
\end{enumerate}
\end{prop}
A formal proof is in Appendix \ref{proof:RGG}. The proof strategy relies on establishing  upper and lower bounds on the probability that any specific group of nodes forms an identifying star. For the upper bound, we note that for the ``spokes edges" of a star to switch between $t=1$ and $t=2$, all ``leaf nodes" must randomly fall within a small, fixed distance, from the central node: an event whose probability shrinks rapidly with $N$. For the lower bound, the proof constructs an ``isolated bubble" scenario: the star's nodes are clustered tightly together, all other  nodes are forced far enough away to prevent any outside interference, and the internal links follow a specific sequence of forming and breaking. Together, the two bounds imply that the probability of a specific $K$-star acting as an identifying star decays proportionally to $N^{-(K-1)}$. Because the total number of possible node combinations to form a $K$-star grows proportionally to $N^K$, multiplying these factors shows that the expected total number of identifying stars scales linearly with the overall network size, $N$. Note that by Corollary \ref{cor: separation-dyads-corro}, this is the fastest rate the number of stable systems could reach under any DGP.\\

\section{Monte Carlo experiments}\label{sec: Monte_Carlo_experiments}

The Monte Carlo design uses a random geometric graph to construct
an ``opportunity graph'' for link formation (as in the second example considered in the previous section). Specifically agents
are scattered uniformly on the two-dimensional plane
\[
\left[0,\sqrt{N}\right]\times\left[0,\sqrt{N}\right].
\]
The initial network is then generated according to the rule
\[
D_{ij0}=\mathbf{1}\left(A_{ij}-U_{ij0}\geq0\right),
\]
with $U_{ij0}$ logistic and $A_{ij}$ taking one of two values.
If the Euclidean distance between $i$ and $j$ is less than or equal
to $r,$ then $A_{ij}=\ln\left(\frac{0.75}{1-0.75}\right)$, otherwise
$A_{ij}$ equals negative infinity. This calibration means that,
in the initial period, agents that are less than $r$ apart link with
probability 0.75, while those greater than $r$ apart link with probability
zero. 

The expected degree of a randomly sampled agent in $t=0$ is approximately
$0.75\pi r^{2}$ (An exact expression for average degree, which takes
into account boundary effects, can be calculated along the lines of
\citet{kostin2010probability}). An implication of this set-up is that the
initial condition is sparse: average degree does not increase with
network size. By varying the value of $r$, we can manipulate the average
degree, and hence connectivity, of the initial condition. We choose
values of $r$ such that in large networks average degree in period
$t=0$ is 2, 3 or 4.\footnote{As a point of reference \citet*{McPherson_et_al_ASR06} estimate the
average size of adult Americans' core discussion networks (i.e., confidants
with whom individuals discuss important matters) was about three in
1985 and two in 2004.} With an average degree of 2, the resulting initial condition consists
of many small disconnected components. When average degree equals
3 a large giant component begins to form. Finally when average degree
equals 4, more than half of agents are part of one giant component (see Table
\ref{tab: MonteCarlo_DesignStats}). It is well-known that a phase-transition
occurs at an average degree of 3 in random geometric graphs \citep{Penrose_Bk2003}.

In period $t=1,2,3$ links evolve according to rule (\ref{eq:ij_Link_Model})
with $\alpha_0=\beta_0=1.$\footnote{Since agents greater that $r$ apart will never link, average degree
in the network is bounded above by $\pi r^{2}$ in periods $t=1,2,3$.
This corresponds to maximum average degrees of (approximately) 2.67,
4 and 5.33 across the three designs.} A key feature of this design is that it generates homophilous link
formation based on location (which is unobserved by the econometrician).
Although there are no network effects in operation in period $t=0$,
measured transitivity of the initial network is quite high, with the
clustering coefficient exceeding 0.4 in all cases (see Table \ref{tab: MonteCarlo_DesignStats}).
In subsequent periods transitivity and average degree increase as
agents form additional links (on net) in response to state dependence
and a taste for triadic closure. Link clustering in these designs therefore
arises from both homophily and a taste for triadic closure, making
them an appropriate test case for evaluating the small sample relevance
of Theorem \ref{thm: asymptotic-normality}. 

A typical sequence of networks, with $N = 100$ and average degree in the initial network set equal to 2, is depicted in Figure \ref{fig: MonteCarlo}. The figure shows that identifying star systems can come in different forms and arise in both sparse and relatively dense regions. In the figure the larger Berkeley
blue and brown nodes correspond to identifying 2-star and 3-star systems respectively, while the California gold nodes are their immediate neighbors.

\begin{table}[!htb]
    \caption{Basic Properties of Simulated Network Sequences}
    \label{tab: MonteCarlo_DesignStats}
    \centerline{
    \begin{tabular}{l|ccc|ccc|ccc|} \hline 
Asym. Degree & \multicolumn{3}{c|}{2} & \multicolumn{3}{c|}{3} & \multicolumn{3}{c|}{4} \\ 
\cline{1-4}\cline{5-7}\cline{8-10}  
Period & $(N-1)\mathbb{E}[D_{it}]$ & T & GC & $(N-1)\mathbb{E}[D_{it}]$ & T & GC & $(N-1)\mathbb{E}[D_{ijt}]$ & T & GC \\ 
\cline{1-4}\cline{5-7}\cline{8-10}  
$t=0$ &1.98 & 0.44 & 0.01 & 2.96 & 0.44 & 0.05 & 3.94 & 0.44 & 0.58\\ 
$t=1$ &2.41 & 0.58 & 0.01 & 3.68 & 0.58 & 0.08 & 4.97 & 0.58 & 0.83\\ 
$t=2$ &2.49 & 0.59 & 0.01 & 3.80 & 0.59 & 0.09 & 5.12 & 0.59 & 0.85\\ 
$t=3$ &2.50 & 0.59 & 0.01 & 3.82 & 0.59 & 0.09 & 5.14 & 0.59 & 0.85\\ 
\hline 
\end{tabular}
    }
    \smallskip{}
{\small\uline{Notes:}}{\small{} The table reports period-specific network
summary statistics across the $B=1,000$ Monte Carlo simulations for
each design ($N=5,000$). See main text for other design details.
The $(N-1)\mathbb{E}\left[D_{ijt}\right]$ column gives the average degree,
T the global clustering coefficient, or transitivity index, and GC the
fraction of agents that are part of the largest giant component.}{\small\par}
\end{table}

Table \ref{tab: MonteCarlo_CSS_estimates} summarizes the sampling properties of the CSS-MLE for $\alpha_0,\beta_0$ across the different designs. Panel A presents results for the conditional $2$-star system estimator, whereas Panel B reports results for the conditional $2 \ \& \ 3$-stars system estimator, which pools information from all stable star systems of sizes $2$ \emph{and} $3$.  The number of agents is set to $N=5,000$ and we complete $B=1,000$ Monte Carlo replications. For the designs considered here, with five thousand agents, the number of stable $2$-star systems
averages between $136$ and $290$ while the number of stable $3$-star systems ranges between $32$ and $60$. This implies a ratio of stable $3$-star systems to stable $2$-star systems of about $20-24$ percent.\footnote{In light of the relatively small number of available $3$-stars under this DGP, we have not searched for additional identifying star systems of larger dimensions.} 



Inspection of Table \ref{tab: MonteCarlo_CSS_estimates} shows that both estimators are approximately
mean and median unbiased and that the associated Wald-based confidence intervals have
coverage close to nominal 95 percent coverage. Furthermore, consistent with the theory, the conditional $2 \ \& \ 3$-star estimator (Panel B) is relatively more precise than its $2$-star counterpart (Panel A); as indicated by both a smaller mean absolute error and mean standard deviation across all designs.


\begin{table}[!htb]
    \caption{Sampling properties of the conditional star systems estimator}
    \label{tab: MonteCarlo_CSS_estimates}
    \centerline{
    \begin{tabular}{l|cc|cc|cc|} \hline
Asymptotic Degree & \multicolumn{2}{c|}{2} & \multicolumn{2}{c|}{3} & \multicolumn{2}{c|}{4} \\ 
\cline{1-3}\cline{4-5}\cline{6-7}
$N=5000$ & $\alpha$ & $\beta$ & $\alpha$ & $\beta$ & $\alpha$ & $\beta$ \\ 
\hline
\multicolumn{7}{l}{\textbf{Panel A: Conditional 2-star systems estimator}} \\ 
\hline
Mean Bias & 0.0036 & 0.0326 & 0.0290 & 0.0443 & 0.0562 & 0.0400\\ 
Median Bias & -0.0020 & 0.0040 & 0.0090 & 0.0156 & 0.0266 & 0.0168\\ 
Mean Abs. Err. & 0.1835 & 0.1986 & 0.2377 & 0.2003 & 0.3169 & 0.2096\\ 
Std. Dev. & 0.2313 & 0.2554 & 0.2954 & 0.2525 & 0.4028 & 0.2688\\ 
Mean Std. Err. & 0.2315 & 0.2448 & 0.2971 & 0.2442 & 0.3926 & 0.2608\\ 
Coverage (95\% CI) & 0.9540 & 0.9550 & 0.9610 & 0.9490 & 0.9640 & 0.9610\\ 
\hline
\multicolumn{7}{l}{\textbf{Panel B: Conditional 2 \& 3-star systems estimator}} \\ 
\hline
Mean Bias & 0.0035 & 0.0289 & 0.0295 & 0.0298 & 0.0399 & 0.0350\\ 
Median Bias & 0.0083 & 0.0073 & 0.0162 & 0.0145 & 0.0102 & 0.0172\\ 
Mean Abs. Err. & 0.1633 & 0.1717 & 0.2011 & 0.1672 & 0.2738 & 0.1832\\ 
Std. Dev. & 0.2015 & 0.2169 & 0.2510 & 0.2112 & 0.3461 & 0.2326\\ 
Mean Std. Err. & 0.1993 & 0.2079 & 0.2513 & 0.2084 & 0.3295 & 0.2271\\ 
Coverage (95\% CI) & 0.9550 & 0.9530 & 0.9590 & 0.9560 & 0.9520 & 0.9590\\ 
\hline
Avg. \# of Stable 2-star systems & \multicolumn{2}{c|}{290.0} & \multicolumn{2}{c|}{203.3} & \multicolumn{2}{c|}{136.8}\\ 
Avg. \# of Stable 3-star systems & \multicolumn{2}{c|}{59.5} & \multicolumn{2}{c|}{47.1} & \multicolumn{2}{c|}{32.2}\\ 
\hline
\end{tabular}
    }
{\small\uline{Notes:}}{\small{} The table reports the sampling properties of $\hat\theta$ (using 2-star systems in Panel A, and then 2-star and 3-star systems together in Panel B) across the $B=1,000$ Monte Carlo simulations for each design ($N=5,000$). See the main text for other design details.}{\small\par}
\end{table}

\newpage 
\begin{figure}[htbp]
\caption{Illustrative network sequence from Monte Carlo Design}
\label{fig: MonteCarlo}
\begin{center}

\begin{subfigure}{0.45\textwidth}
    \centering
        \caption*{$t=0$}
    \includegraphics[width=\linewidth]{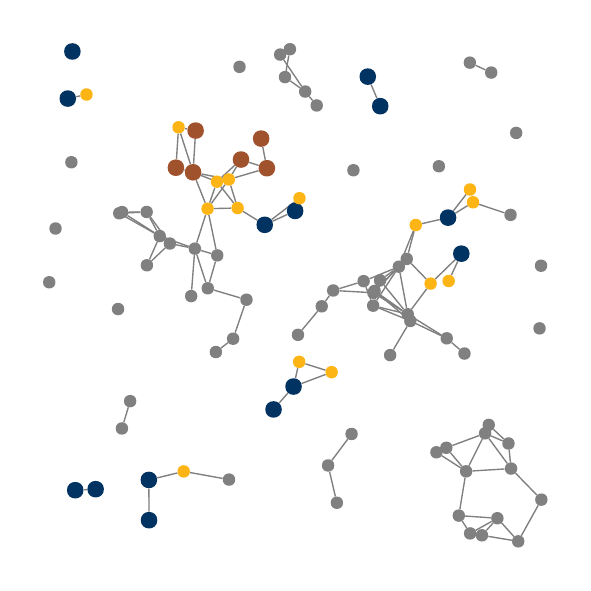}
\end{subfigure}
\hspace{-0.03\textwidth}
\begin{subfigure}{0.45\textwidth}
    \centering
        \caption*{$t=1$}
    \includegraphics[width=\linewidth]{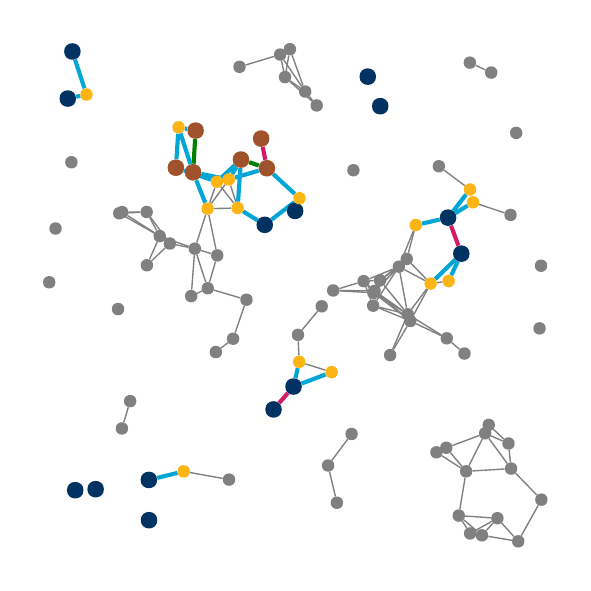}
\end{subfigure}

\vspace{-0.1cm}

\begin{subfigure}{0.45\textwidth}
    \centering
        \caption*{$t=2$}
    \includegraphics[width=\linewidth]{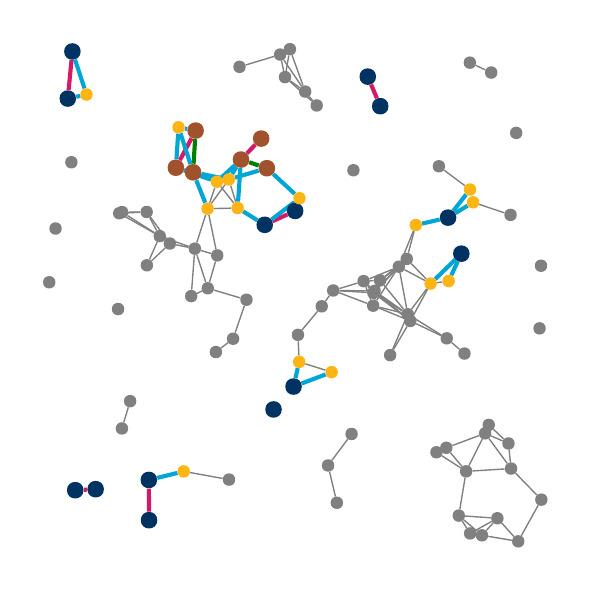}
\end{subfigure}
\hspace{-0.03\textwidth}
\begin{subfigure}{0.45\textwidth}
    \centering
        \caption*{$t=3$}
    \includegraphics[width=\linewidth]{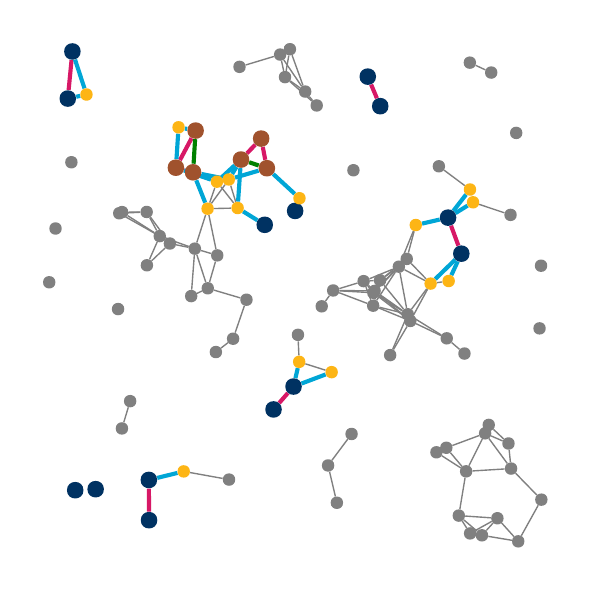}
\end{subfigure}

\vspace{-0.4cm}

\includegraphics[
    width=0.5\linewidth
]{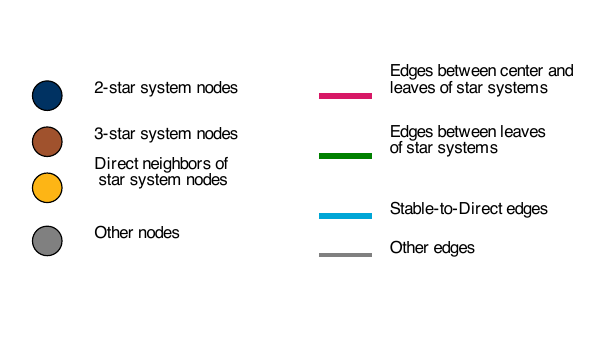} 
\end{center}
\noindent \underline{Notes:} A typical sequence of networks for the case where $N = 100$ and average degree in the initial network is 2.

\label{fig:network_evolution}

\end{figure}

\newpage 

\section{Empirical Application: Co-authorship Networks}\label{sec : empirical_application}

\cite{Ductoretal2014} studied how researchers’ co-authorship network helps predict their future research output. To illustrate our methodology in a concrete setting, we repurpose their data  -- assembled from EconLit, a bibliography of journal articles compiled by the American Economic Association -- and model the choice of co-authors directly.\footnote{The data that support the findings of this study are openly available on the \textit{Harvard Dataverse}, DOI: \url{https://doi.org/10.7910/DVN/27138}.} The \cite{Ductoretal2014} datasets spans the period 1970-1999 and is restricted to papers with at most three authors. We focus on the most recent four-year window available: $t=1996,1997,1998,1999$. For each year $t$ and each pair of authors $i$ and $j$, we set $D_{ijt}=1$ if authors $i$ and $j$ co-authored at least one paper in year $t$, and set $D_{ijt}=0$ otherwise. \\
Descriptive statistics for the co-authorship network sequence are detailed in Table \ref{tab: EmpiricalApplication_SummaryStats}. Our sample size is $N=56,156$ authors. The networks are fairly sparse with averages degrees ranging between 0.36 and 0.43 and a share of isolates consistently above $70$\%.  This structure coincides with a \emph{low} average degree random geometric graph; each of our four networks is fragmented into many small connected components.

In the data, many of the observed triangles are due to triple-authored papers (versus three separate bilateral collaborations). Whether our statistical model appropriately captures the micro-foundations of collaborator choices is doubtful. Our goal is simply to provide a worked example on a real world dataset as a ``proof-of-concept".



\begin{table}[!htb]
    \caption{Basic properties of the co-authorship network over 1996-1999}
    \label{tab: EmpiricalApplication_SummaryStats}
\centerline{
    \begin{tabular}{l|cccccc} \hline 
\cline{1-7} 
Period & \# of authors & Average degree & Density & Share Isolated & T & GC  \\ 
\cline{1-7}  
$t=0$ & 56156 & 0.36 & 6.3e-6 & 0.78 & 0.48 & 0.0010\\ 
$t=1$ & 56156 & 0.37 & 6.6e-6 & 0.77 & 0.48 & 0.0024\\ 
$t=2$ & 56156 & 0.40 & 7.1e-6 & 0.75 & 0.48 & 0.0019\\ 
$t=3$ & 56156 & 0.43 & 7.7e-6 & 0.73 & 0.50 & 0.0013\\ 
\hline 
\end{tabular} 

    }
\smallskip{}
{\small\uline{Notes:}}{\small{} 
The second column reports average degree or the average number of collaborations per author, the third column documents network density or the share of realized collaborations, the fourth column indicates the share of authors with no collaborations,
T is the global clustering coefficient or transitivity index, and GC denotes the
fraction of agents that are part of the largest component.}{\small\par}
\end{table}

We report four estimates of $\alpha$ and $\beta$. First, we fit a standard logistic regression of $D_{ijt}$ on $D_{ijt-1}$ and $R_{ijt-1}$, without controlling for fixed effects. We then report three versions of our conditional star-system likelihood estimator, based respectively on: (i) identifying $2$-star systems only, (ii) identifying $2$ and $3$-star systems, and (iii) identifying $2$, $3$, and $4$-star systems. The results are collected in Table \ref{tab: Empirical_Application_Results}. 

A first takeaway from Table \ref{tab: Empirical_Application_Results} is that our co-authorship network features a sizable number of identifying star systems, including $3798$ identifying $2$-star systems. This represents a fairly large effective sample size for our estimators and supports reliable inference in this empirical setting. Their number, however, declines rapidly with system size: from $3798$ identifying $2$-star systems to $560$ identifying $3$-star systems, and only $61$ identifying  $4$-star systems.  Consequently, incorporating stable $4$-systems in the last column of Table \ref{tab: Empirical_Application_Results} produces very modest changes in parameter estimates and  standard errors relative to solely relying on identifying 2 and 3-star systems (Column 3).


The logistic regression results suggests that sequences of the form $D_{ikt-1}=1, D_{jkt-1}=1, D_{ijt}=1$ (i.e., instances of links forming to close an open two-star from the previous period) are relatively infrequent in the data. This is indicated by the negative estimate of $\beta$. In contrast, the positive estimate of $\alpha$ indicates an overall high-level of bi-lateral link persistence. Some portion of these patterns simply reflects the weaker persistence of three-way collaborations across periods compared to two-way ones.

Columns 2 to 4 report CSS-MLE estimates. In general, these estimates suggest that the high transitivity index observed ($\sim 0.5$ in Table \ref{tab: EmpiricalApplication_SummaryStats}) is mainly driven by homophily rather than an intrinsic preference for transitive triplets. Furthermore, the coefficient in the standard logit regression is considerably larger (less negative), indicating that the logit model overstates the role of transitivity by failing to account for homophily-driven clustering.

\begin{table}[!htb]
    \caption{Conditional star system likelihood estimates for the period 1996-1999}\label{tab: Empirical_Application_Results}
\centerline{
    \begin{tabular}{c|c c c c} \hline 
& Logit & 2-star MLE & 2\&3-star MLE &  2\&3\&4-star MLE  \\ 
\hline 
$ \hat{\alpha} $ &2.798 & 0.251 & 0.357 & 0.362\\ 
 & \footnotesize{(0.096)} & \footnotesize{(0.111)} & \footnotesize{(0.092)} & \footnotesize{(0.089)} \\ 
$ \hat{\beta} $ &-0.119 & -0.656 & -0.564 & -0.511\\ 
 & \footnotesize{(0.035)} & \footnotesize{(0.312)} & \footnotesize{(0.204)} & \footnotesize{(0.198)} \\ 
\hline 
\# of authors &56156 &56156 & 56156 &56156\\ 
\# of identifying 2-star systems & &3798 & 3798 &3798\\ 
\# of identifying 3-star systems &  &  &560 &560\\ 
\# of identifying 4-star systems & &  & &61\\ 
\# of relevant star systems & &365 & 532 &555\\ 
\hline 
\end{tabular}
    }
\smallskip{}
{\small\uline{Notes:}}{\small{} The last row of the table indicates the number of identifying star systems that have a non zero contribution to the conditional star system likelihood. Asymptotic standard errors are reported in parentheses.}{\small\par}
\end{table}

\section{Conclusion}

The analysis of this paper suggests that simple models of dynamic network formation with rich heterogeneity structures can be identified from a single network sequence. Furthermore, consistent estimation, as well as large single network inference, is feasible in some settings. "Watch out, you might get what you're after." Our reliance on logit errors, as in related single agent panel settings, is consequential. At the same time, we expect aspects of our analysis to extend beyond this assumption.

Our work suggests several open questions. First, questions of (asymptotic) efficiency and information accumulation are delicate in our setting. The maximal rate-of-convergence for our conditional star-system maximum likelihood estimator is $\sqrt{N}$. Recall that we observe $\binom{N}{2}=O(N^2)$ sequences of link decisions; so a $\sqrt{N}$ rate-of-convergence is analogous to a $N^{\frac{1}{4}}$ rate in the traditional single-agent panel setting. Our examples suggest that achieving this rate depends upon features of the heterogeneity distribution and initial network configuration.

In practice researchers can count the number of identifying star-systems in their dataset. As noted early, with $\theta$ consisting of just two parameters, and the CSS likelihood being globally concave, we expect our asymptotic theory to be a good guide to finite-network properties even in settings with only a modest number of identifying star-systems. Whether such star-systems are present in real-world networks, however, is an open question. Our empirical illustration is hopeful in this regard.

Our CSS-MLE estimator utilizes only a subset of the information contained in the ``full" conditional likelihood. What is the nature of this information loss? Also of interest is understanding the differential information content of star-systems of different sizes.

It would also be of interest to study the identifiability of other models of link formation. Our focus on transitivity is inspired by empirical research on social and economic networks, but other types of preferences over (local) network structure are of interest as well (see \cite{Pelican_Graham_ReStud2026}). Other questions involve the inclusion of time-varying regressors and the identifiability of average effects (cf., \cite{dano2023transition}, \cite{dano2025binary}, \cite{gao2026identification}).


\begin{center}
    \textbf{Data Availability Statement}
\end{center}

The data used in this paper were extracted from a larger dataset assembled by \cite{Ductoretal2014} and available in the associated Harvard Dataverse respository \citep{Ductoret2014_data}. We certify that the author(s) of the manuscript have legitimate access to, and permission to use, the data used in this manuscript.

{\small
		\bibliography{references}
	}

\newpage

\appendix

{\centering \Large \bfseries Supplemental Web Appendices \par}
\vspace{1cm}
These appendices include proofs of the theorems stated in the main text, as well as statements and proofs of supplemental lemmata. All notation is as established in the main text, unless stated otherwise. Equation numbering continues in sequence with that established in the main text. 

\section{Preliminary results}\label{app: preliminary-results}

In this Appendix we prove some of the Lemmas stated in the main text. We also state (and prove as needed) several supplemental results which underlie the proofs of the main Theorems, which appear in Appendix \ref{app: main-results} below.

Throughout a $C$ denotes a generic constant, the magnitude of which may vary throughout. The abbreviations TI, CSI, MI, HI and JI respectively denote Triangle Inequality, Cauchy-Schwarz Inequality, Hoeffding, Markov Inequality, and Jensen's Inequality. The Law of Iterated Expectations is abbreviated by LIE.

\subsection{Preliminary results}\label{app: preliminary-results-appendix-only}
To prove Theorem \ref{thm: markov_draw}, stated in the main text, we utilize a textbook result on finite Markov chains (e.g., Theorem 7.10 on p. 172 of \cite{Mitzenmacher_Upfal_PC05}).
\begin{thm}
\textsc{\label{thm: stationary_distribution}(Stationary Distribution)
}Consider a finite, irreducible, and ergodic Markov chain with $L$
states and transition matrix $\mathbf{P}=\left[p_{lm}\right]_{1\leq l,m\leq L}.$
If (i) there exists a vector of non-negative numbers $\underline{\pi}=\left(\pi_{1},\ldots,\pi_{L}\right)'$
such that (i) $\sum_{l=1}^{L}\pi_{l}=1$ and (ii) $\pi_{l}p_{lm}=\pi_{m}p_{ml}$
for any pair $l,m=1,\ldots,L$, then $\underline{\pi}$ is the stationary
distribution of the Markov chain with transition matrix $\mathbf{P}$.
\end{thm}

To demonstrate consistency of the conditional star-system MLE, Theorem \ref{thm: consistency} in the main text, we invoke a variant of the standard M-Estimation consistency theorem (e.g. Theorem 2.1 in \cite{Newey_McFadden_HBE94} or Theorem 4.1 in \cite{Amemiya1985}). We present and prove the precise result we use here. 
\begin{lem}{\label{lem: consistency}}\textsc{(Consistency under Uniform Approximation and Identification)}
    Let $\Theta \subset \mathbb{R}^{\dim\left(\theta\right)}$ be compact. For each $N$, let 
$\tilde Q_N : \Theta \to \mathbb{R}$  and 
$\widehat Q_N : \Theta \to \mathbb{R}$ be two stochastic functions. 
Define the estimator
\[
\hat\theta_N \in \arg\max_{\theta \in \Theta} \hat Q_N(\theta).
\]

Assume:
\begin{enumerate}[label=\Alph*]
  \item \textbf{Uniform (in $\theta$) approximation of the target}:
  \[
  \sup_{\theta \in \Theta} 
  \big| \hat Q_N(\theta) - \tilde Q_N(\theta) \big| 
  \xrightarrow{p} 0.
  \]
  (For example, via a uniform law of large numbers or stochastic equicontinuity.)

  \item \textbf{Well-separated maximizer (i.e., identification assumption)}: 
  There exists $\theta_0 \in \Theta$ such that for every $\varepsilon > 0$,
  \[
  \Pr \left(\liminf_{N \to \infty} 
  \Big\{
  \tilde Q_N(\theta_0) - 
  \sup_{\|\theta - \theta_0\| \ge \varepsilon} \tilde Q_N(\theta)
  \Big\}
  > 0\right) = 1.
  \]
\end{enumerate}

\noindent
Then
\[
\hat\theta_N \xrightarrow{p} \theta_0.
\]
\end{lem}
\begin{proof}
Fix $\varepsilon > 0$. Under the event $\mathbf{1}\left( \liminf_{N \to \infty} 
  \Big\{
  \tilde Q_N(\theta_0) - 
  \sup_{\|\theta - \theta_0\| \ge \varepsilon} \tilde Q_N(\theta)
  \Big\}
  > 0 \right)$, for large $N$ there exists $\eta > 0$ such that
\[
\tilde Q_N(\theta_0) \ge 
\sup_{\|\theta - \theta_0\| \ge \varepsilon} \tilde Q_N(\theta) + 3\eta.
\]
If $\sup_{\Theta} |\hat Q_N - \tilde Q_N| \le \eta$, then
\[
\hat Q_N(\theta_0) 
\ge 
\sup_{\|\theta - \theta_0\| \ge \varepsilon} \hat Q_N(\theta) + \eta,
\]
so any maximizer of $\hat Q_N$ lies within $\varepsilon$ of $\theta_0$.
Since the events $\mathbf{1}\left( \liminf_{N \to \infty} 
  \Big\{
  \tilde Q_N(\theta_0) - 
  \sup_{\|\theta - \theta_0\| \ge \varepsilon} \tilde Q_N(\theta)
  \Big\}
  > 0 \right)$ and  $\mathbf{1}\left(\sup_{\Theta}|\hat Q_N - \tilde Q_N|\le \eta\right)$ occur with
probability approaching one by A and B, it follows that 
$\hat\theta_N \xrightarrow{p} \theta_0$.
\end{proof}

\subsection{Proofs of preliminary results appearing in the main text}\label{app: preliminary-results-main-text}

\noindent \textbf{Proof of Lemma \ref{lem: permutation-lemma}} \label{app: permutation-lemma_proof} \textsc{(Permutations)}

\begin{proof}
\underline{Part (A)}: Generically we have, for any $ij\in\mathcal{S}_{p}$
(with $\mathcal{S}_{p}$ an identifying $p$-star system with $i$
at its center) ,
\begin{align}
R_{ij2}-R_{ij1}= & \sum_{l=1}^{N}\left(D_{il2}D_{jl2}-D_{il1}D_{jl1}\right)\nonumber \\
= & \sum_{k\in v\left(\mathcal{S}_{p}\right)}\left(D_{ik2}D_{jk2}-D_{ik1}D_{jk1}\right)+\sum_{l\notin v\left(\mathcal{S}_{p}\right)}\left(D_{il2}D_{jl2}-D_{il1}D_{jl1}\right)\nonumber \\
= & \sum_{k\in v\left(\mathcal{S}_{p}\right)}\left(D_{ik2}D_{jk2}-D_{ik1}D_{jk1}\right)+0\nonumber \\
= & \sum_{k\in v\left(\mathcal{S}_{p}\right)}\left(D_{ik2}-D_{ik1}\right)D_{jk1},\label{eq: permutation_lemma_part_A_r1}
\end{align}
or, equivalently,
\[
R_{ij2}=R_{ij1}+\sum_{k\in v\left(\mathcal{S}_{p}\right)}\left(D_{ik2}-D_{ik1}\right)D_{jk1}.
\]
The third and fourth equalities in \eqref{eq: permutation_lemma_part_A_r1} follow, respectively from parts \textcolor{black}{(i)} and \textcolor{black}{(iii)} of Definition \ref{def: stable-neighborhood}. 
Next observe that, since $i$ is the center of $\mathcal{S}_{p},$ we
have for any $\mathbf{B}^{3}\in\mathbb{B}_{K}\left(\mathbf{D}^{3}\right)$,
either $\left(B_{ik1},B_{ik2}\right)=\left(D_{ik1},D_{ik2}\right)$
for all $k\in v\left(\mathcal{S}_{p}\right)$ or $\left(B_{ik1},B_{ik2}\right)=\left(D_{ik2},D_{ik1}\right)$
for all $k\in v\left(\mathcal{S}_{p}\right)$. The definition of $\mathbb{B}_{K}\left(\mathbf{D}^{3}\right)$ then implies:
\begin{align}
S_{ij1}= & \sum_{k\in v\left(\mathcal{S}_{p}\right)}B_{ik1}D_{jk1}+\sum_{l\notin v\left(\mathcal{S}_{p}\right)}D_{il1}D_{jl1}\nonumber \\
= & \sum_{k\in v\left(\mathcal{S}_{p}\right)}\left(B_{ik1}D_{jk1}-D_{ik1}D_{jk1}\right)+\sum_{l\notin v\left(\mathcal{S}_{p}\right)}D_{il1}D_{jl1}+\sum_{k\in v\left(\mathcal{S}_{p}\right)}D_{ik1}D_{jk1}\nonumber \\
= & \sum_{k\in v\left(\mathcal{S}_{p}\right)}\left(B_{ik1}D_{jk1}-D_{ik1}D_{jk1}\right)+\sum_{l}D_{il1}D_{jl1}\nonumber \\
= & \sum_{k\in v\left(\mathcal{S}_{p}\right)}\left(B_{ik1}-D_{ik1}\right)D_{jk1}+R_{ij1}\nonumber \\
= & \sum_{k\in v\left(\mathcal{S}_{p}\right)}\left(D_{ik2}-D_{ik1}\right)D_{jk1}-\sum_{k\in v\left(\mathcal{S}_{p}\right)}\left(D_{ik2}-B_{ik1}\right)D_{jk1}+R_{ij1}\nonumber \\
= & R_{ij2}-R_{ij1}-\sum_{k\in v\left(\mathcal{S}_{p}\right)}\left(D_{ik2}-B_{ik1}\right)D_{jk1}+R_{ij1}\nonumber \\
= & R_{ij2}-\sum_{k\in v\left(\mathcal{S}_{p}\right)}\left(D_{ik2}-B_{ik1}\right)D_{jk1}.\label{eq: permutation_lemma_part_A_r2}
\end{align}
Hence, if $\left(B_{ik1},B_{ik2}\right)=\left(D_{ik1},D_{ik2}\right)$
for all $k\in v\left(\mathcal{S}_{p}\right)$, we have, using \eqref{eq: permutation_lemma_part_A_r1}
and \eqref{eq: permutation_lemma_part_A_r2}, that $S_{ij1}=R_{ij1}$
(and, by symmetry, that $S_{ij2}=R_{ij2}$). If $\left(B_{ik1},B_{ik2}\right)=\left(D_{ik2},D_{ik1}\right)$ for
all $k\in v\left(\mathcal{S}_{p}\right)$, \eqref{eq: permutation_lemma_part_A_r2} instead gives $S_{ij1}=R_{ij2}$
(and, by symmetry, that $S_{ij2}=R_{ij1}$). 

\underline{Part (B)}: Generically we have, for any
$(j,k)\in v(\mathcal{S}_{p})\times v(\mathcal{S}_{p})$ such that $jk\in \mathcal{D}_{p}^{fd}$
\begin{align}
R_{jk2}-R_{jk1}= & \sum_{l=1}^{N}\left(D_{jl2}D_{kl2}-D_{jl1}D_{kl1}\right)\notag\\
= & \sum_{l\in v\left(\mathcal{S}_{p}\right)}\left(D_{jl2}D_{kl2}-D_{jl1}D_{kl1}\right)+\sum_{m\notin v\left(\mathcal{S}_{p}\right)}\left(D_{jm2}D_{km2}-D_{jm1}D_{km1}\right)\notag\\
= & \sum_{l\in v\left(\mathcal{S}_{p}\right)}\left(D_{jl2}D_{kl2}-D_{jl1}D_{kl1}\right)+0\notag\\
= & D_{ij2}D_{ik2}-D_{ij1}D_{ik1}+\sum_{l\in v\left(\mathcal{S}_{p}\right),l\neq i}\left(D_{jl2}D_{kl2}-D_{jl1}D_{kl1}\right)\notag\\
= & D_{ij2}D_{ik2}-D_{ij1}D_{ik1} + 0.\label{eq: permutation_lemma_part_B_r1}
\end{align}
The third and fifth equalities above follow from, respectively, parts \textcolor{black}{(i)} and \textcolor{black}{(iii)} of Definition \ref{def: stable-neighborhood}.
Next observe that, invoking parts \textcolor{black}{(i)} and \textcolor{black}{(iii)} of Definition \ref{def: stable-neighborhood}, and also using \eqref{eq: permutation_lemma_part_B_r1},
\begin{align*}
S_{jk1} = & B_{ij1}B_{ik1} + \sum_{l\neq i}D_{jl1}D_{kl1} \\
        = & B_{ij1}B_{ik1} + \sum_{l\neq i}D_{jl2}D_{kl2} \\
= & R_{jk2}-\left(D_{ij2}D_{ik2}-B_{ij1}B_{ik1}\right).
\end{align*}
Hence if $\left(B_{ij1},B_{ij2}\right)=\left(D_{ij1},D_{ij2}\right)$
and $\left(B_{ik1},B_{ik2}\right)=\left(D_{ik1},D_{ik2}\right)$,
then $S_{jk1}=R_{jk1}$ (and, by symmetry, $S_{jk2}=R_{jk2}$; the first equality above).
If, instead, $\left(B_{ij1},B_{ij2}\right)=\left(D_{ij2},D_{ij1}\right)$
and $\left(B_{ik1},B_{ik2}\right)=\left(D_{ik2},D_{ik1}\right)$,
then $S_{jk1}=R_{jk2}$ (and, by symmetry, $S_{jk2}=R_{jk1}$; the last equality above). 

\underline{Part (C)}: Recall that $j\in v\left(\mathcal{S}_{p}\right)$ and $l\notin v\left(\mathcal{S}_{p}\right)$. There are two cases to consider. In case (i), while agent $l$ is not part of $v\left(\mathcal{S}_{p}\right)$, it's in its period $1$ neighborhood (i.e., $l \in n_1\left(\mathcal{S}_{p}\right)$). In case (ii) it is not in this neighborhood (i.e., $l \not\in n_1\left(\mathcal{S}_{p}\right)$). In the first case, invoking parts (i) to (iii) of Definition \ref{def: stable-neighborhood} (such that $B_{jm1}=D_{jm1}$ for $m \neq i$ and $B_{lm1}=D_{lm1}$ for all $m=1,\ldots,N$), we have:
\begin{align*}
S_{jl1}= & B_{ij1}D_{il1} + \sum_{m\neq i}D_{jm1}D_{lm1}\\
= & B_{ij1}D_{il1}-D_{ij1}D_{il1}+R_{jl1}\\
= & \left(B_{ij1}-D_{ij1}\right)D_{il1}+R_{jl1}.
\end{align*}
Hence if $B_{ij1}=D_{ij1}$ we have that $S_{jl1}=R_{jl1}$ (and,
by symmetry, $S_{jl2}=R_{jl2}$). If, instead, $B_{ij1}=D_{ij2}$, we have (again invoking the parts (i) and (ii) of Definition \ref{def: stable-neighborhood})
\begin{align*}
S_{jl1}= & B_{ij1}D_{il1} + \sum_{m\neq i}D_{jm1}D_{lm1}\\
= & B_{ij1}D_{il1} + \sum_{m\neq i}D_{jm2}D_{lm2}\\
= & B_{ij1}D_{il1}-D_{ij2}D_{il2}+R_{jl2}\\
= & \left(B_{ij1}-D_{ij2}\right)D_{il2}+R_{jl2},
\end{align*}
such that  $S_{jl1}=R_{jl2}$ (and, by symmetry, $S_{jl2}=R_{jl1}$). 

Next consider case (ii) where $l \not\in n_1\left(\mathcal{S}_{p}\right)$. There are two sub-cases. If $l \in v\left(\mathcal{S}_l\right)$ for some other identifying star system with center $i'$ then
\begin{align*}
S_{jl1} = & B_{ij1}D_{il1} + D_{i'j1}B_{i'l1} + \sum_{m\neq i,i'}D_{jm1}D_{lm1}\\
= & B_{ij1}D_{il1} + D_{i'j1}B_{i'l1} + \sum_{m\neq i,i'}D_{jm2}D_{lm2},
\end{align*}
by part (i) of Definition \ref{def: stable-neighborhood}. Next, since  $D_{il1}=D_{il2}=D_{i'j1}=D_{i'j2}=0$ by Lemma \ref{lem: invariance-lemma}, we can use the argument from case (i) above to conclude that $S_{il1}=S_{il2}=R_{il1}=R_{il2}$ and the conclusion trivially holds.
In the second sub-case $l$ is neither in the neighborhood of $\mathcal{S}_p$ nor does it belong to another identifying star system. The link structure of agent $l$ is invariant across the elements of $\mathbf{B}^{3}\in\mathbb{B}_{K}\left(\mathbf{D}^{3}\right)$ and we can show the result by an adaptation of the argument used for case (i). 
\end{proof}

\noindent \textbf{Proof of Lemma \ref{lem: conditional-independence}} \label{app: conditional-independence_proof} \textsc{(Conditional Independence)}
\begin{proof}
For any network sequence $\tilde{\bold{d}}$, denote $Z_{p\mathbf{i}}(\tilde{\bold{d}})$   the indicator of whether $\bold i$ is stable in $\tilde{\bold{d}}$ and $\mathbf{Q}_{p\mathbf{i}}(\tilde{\bold{d}})$ the tuple $\mathbf{Q}_{p\mathbf{i}}$ defined on $\tilde{\bold{d}}$. Let $\mathbf Q_{-\mathbf Z}(\tilde{\boldsymbol d})$ collect all remaining dyads $Q_{ij}$
that do not appear in any $\mathbf Q_{p\mathbf i}(\tilde{\boldsymbol d})$. Write the joint probability mass function of the tuples $(Z_{p\mathbf{i}}\mathbf{Q}_{p\mathbf{i}})_{p\leq K, \bold i}, \mathbf{Q}_{-\bold Z}$:

\begin{align*} 
\Pr\left(\left.(Z_{p\mathbf{i}}\mathbf{Q}_{p\mathbf{i}})_{p\leq K, \mathbf{i}}, \mathbf{Q}_{-\mathbf{Z}}\right|\mathbf{D}_0,\mathbf{D}_3,\mathbf{Z}, \mathbf{A}\right)&=\sum_{\tilde {\mathbf{d}}:\; \forall p, \mathbf{i} \ : \  \mathbf{Q}_{p\mathbf{i}}(\tilde{\mathbf{d}})= \mathbf{Q}_{p\mathbf{i}}; \mathbf Q_{-\mathbf Z}(\tilde{\boldsymbol d})= \mathbf{Q}_{-\bold Z}  } \Pr\left(\tilde {\mathbf{d}}|\mathbf{D}_0,\mathbf{D}_3,\mathbf{Z}, \mathbf{A}\right).
\end{align*}
Note that for a fixed  $\mathbf{D}_0,\mathbf{D}_3,\mathbf{Z}, \mathbf{A}$, the map $\tilde{\mathbf{d}}\rightarrow \left[(Z_{p\mathbf{i}}(\tilde{\mathbf{d}}) \mathbf{Q}_{p\mathbf{i}}(\tilde{\mathbf{d}}))_{p, \mathbf{i}}, \mathbf{Q}_{-\mathbf{Z}}(\tilde{\mathbf{d}})\right]$ is a bijection, therefore:
\begin{align*}
\Pr\left((Z_{p\mathbf{i}}\mathbf{Q}_{p\mathbf{i}})_{p\leq K, \mathbf{i}}, \mathbf{Q}_{-\mathbf{Z}}|\mathbf{D}_0,\mathbf{D}_3,\mathbf{Z}, \mathbf{A}\right)&=\Pr\left({\mathbf{d}}|\mathbf{D}_0,\mathbf{D}_3,\mathbf{Z}, \mathbf{A}\right)\\
&=\prod_{p=2}^K \prod_{\mathbf{i}} \left(\frac{\exp(T_{p\mathbf{i}}'\theta)}{1+\exp(T_{p\mathbf{i}}'\theta)}\right)^{Z_{p\bold i}},
\end{align*}
by Theorem \ref{thm: condition-star-system-likelihood}. 
Let $c(\mathbf{Z},\mathbf{D}_0,\mathbf{D}_3) $ be the number of all potential values of $\bold Q_{-\mathbf{Z}}$, summing over all of the potential values of $\bold Q_{-\mathbf{Z}}$, we get:
\begin{align*}
\Pr\left((Z_{p\mathbf{i}}\mathbf{Q}_{p\mathbf{i}})_{p\leq K, \mathbf{i}}|\mathbf{D}_0,\mathbf{D}_3,\mathbf{Z}, \mathbf{A}\right)&=c(\mathbf{Z},\mathbf{D}_0,\mathbf{D}_3) \cdot \prod_{p=2}^K \prod_{\mathbf{i}} \left(\frac{\exp(T_{p\mathbf{i}}'\theta)}{1+\exp(T_{p\mathbf{i}}'\theta)}\right)^{Z_{p\mathbf{i}}},
\end{align*}
implying that $(Z_{p\mathbf{i}}\mathbf{Q}_{p\mathbf{i}})_{p\leq K, \mathbf{i}}$ are independent conditionally on $\mathbf{D}_0,\mathbf{D}_3,\mathbf{Z}, \mathbf{A} $, with a marginal:
$$\Pr\left(Z_{p\mathbf{i}}\mathbf{Q}_{p\mathbf{i}}|\mathbf{D}_0,\mathbf{D}_3,\mathbf{Z}, \mathbf{A}\right)\propto \left(\frac{\exp(T_{p\mathbf{i}}'\theta)}{1+\exp(T_{p\mathbf{i}}'\theta)}\right)^{Z_{p\mathbf{i}}}.$$
\end{proof}

\section{Main results}\label{app: main-results}
\subsection{ Proof of Proposition \ref{prop: conditioning_set}}

    Part (i) follows from inspection of \eqref{eq: likelihood_chamberlain} and \eqref{eq: conditional_likelihood_v1}. To show (ii) let $\mathbf{b}^3 \in \mathbb{B}\left(\mathbf{d}^{3}\right)$: because $\mathbf{b}_{0}=\mathbf{d}_{0}$, $\mathbf{b}_{3}=\mathbf{d}_{3}$, and $\mathbf{b}_{1}+\mathbf{b}_{2}=\mathbf{d}_{1}+\mathbf{d}_{2}$, then $\sum_{t=1}^{3}b_{ijt}=\sum_{t=1}^{3}d_{ijt}$.
    Moreover, fix a pair $ij$, $\delta \in \{0,1\}$ and $\rho \in \{0, 1,\dots, N-2\}$. Since $\mathbf{b}_{0}=\mathbf{d}_{0}$, then $1\in m_{ij}\left(\delta,\rho;\mathbf{d}^{3}\right)$ if, and only if, $1\in m_{ij}\left(\delta,\rho;\mathbf{b}^{3}\right)$. Given that $\left(s_{ij1},s_{ij2}\right)=\left(r_{ij1},r_{ij2}\right)\ \text{if}\ \left(b_{ij1},b_{ij2}\right)=\left(d_{ij1},d_{ij2}\right)\nonumber $
    and $\left(s_{ij1},s_{ij2}\right)=\left(r_{ij2},r_{ij1}\right)\ \text{if}\ \left(b_{ij1},b_{ij2}\right)=\left(d_{ij2},d_{ij1}\right)\nonumber$, then $|m_{ij}\left(\delta,\rho;\mathbf{b}^{3}\right) \cap \{2,3\} | =|m_{ij}\left(\delta,\rho;\mathbf{d}^{3}\right)\cap \{2,3\} |$, therefore $|m_{ij}\left(\delta,\rho;\mathbf{b}^{3}\right)|= |m_{ij}\left(\delta,\rho;\mathbf{d}^{3}\right)|$.\\

    Conversely, let $\mathbf{b}^3\in \mathbb{B}^*$ and note that $d_{ij0}+d_{ij1}+d_{ij2}=\sum_{\rho=0}^{N-2} m_{ij}(1, \rho; \mathbf{d}^3)=\sum_{\rho=0}^{N-2} m_{ij}(1, \rho; \mathbf{b}^3)=b_{ij0}+b_{ij1}+b_{ij2}$. Since: $\mathbf{b}_0=\mathbf{d}_0$ by assumption, we get: $d_{ij1}+d_{ij2}= b_{ij1}+b_{ij2}$. The fact that $\sum_{t=1}^{3}b_{ijt}=\sum_{t=1}^{3}d_{ijt}$ for all $i,j$ allows us to conclude that $\mathbf{d}_3=\mathbf{b}_3$.

\subsection{Proof of Theorem \ref{thm: condition-star-system-likelihood} (Conditional Star-System Likelihood)} 

In order to simplify (\ref{eq: conditional_likelihood_ss_v1}) it
is convenient to analyze its inverse, which consists of the sum of
$2^{\sum_{p=2}^K \mathbf{m}_{p,N}}$ ratios of the form 
\begin{equation}
\frac{\Pr\left(\mathbf{D}^{3}=\mathbf{b}^{3};\theta,\mathbf{A},\pi\right)}{\Pr\left(\mathbf{D}^{3}=\mathbf{d}^{3};\theta,\mathbf{A},\pi\right)},\label{eq: key_likelihood_ratio}
\end{equation}
with $\mathbf{b}^{3}\in\mathbb{B}_{K}\left(\mathbf{d}^{3}\right)$.
Note that the cardinality of the set $\mathbb{B}_{K}\left(\mathbf{d}^{3}\right)$
is $2^{\sum_{p=2}^K \mathbf{m}_{p,N}}$. 

From the dyad partition given in equation \eqref{eq: dyad_partitioning} above, we can decompose the denominator of \eqref{eq: key_likelihood_ratio} as follows:
\begin{align}
\Pr\left(\mathbf{d}^{3};\theta,\mathbf{A},\pi\right)= & \pi\left(\mathbf{d}_{0};\mathbf{A}\right)\label{eq: initial_condition}\\
 & \times\prod_{p=2}^K \prod_{ij\in\mathcal{D}_{p}}\prod_{t=1}^{3}F\left(\alpha d_{ijt-1}+\beta r_{ijt-1}+A_{ij}\right)^{d_{ij}}\notag\\
 &\times\left[1-F\left(\alpha d_{ijt-1}+\beta r_{ijt-1}+A_{ij}\right)\right]^{1-d_{ij}}\label{eq: identifying_dyads}\\
 & \times\prod_{p=2}^K \prod_{ij\in\mathcal{D}_{p}^{fd}}\prod_{t=1}^{3}F\left(\alpha d_{ijt-1}+\beta r_{ijt-1}+A_{ij}\right)^{d_{ij}}\notag\\
 &\times\left[1-F\left(\alpha d_{ijt-1}+\beta r_{ijt-1}+A_{ij}\right)\right]^{1-d_{ij}}\label{eq: other_dyads_in_star_system}\\
 & \times\prod_{p=2}^K\prod_{ij\in\mathcal{D}_{p,o}^{pd}}\prod_{t=1}^{3}F\left(\alpha d_{ijt-1}+\beta r_{ijt-1}+A_{ij}\right)^{d_{ij}}\notag\\
 &\times\left[1-F\left(\alpha d_{ijt-1}+\beta r_{ijt-1}+A_{ij}\right)\right]^{1-d_{ij}}\label{eq: other_dyads_partially_in_star_system}\\
 & \times \prod_{p=2}^K\prod_{l=2}^K\prod_{ij\in\mathcal{D}_{p,l}^{pd}}\prod_{t=1}^{3}F\left(\alpha d_{ijt-1}+\beta r_{ijt-1}+A_{ij}\right)^{d_{ij}}\notag\\
 &\times\left[1-F\left(\alpha d_{ijt-1}+\beta r_{ijt-1}+A_{ij}\right)\right]^{1-d_{ij}}\label{eq: all_other_dyads} \\
 &\times\prod_{ij\in\mathcal{D}_{K}^{od}}\prod_{t=1}^{3}F\left(\alpha d_{ijt-1}+\beta r_{ijt-1}+A_{ij}\right)^{d_{ij}}\notag\\
 &\times\left[1-F\left(\alpha d_{ijt-1}+\beta r_{ijt-1}+A_{ij}\right)\right]^{1-d_{ij}}
\end{align}
An analogous decomposition applies to the numerator of  \eqref{eq: key_likelihood_ratio}.

In order to analyze \eqref{eq: key_likelihood_ratio} it is convenient use the dyad partition given by \eqref{eq: dyad_partitioning} and also underlying the likelihood decomposition of equations \eqref{eq: initial_condition} to \eqref{eq: all_other_dyads} immediately above.

\paragraph{Initial conditions.} To evaluate (\ref{eq: key_likelihood_ratio}) we begin with the observation for any $\mathbf{b}^{3}\in\mathbb{B}_{K}\left(\mathbf{d}^{3}\right)$ we have $\mathbf{b}_{0}=\mathbf{d}_{0}$. Therefore the $\pi\left(\mathbf{b}_{0},\mathbf{a}\right)$
and $\pi\left(\mathbf{d}_{0},\mathbf{a}\right)$ terms in, respectively,
the numerator and denominator of \eqref{eq: key_likelihood_ratio} are common factors. 

\paragraph{Dyads in $\mathcal{D}_{K}^{od}$.} Next consider the contributions, to the numerator and denominator of \eqref{eq: key_likelihood_ratio}, of all those dyads where \emph{neither} agent belongs to an identifying star system. These contributions, like those of the initial condition, also coincide. This follows since swapping the edges of an identifying star system leaves the systematic utility associated with (potential) links in $\mathcal{D}_{K}^{od}$ unchanged. The transitivity term for dyad $kl\in\mathcal{D}_{K}^{od}$ equals $r_{klt}=\sum_{m} d_{kmt}d_{lmt}$; none of the links entering this term vary over $\mathbf{b}^{3}\in\mathbb{B}_{K}\left(\mathbf{d}^{3}\right)$. 

We now turn to evaluating the contributions of dyads where at least one agent belongs to an identifying star system to ratio \eqref{eq: key_likelihood_ratio}. 

\paragraph{Dyads in $\bigcup\limits_{p=2}^{K} \mathcal{D}_p$.} Let $ij\in \mathcal{S}_{p}$ be a dyad in $\mathcal{S}_{p}$, an identifying star system of size $p$. Without loss of generality, suppose that $i$ is the center node of $\mathcal{S}_{p}$, and consider the periods $t=1,2,3$ contributions of this dyad to \eqref{eq: key_likelihood_ratio}. There are two cases to consider: that where (i) the edges in $\mathcal{S}_{p}$ are swapped in $\mathbf{b}^3$ relative to their configuration in $\mathbf{d}^3$, and (ii) where they are not swapped. In the second case, the likelihood contributions of $ij$ to the numerator and denominator of \eqref{eq: key_likelihood_ratio} coincide.  

If the former case we have either \\
(a) $(d_{ij1},d_{ij2})=(1,0)$, such that $(b_{ij1},b_{ij2})=(0,1)$, or \\
(b) $(d_{ij1},d_{ij2})=(0,1)$, such that $(b_{ij1},b_{ij2})=(1,0)$. \\
In subcase (a) $ij$'s period's $t=1,2,3$ contribution to \eqref{eq: key_likelihood_ratio} equals:
\begin{align*} 
b^{10}(q_{ij};\theta)\overset{def}{\equiv} &\frac{\left[1-F\left(\alpha b_{ij0}+\beta s_{ij0}+A_{ij}\right)\right]F\left(\beta s_{ij1}+A_{ij}\right)}{F\left(\alpha d_{ij0}+\beta r_{ij0}+A_{ij}\right)\left[1-F\left(\alpha+\beta r_{ij1}+A_{ij}\right)\right]} 
\\
&\times \frac{F\left(\alpha+\beta s_{ij2}+A_{ij}\right)^{b_{ij3}}\left[1-F\left(\alpha+\beta s_{ij2}+A_{ij}\right)\right]^{1-b_{ij3}}}{F\left(\beta r_{ij2}+A_{ij}\right)^{d_{ij3}}\left[1-F\left(\beta r_{ij2}+A_{ij}\right)\right]^{1-d_{ij3}}}\\
=&\frac{\left[1-F\left(\alpha d_{ij0}+\beta r_{ij0}+A_{ij}\right)\right]F\left(\beta r_{ij2}+A_{ij}\right)}{F\left(\alpha d_{ij0}+\beta r_{ij0}+A_{ij}\right)\left[1-F\left(\alpha+\beta r_{ij1}+A_{ij}\right)\right]} \\
&\times\frac{F\left(\alpha+\beta r_{ij1}+A_{ij}\right)^{d_{ij3}}\left[1-F\left(\alpha+\beta r_{ij1}+A_{ij}\right)\right]^{1-d_{ij3}}}{F\left(\beta r_{ij2}+A_{ij}\right)^{d_{ij3}}\left[1-F\left(\beta r_{ij2}+A_{ij}\right)\right]^{1-d_{ij3}}} \\
=&\frac{\left[1-F\left(\alpha d_{ij0}+\beta r_{ij0}+A_{ij}\right)\right]F\left(\beta r_{ij2}+A_{ij}\right)^{1-d_{ij3}}F\left(\alpha+\beta r_{ij1}+A_{ij}\right)^{d_{ij3}}}{F\left(\alpha d_{ij0}+\beta r_{ij0}+A_{ij}\right)\left[1-F\left(\beta r_{ij2}+A_{ij}\right)\right]^{1-d_{ij3}}\left[1-F\left(\alpha+\beta r_{ij1}+A_{ij}\right)\right]^{d_{ij3}}} \\
=&\frac{\left[1-F\left(\alpha d_{ij0}+\beta r_{ij0}+A_{ij}\right)\right]F\left(\alpha d_{ij3}+\beta r_{ij2}+A_{ij}\right)^{1-d_{ij3}}F\left(\alpha d_{ij3}+\beta r_{ij1}+A_{ij}\right)^{d_{ij3}}}{F\left(\alpha d_{ij0}+\beta r_{ij0}+A_{ij}\right)\left[1-F\left(\alpha d_{ij3}+\beta r_{ij2}+A_{ij}\right)\right]^{1-d_{ij3}}\left[1-F\left(\alpha d_{ij3}+\beta r_{ij1}+A_{ij}\right)\right]^{d_{ij3}}}\\
=&\frac{1-F\left(\alpha d_{ij0}+\beta r_{ij0}+A_{ij}\right)}{F\left(\alpha d_{ij0}+\beta r_{ij0}+A_{ij}\right)}\frac{F\left(\alpha d_{ij3}+\beta r_{ij1}+A_{ij}\right)}{1-F\left(\alpha d_{ij3}+\beta r_{ij1}+A_{ij}\right)}\\
&\times\left(\frac{1-F\left(\beta r_{ij1}+A_{ij}\right)}{F\left(\beta r_{ij1}+A_{ij}\right)}\frac{F\left(\beta r_{ij2}+A_{ij}\right)}{1-F\left(\beta r_{ij2}+A_{ij}\right)}\right)^{1-d_{ij3}}
\end{align*}
where we use $\mathbf{b}_{0}=\mathbf{d}_{0}$, $\ \mathbf{b}_{3}=\mathbf{d}_{3}$ and Lemma \ref{lem: permutation-lemma} for the second equality and ideas from \cite{Honore_Kyriazidou_EM00} to formulate the expression to the right of the final equality.

In the logit case, this ratio simplifies to 
\begin{align*}
    b^{10}(q_{ij};\theta) = \exp(\alpha \left[d_{ij3}-d_{ij0}\right]+\beta\left[(r_{ij1}-r_{ij0})+(1-d_{ij3})(r_{ij2}-r_{ij1})\right]),
\end{align*}
Observe that this expression does not depend on $A_{ij}$.

In subcase (b), where $(d_{ij1},d_{ij2})=(0,1)$ such that $(b_{ij1},b_{ij2})=(1,0)$, we instead get a contribution of
\begin{align*} 
b^{01}(q_{ij};\theta)\overset{def}{\equiv}&\frac{F\left(\alpha b_{ij0}+\beta s_{ij0}+A_{ij}\right)\left[1-F\left(\alpha+\beta s_{ij1}+A_{ij}\right)\right]}{\left[1-F\left(\alpha d_{ij0}+\beta r_{ij0}+A_{ij}\right)\right]F\left(\beta r_{ij1}+A_{ij}\right)} \\
&\times\frac{F\left(\beta s_{ij2}+A_{ij}\right)^{b_{ij3}}\left[1-F\left(\beta s_{ij2}+A_{ij}\right)\right]^{1-b_{ij3}}}{F\left(\alpha+\beta r_{ij2}+A_{ij}\right)^{d_{ij3}}\left[1-F\left(\alpha+\beta r_{ij2}+A_{ij}\right)\right]^{1-d_{ij3}}} \\
=&\frac{F\left(\alpha d_{ij0}+\beta r_{ij0}+A_{ij}\right)\left[1-F\left(\alpha+\beta r_{ij2}+A_{ij}\right)\right]}{\left[1-F\left(\alpha d_{ij0}+\beta r_{ij0}+A_{ij}\right)\right]F\left(\beta r_{ij1}+A_{ij}\right)} \\
&\times\frac{F\left(\beta r_{ij1}+A_{ij}\right)^{d_{ij3}}\left[1-F\left(\beta r_{ij1}+A_{ij}\right)\right]^{1-d_{ij3}}}{F\left(\alpha+\beta r_{ij2}+A_{ij}\right)^{d_{ij3}}\left[1-F\left(\alpha+\beta r_{ij2}+A_{ij}\right)\right]^{1-d_{ij3}}} \\
=&\frac{F\left(\alpha d_{ij0}+\beta r_{ij0}+A_{ij}\right)\left[1-F\left(\alpha+\beta r_{ij2}+A_{ij}\right)\right]^{d_{ij3}}\left[1-F\left(\beta r_{ij1}+A_{ij}\right)\right]^{1-d_{ij3}}}{\left[1-F\left(\alpha d_{ij0}+\beta r_{ij0}+A_{ij}\right)\right]F\left(\alpha+\beta r_{ij2}+A_{ij}\right)^{d_{ij3}}F\left(\beta r_{ij1}+A_{ij}\right)^{1-d_{ij3}}} \\
=&\frac{F\left(\alpha d_{ij0}+\beta r_{ij0}+A_{ij}\right)}{1-F\left(\alpha d_{ij0}+\beta r_{ij0}+A_{ij}\right)}\frac{1-F\left(\alpha d_{ij3}+\beta r_{ij1}+A_{ij}\right)}{F\left(\alpha d_{ij3}+\beta r_{ij1}+A_{ij}\right)} \\
&\times \left(\frac{F\left(\alpha+\beta r_{ij1}+A_{ij}\right)}{1-F\left(\alpha+\beta r_{ij1}+A_{ij}\right)}\frac{1-F\left(\alpha+\beta r_{ij2}+A_{ij}\right)}{F\left(\alpha+\beta r_{ij2}+A_{ij}\right)}\right)^{d_{ij3}}.
\end{align*}
In the logit case we further get
\begin{align*}
    b^{01}(q_{ij};\theta)=\exp(-\alpha \left[d_{ij3}-d_{ij0}\right]-\beta\left[(r_{ij1}-r_{ij0})+d_{ij3}(r_{ij2}-r_{ij1})\right]),
\end{align*}
which again does not vary with $A_{ij}$.
For the analysis which follows, it is convenient -- \emph{for the logit case} -- to combine $b^{01}(q_{ij};\theta)$ and $b^{10}(q_{ij};\theta)$ into a single expression of, for example, 
\begin{align*}
    b(q_{ij};\theta)=&\exp\Bigl(-\alpha(d_{ij2}-d_{ij1})(d_{ij3}-d_{ij0}) \\
    & -\beta(d_{ij2}-d_{ij1})(r_{ij1}-r_{ij0}) \\
    & + \beta(d_{ij1}(1-d_{ij2})(1-d_{ij3})-(1-d_{ij1})d_{ij2}d_{ij3}))(r_{ij2}-r_{ij1})\Bigr).
\end{align*}
This expression can be simplified further by noting that by Definition \ref{def: identifying-star-system} (ii), since $i$ is the center of the star system, $d_{ij2}=(1-d_{ij1})$. Hence, after some manipulation we get the form for given in the statement of Theorem \ref{thm: condition-star-system-likelihood}, equation \eqref{eq: b_function_def} of the main text.

\paragraph{Dyads in sets $\mathcal{D}_{p,l}^{pd}$.}  We next turn to dyads composed of agents that belong to two different identifying star-systems. It turns out that for any $(p,l)\in\{2,\ldots,K\}^2$, the contributions of dyads $ij\in \mathcal{D}_{p,l}^{pd}$ to the numerator and denominator of \eqref{eq: key_likelihood_ratio} coincide. To see this, recall the notation $s_{ijt}=\sum_{k}b_{ikt}b_{jkt}$ for $\mathbf{b}_{t}=\left[b_{ijt}\right]_{1\leq i,j\leq N}$, and notice first that by part (A) of Lemma \ref{lem: permutation-lemma}, $d_{ij1}=d_{ij2}=0=b_{ij1}=b_{ij2}$. Hence,  the periods $t=1,2,3$ contributions to (\ref{eq: key_likelihood_ratio}) for dyad $ij\in \mathcal{D}_{p,l}^{pd}$ equal
\begin{align*} 
&\frac{\left[1-F\left(\alpha b_{ij0}+\beta s_{ij0}+A_{ij}\right)\right]\left[1-F\left(\beta s_{ij1}+A_{ij}\right)\right]F\left(\beta s_{ij2}+A_{ij}\right)^{b_{ij3}}\left[1-F\left(\beta s_{ij2}+A_{ij}\right)\right]^{1-b_{ij3}}}{\left[1-F\left(\alpha d_{ij0}+\beta r_{ij0}+A_{ij}\right)\right]\left[1-F\left(\beta r_{ij1}+A_{ij}\right)\right]F\left(\beta r_{ij2}+A_{ij}\right)^{d_{ij3}}\left[1-F\left(\beta r_{ij2}+A_{ij}\right)\right]^{1-d_{ij3}}} 
\\
&=\frac{\left[1-F\left(\beta s_{ij1}+A_{ij}\right)\right]F\left(\beta s_{ij2}+A_{ij}\right)^{d_{ij3}}\left[1-F\left(\beta s_{ij2}+A_{ij}\right)\right]^{1-d_{ij3}}}{\left[1-F\left(\beta r_{ij1}+A_{ij}\right)\right]F\left(\beta r_{ij2}+A_{ij}\right)^{d_{ij3}}\left[1-F\left(\beta r_{ij2}+A_{ij}\right)\right]^{1-d_{ij3}}} \\
&=\frac{\left[1-F\left(\beta r_{ij1}+A_{ij}\right)\right]F\left(\beta r_{ij1}+A_{ij}\right)^{d_{ij3}}\left[1-F\left(\beta r_{ij1}+A_{ij}\right)\right]^{1-d_{ij3}}}{\left[1-F\left(\beta r_{ij1}+A_{ij}\right)\right]F\left(\beta r_{ij1}+A_{ij}\right)^{d_{ij3}}\left[1-F\left(\beta r_{ij1}+A_{ij}\right)\right]^{1-d_{ij3}}} \\
&=1
\end{align*}
where the first equality leverages the fact that $\mathbf{b}_{0}=\mathbf{d}_{0}$ and $\ \mathbf{b}_{3}=\mathbf{d}_{3}$, and the second equality draws on both part (B) of Lemma \ref{lem: permutation-lemma} and Lemma \ref{lem: permutation-lemma} (C). 

\paragraph{Dyads in sets $\mathcal{D}_{p,o}^{pd}$.}  The analysis of dyads with one agent in an identifying star-system and the other not is more complicated. Fix $p\in \{2,\ldots,K\}$ and let consider the periods $t=1,2,3$ contributions to \eqref{eq: key_likelihood_ratio} for dyad $ij\in \mathcal{D}_{p,o}^{pd}$, where, without loss of generality, we let $i \in \mathcal{V}_{p}$ and $j\in \mathcal{V}^{o}$. Let $\mathcal{S}_{p}$ denote the $p$-star system to which $i$ belongs.  By neighborhood stability of $\mathcal{S}_{p}$ (Definition \ref{def: stable-neighborhood}), $d_{ij1}=d_{ij2}$, and by the definition of $\mathbb{B}_{K}\left(\mathbf{d}^{3}\right)$,  $b_{ij1}=d_{ij1}\land b_{ij2}=d_{ij2}$. Therefore there are two cases to consider. In the first, if $(d_{ij1},d_{ij2})=(1,1)$, which yields a likelihood contribution of
\begin{align*} 
&\frac{F\left(\alpha b_{ij0}+\beta s_{ij0}+A_{ij}\right)F\left(\alpha+\beta s_{ij1}+A_{ij}\right)F\left(\alpha+\beta s_{ij2}+A_{ij}\right)^{b_{ij3}}\left[1-F\left(\alpha+\beta s_{ij2}+A_{ij}\right)\right]^{1-b_{ij3}}}{F\left(\alpha d_{ij0}+\beta r_{ij0}+A_{ij}\right)F\left(\alpha+\beta r_{ij1}+A_{ij}\right)F\left(\alpha+\beta r_{ij2}+A_{ij}\right)^{d_{ij3}}\left[1-F\left(\alpha+\beta r_{ij2}+A_{ij}\right)\right]^{1-d_{ij3}}} \\
&=\frac{F\left(\alpha+\beta s_{ij1}+A_{ij}\right)F\left(\alpha+\beta s_{ij2}+A_{ij}\right)^{d_{ij3}}\left[1-F\left(\alpha+\beta s_{ij2}+A_{ij}\right)\right]^{1-d_{ij3}}}{F\left(\alpha+\beta r_{ij1}+A_{ij}\right)F\left(\alpha+\beta r_{ij2}+A_{ij}\right)^{d_{ij3}}\left[1-F\left(\alpha+\beta r_{ij2}+A_{ij}\right)\right]^{1-d_{ij3}}}
\end{align*}
where we use $\mathbf{b}_{0}=\mathbf{d}_{0}$ and $\ \mathbf{b}_{3}=\mathbf{d}_{3}$. If the links of $\mathcal{S}_{p}$ are not swapped in $\mathbf{b}^3$ relative to their configuration in $\mathbf{d}^3$, then, by Lemma \ref{lem: permutation-lemma}, $(s_{ij1},s_{ij2})=(r_{ij1},r_{ij2})$ and the ratio above equals 1. If, instead, the links of $S_{p}$ are swapped, then by Lemma \ref{lem: permutation-lemma}, the ratio above equals
\begin{align*}
    c^{11}(q_{ij};\theta)&\overset{def}{\equiv}\frac{F\left(\alpha+\beta r_{ij2}+A_{ij}\right)F\left(\alpha+\beta r_{ij1}+A_{ij}\right)^{d_{ij3}}\left[1-F\left(\alpha+\beta r_{ij1}+A_{ij}\right)\right]^{1-d_{ij3}}}{F\left(\alpha+\beta r_{ij1}+A_{ij}\right)F\left(\alpha+\beta r_{ij2}+A_{ij}\right)^{d_{ij3}}\left[1-F\left(\alpha+\beta r_{ij2}+A_{ij}\right)\right]^{1-d_{ij3}}}
\end{align*}
In the logit case this simplifies further to $c^{11}(q_{ij};\theta)=\exp\left(\beta(1-d_{ij3})(r_{ij2}-r_{ij1})\right)$. In the second case to consider, if $(d_{ij1},d_{ij2})=(0,0)$ such that the likelihood contribution equals
\begin{align*} 
&\frac{\left[1-F\left(\alpha b_{ij0}+\beta s_{ij0}+A_{ij}\right)\right]\left[1-F\left(\beta s_{ij1}+A_{ij}\right)\right]F\left(\beta s_{ij2}+A_{ij}\right)^{b_{ij3}}\left[1-F\left(\beta s_{ij2}+A_{ij}\right)\right]^{1-b_{ij3}}}{\left[1-F\left(\alpha d_{ij0}+\beta r_{ij0}+A_{ij}\right)\right]\left[1-F\left(\beta r_{ij1}+A_{ij}\right)\right]F\left(\beta r_{ij2}+A_{ij}\right)^{d_{ij3}}\left[1-F\left(\beta r_{ij2}+A_{ij}\right)\right]^{1-d_{ij3}}} \\
&= \frac{\left[1-F\left(\beta s_{ij1}+A_{ij}\right)\right]F\left(\beta s_{ij2}+A_{ij}\right)^{d_{ij3}}\left[1-F\left(\beta s_{ij2}+A_{ij}\right)\right]^{1-d_{ij3}}}{\left[1-F\left(\beta r_{ij1}+A_{ij}\right)\right]F\left(\beta r_{ij2}+A_{ij}\right)^{d_{ij3}}\left[1-F\left(\beta r_{ij2}+A_{ij}\right)\right]^{1-d_{ij3}}} 
\end{align*}
Again, if the links of $\mathcal{S}_{p}$ are not swapped in $\mathbf{b}^3$, then by Lemma \ref{lem: permutation-lemma}, $(s_{ij1},s_{ij2})=(r_{ij1},r_{ij2})$ and the ratio equals 1. Otherwise, Lemma \ref{lem: permutation-lemma} implies that it equals
\begin{align*}
c^{00}(q_{ij};\theta)&\overset{def}{\equiv} \frac{\left[1-F\left(\beta r_{ij2}+A_{ij}\right)\right]F\left(\beta r_{ij1}+A_{ij}\right)^{d_{ij3}}\left[1-F\left(\beta r_{ij1}+A_{ij}\right)\right]^{1-d_{ij3}}}{\left[1-F\left(\beta r_{ij1}+A_{ij}\right)\right]F\left(\beta r_{ij2}+A_{ij}\right)^{d_{ij3}}\left[1-F\left(\beta r_{ij2}+A_{ij}\right)\right]^{1-d_{ij3}}}. 
\end{align*}
In the logit case, this simplified to $c^{00}(q_{ij};\theta)=\exp\left(-\beta d_{ij3}(r_{ij2}-r_{ij1})\right)$. We can -- \emph{in the logit case} -- combine $c^{11}(q_{ij};\theta)$ and $c^{00}(q_{ij};\theta)$ into a single expression of
\begin{align*}
    c(q_{ij};\theta)=\exp\left(\beta (d_{ij1}-d_{ij3})(r_{ij2}-r_{ij1})\right)
\end{align*}
It is convenient for what follows to show, again for the logit case, that $c(q_{ij};\theta)=b(q_{ij};\theta)$ for dyad $ij\in \mathcal{D}_{p,o}^{pd}$. Since $d_{ij1}=d_{ij2}$ we have:
\begin{align*}
    b(q_{ij};\theta)&=\exp\Bigl(-\alpha \left(d_{ij2}-d_{ij1}\right)\left(d_{ij3}-d_{ij0}\right)-\beta \Bigl(\left[d_{ij3}-d_{ij1}\right](r_{ij2}-r_{ij0})-\left[d_{ij3}-d_{ij2}\right](r_{ij1}-r_{ij0})\Bigr)\Bigr) \\
    &=\exp\Bigl(-\beta \Bigl(\left[d_{ij3}-d_{ij1}\right](r_{ij2}-r_{ij0})-\left[d_{ij3}-d_{ij1}\right](r_{ij1}-r_{ij0})\Bigr)\Bigr) \\
    &=\exp\Bigl(-\beta \Bigl(\left[d_{ij3}-d_{ij1}\right](r_{ij2}-r_{ij1})\Bigr) \\
    &=c(q_{ij};\theta).
\end{align*}

\paragraph{Dyads in sets $\mathcal{D}_{p}^{fd}$.} All that remains to consider are the period $t=1,2,3$ contributions to \eqref{eq: key_likelihood_ratio} for those dyads where both agents belong to the same identifying star system but the dyad itself does not (i.e., $ij\in \mathcal{D}_{p}^{fd}$, with $i,j\in v(\mathcal{S}_{p})$ for some $p$-star identifying system $\mathcal{S}_{p}$). If the links of $\mathcal{S}_{p}$ are not swapped in $\mathbf{b}^3$ relative to $\mathbf{d}^3$, then, by Lemma \ref{lem: permutation-lemma}, the contributions of $ij$ to the numerator and denominator of \eqref{eq: key_likelihood_ratio} coincide. Otherwise, if the links of $\mathcal{S}_{p}$ are swapped in $\mathbf{b}^3$, there are two cases to consider: (i) $(d_{ij1},d_{ij2})=(1,1)=(b_{ij1},b_{ij2})$ and (ii) $(d_{ij1},d_{ij2})=(0,0)=(b_{ij1},b_{ij2})$. In the first case we have
\begin{align*} 
&\frac{F\left(\alpha b_{ij0}+\beta s_{ij0}+A_{ij}\right)F\left(\alpha+\beta s_{ij1}+A_{ij}\right)F\left(\alpha+\beta s_{ij2}+A_{ij}\right)^{b_{ij3}}\left[1-F\left(\alpha+\beta s_{ij2}+A_{ij}\right)\right]^{1-b_{ij3}}}{F\left(\alpha d_{ij0}+\beta r_{ij0}+A_{ij}\right)F\left(\alpha+\beta r_{ij1}+A_{ij}\right)F\left(\alpha+\beta r_{ij2}+A_{ij}\right)^{d_{ij3}}\left[1-F\left(\alpha+\beta r_{ij2}+A_{ij}\right)\right]^{1-d_{ij3}}}  \\
&=\frac{F\left(\alpha d_{ij0}+\beta r_{ij0}+A_{ij}\right)F\left(\alpha+\beta r_{ij2}+A_{ij}\right)F\left(\alpha+\beta r_{ij1}+A_{ij}\right)^{d_{ij3}}\left[1-F\left(\alpha+\beta r_{ij1}+A_{ij}\right)\right]^{1-d_{ij3}}}{F\left(\alpha d_{ij0}+\beta r_{ij0}+A_{ij}\right)F\left(\alpha+\beta r_{ij1}+A_{ij}\right)F\left(\alpha+\beta r_{ij2}+A_{ij}\right)^{d_{ij3}}\left[1-F\left(\alpha+\beta r_{ij2}+A_{ij}\right)\right]^{1-d_{ij3}}} \\
&=\frac{\left[1-F\left(\alpha+\beta r_{ij1}+A_{ij}\right)\right]^{1-d_{ij3}}F\left(\alpha+\beta r_{ij2}+A_{ij}\right)^{1-d_{ij3}}}{F\left(\alpha+\beta r_{ij1}+A_{ij}\right)^{1-d_{ij3}}\left[1-F\left(\alpha+\beta r_{ij2}+A_{ij}\right)\right]^{1-d_{ij3}}}
\end{align*}
where the first equality follows from $\mathbf{b}_{0}=\mathbf{d}_{0}$, $\mathbf{b}_{3}=\mathbf{d}_{3}$, and Lemma \ref{lem: permutation-lemma}. In the logit case, this expression evaluates to $c^{11}(q_{ij};\theta)=\exp\left(\beta(1-d_{ij3})(r_{ij2}-r_{ij1})\right)$, identical to the case where $ij\in \mathcal{D}_{p,o}^{pd}$. In the second case, we get by similar arguments,
\begin{align*} 
&\frac{\left[1-F\left(\alpha b_{ij0}+\beta s_{ij0}+A_{ij}\right)\right]\left[1-F\left(\beta s_{ij1}+A_{ij}\right)\right]F\left(\beta s_{ij2}+A_{ij}\right)^{b_{ij3}}\left[1-F\left(\beta s_{ij2}+A_{ij}\right)\right]^{1-b_{ij3}}}{\left[1-F\left(\alpha d_{ij0}+\beta r_{ij0}+A_{ij}\right)\right]\left[1-F\left(\beta r_{ij1}+A_{ij}\right)\right]F\left(\beta r_{ij2}+A_{ij}\right)^{d_{ij3}}\left[1-F\left(\beta r_{ij2}+A_{ij}\right)\right]^{1-d_{ij3}}} \\
&=\frac{\left[1-F\left(\alpha d_{ij0}+\beta r_{ij0}+A_{ij}\right)\right]\left[1-F\left(\beta r_{ij2}+A_{ij}\right)\right]F\left(\beta r_{ij1}+A_{ij}\right)^{d_{ij3}}\left[1-F\left(\beta r_{ij1}+A_{ij}\right)\right]^{1-d_{ij3}}}{\left[1-F\left(\alpha d_{ij0}+\beta r_{ij0}+A_{ij}\right)\right]\left[1-F\left(\beta r_{ij1}+A_{ij}\right)\right]F\left(\beta r_{ij2}+A_{ij}\right)^{d_{ij3}}\left[1-F\left(\beta r_{ij2}+A_{ij}\right)\right]^{1-d_{ij3}}} \\
&=\frac{F\left(\beta r_{ij1}+A_{ij}\right)^{d_{ij3}}\left[1-F\left(\beta r_{ij2}+A_{ij}\right)\right]^{d_{ij3}}}{\left[1-F\left(\beta r_{ij1}+A_{ij}\right)\right]^{d_{ij3}}F\left(\beta r_{ij2}+A_{ij}\right)^{d_{ij3}}}. 
\end{align*}
In the logit case this ratio simplifies to $c^{00}(q_{ij};\theta)=\exp\left(-\beta d_{ij3}(r_{ij2}-r_{ij1})\right)$. As in the \emph{partial dyad} case, we can combine $c^{00}(q_{ij};\theta)$ and $c^{11}(q_{ij};\theta)$ into a generic expression of $c(q_{ij};\theta)=\exp\left(\beta (d_{ij1}-d_{ij3})(r_{ij2}-r_{ij1})\right)$. As in our analysis for dyads in $\mathcal{D}_{p,o}^{pd}$, since $d_{ij1}=d_{ij2}$, we also obtain the equivalence $c(q_{ij};\theta)=b(q_{ij};\theta)$ for dyads in $\mathcal{D}_{p}^{fd}$. \\

The analysis above indicates that the periods $t=1,2,3$ contributions of all dyads either coincide or, in the logit case, equal $b(q_{ij};\theta)$. The precise contribution of $ij$ depends upon which dyad set it belongs to and whether the edges of various identifying star systems are swapped in $\mathbf{b}^3$ relative to their configuration in $\mathbf{d}^3$. 

Let $v\left(\mathcal{S}_{p\mathbf{i}},\mathbf{b}^3\right)$ equal $1$ if the edges of identifying star system $\mathcal{S}_{p\mathbf{i}}$ are swapped and zero otherwise. The analysis above allows us to re-write \eqref{eq: key_likelihood_ratio} as.
\begin{align*}
\frac{\Pr\left(\mathbf{D}^{3}=\mathbf{b}^{3};\theta,\mathbf{A},\pi\right)}{\Pr\left(\mathbf{D}^{3}=\mathbf{d}^{3};\theta,\mathbf{A},\pi\right)} = \prod_{p=2}^K \prod_{\mathbf{i}=\mathbf{1}}^{\mathbf{m}_{p,N}} \left[\left(1-v\left(\mathcal{S}_{p\mathbf{i}},\mathbf{b}^3\right)\right) +  v\left(\mathcal{S}_{p\mathbf{i}},\mathbf{b}^3\right)\prod_{ij\in \bar{S}_{p\mathbf{i}}} b(q_{ij};\theta)\right],
\end{align*}
where we recall the definition $\bar{\mathcal{S}}_{p\mathbf{i}}=\mathcal{S}_{p\mathbf{i}}\cup \mathcal{S}_{p\mathbf{i}}^{fd} \cup \mathcal{S}_{p\mathbf{i},o}^{pd}$ with $\mathcal{S}_{p\mathbf{i}}^{fd}=\{ij \in \mathcal{D}_{p}^{fd}|  (i,j)\in v(\mathcal{S}_{p\mathbf{i}})\times v(\mathcal{S}_{p\mathbf{i}})\}$ and $\mathcal{S}_{p\mathbf{i},o}^{pd}=\{ij \in \mathcal{D}_{p,o}^{pd}|  (i,j)\in v(\mathcal{S}_{p\mathbf{i}})\times \mathcal{V}^{o}\}$.

Summing \eqref{eq: key_likelihood_ratio} over all $\mathbf{b}^{3}\in\mathbb{B}_{K}\left(\mathbf{d}^{3}\right)$ thus yields
\begin{align*}          
    \frac{\sum_{\mathbf{b}^{3}\in\mathbb{B}_{K}\left(\mathbf{d}^{3}\right)}\Pr\left(\mathbf{D}^{3}=\mathbf{b}^{3};\theta,\mathbf{A},\pi\right)}{\Pr\left(\mathbf{D}^{3}=\mathbf{d}^{3};\theta,\mathbf{A},\pi\right)}&=\prod_{p=2}^K \prod_{\mathbf{i}=\mathbf{1}}^{\mathbf{m}_{p,N}} \left[1+\prod_{ij\in \bar{S}_{p\mathbf{i}}} b(q_{ij};\theta)\right].
\end{align*}
Note that the product to the right of the equality above evaluates to a sum of $2^{\sum_{p=2}^K \mathbf{m}_{p,N}}$ terms,
one for each element of $\mathbb{B}_{K}\left(\mathbf{d}^{3}\right)$, as required. Each identifying star system $\mathbf{i}=\mathbf{1},\ldots,\mathbf{m}_{p,N}$ (for $p=2,\ldots,K$) either contributes $1$ (when not swapped in $\mathbf{b}^3$) or $\prod_{ij\in \bar{S}_{p\mathbf{i}}} b(q_{ij};\theta)$ (when swapped in $\mathbf{b}^3$) to the product which forms each summand to the left of the equality.

Inverting yields \eqref{eq: conditional_likelihood_ss_v2} as claimed.

\subsection{Proof of Theorem \ref{thm: markov_draw} (Null Sampler)}\label{app: markov_draw_proof}
Denote the state transition graph of the Markov Chain defined by Algorithm
\ref{alg: markov_draw} by $\Phi=\left(\mathcal{V}_{\phi},\mathcal{A}_{\phi}\right)$.
The vertex set, $\mathcal{V}_{\phi}$, coincides with the set of all
network sequences in $\mathbb{B}_{K}\left(\mathbf{d}^{3}\right)$. We use $v_{\phi}$ to denote a generic element of the vertex set. We first characterize the arc set of the state transition graph, allowing us to verify irreducibility and ergodicity of the Markov chain. The main result then follows by checking conditions (i) and (ii) of Theorem \ref{thm: stationary_distribution}.

The arc set, $\mathcal{A}_{\phi}$, is defined as follows.
\begin{enumerate}
\item For each vertex $v_{\phi}\in\mathcal{V}_{\phi}$ there exists the
self-loop $\left(v_{\phi},v_{\phi}\right)$ with probability weight
$\rho$ (see Step 2 of Algorithm \ref{alg: markov_draw}).
\item We say that $v_{\phi}'$ is a \textit{neighbor} of $v_{\phi}$ (with
$v_{\phi}\neq v_{\phi}'$) if it is possible to move (from network
sequence) $v_{\phi}$ to $v_{\phi}'$ in a single step of our algorithm
(i.e., by swapping the edges of a single star system). Let $\mathcal{N}\left(v_{\phi}\right)$
denote the set of vertices adjacent to $v_{\phi}$.
For each such neighbor, the arc $\left(v_{\phi},v_{\phi}'\right)$
is present in $\mathcal{A}_{\phi}$. Note $\left|\mathcal{N}\left(v_{\phi}\right)\right|=\mathbf{m}_{2:K,N}$
for all $v_{\phi}\in\mathcal{V}_{\phi}$. This follows since Algorithm
\ref{alg: markov_draw}, in step $3$, selects exactly one out of the
$\mathbf{m}_{2:K,N}=\sum_{p=2}^{K}\mathbf{m}_{p,N}$ stable star
systems uniformly at random. Let $\tilde{\mathbf{b}}^{3}$ be the network
sequence associated with vertex $v_{\phi}'$ and $\mathbf{b}^{3}$
the one asssociated with $v_{\phi}$. The probability weight on arc  $\left(v_{\phi},v_{\phi}'\right)$ equals
\begin{equation}
   p\left(\mathbf{b}^{3},\tilde{\mathbf{b}}^{3}\right)=\frac{1-\rho}{\mathbf{m}_{2:K,N}}\min\left\{ 1,\frac{\ell^{css}\left(\tilde{\mathbf{b}}^{3};\theta_{0}\right)}{\ell^{css}\left(\mathbf{b}^{3};\theta_{0}\right)}\right\}.\label{eq: arc_prob}
\end{equation}
\item Finally there exists an additional self loop $\left(v_{\phi},v_{\phi}\right)$
if all the arcs leaving $v_{\phi},$ defined in bullet points 1 and 2 above do not
have probability weights which sum to $1$. The probability attached
to this additional self-loop is one minus the probability of all other arcs with
source (tail) $v_{\phi}$.
\end{enumerate}
Let $v_{\phi}$ be the (unique) vertex associated with network sequence
$\mathbf{b}^{3}\in\mathbb{B}_{K}\left(\mathbf{d}^{3}\right)$ and $v_{\phi}'$
the one associated with $\tilde{\mathbf{b}}^{3}\in\mathbb{B}_{K}\left(\mathbf{d}^{3}\right)$.
Let $\delta\left(v_{\phi}\right)=\ell^{css}\left(\mathbf{b}^{3};\theta_{0}\right)$
and $\delta\left(v_{\phi}'\right)=\ell^{css}\left(\mathbf{\tilde{b}}^{3};\theta_{0}\right)$
be the associated conditional star system likelihood probabilities.
Let $L=2^{\mathbf{m}_{2:K,N}}$ denote the cardinality of
the set $\mathbb{B}_{K}\left(\mathbf{d}^{3}\right)$ and hence the
number of vertices in the state transition graph. The $L\times L$
transition matrix of edge weights for this graph, $\mathbf{P}=\left[p_{v_{\phi},v_{\phi}'}\right]_{v_{\phi},v_{\phi}'\in\mathcal{V}_{\phi}}$,
has elements
\begin{equation}
p_{v_{\phi},v_{\phi}'}=\left\{ \begin{array}{ll}
\frac{1-\rho}{\mathbf{m}_{2:K,N}}\min\left\{ 1,\frac{\delta\left(v_{\phi}'\right)}{\delta\left(v_{\phi}\right)}\right\}  & \text{if } v_{\phi}\neq v_{\phi}' \text{ and } v_{\phi}'\in\mathcal{N}\left(v_{\phi}\right) \\
0 & \text{if } v_{\phi}\neq v_{\phi}' \text{ and } v_{\phi}'\notin\mathcal{N}\left(v_{\phi}\right)\\
1-\sum_{v_{\phi}'\neq v_{\phi}}p_{v_{\phi},v_{\phi}'} & \text{if } v_{\phi} = v_{\phi}'
\end{array}\right..\label{eq: transition_probabilities}
\end{equation}

Next observe that the state transition graph is strongly connected.
Studying equation \eqref{eq: arc_prob} we see that $p\left(\mathbf{b}^{3},\tilde{\mathbf{b}}^{3}\right)$ $>0$
for all $\mathbf{b}^{3},\tilde{\mathbf{b}}^{3}\in\mathbb{B}_{K}\left(\mathbf{d}^{3}\right)$,
and hence that the weight on arc $\left(v_{\phi},v_{\phi}'\right)$ is strictly positive
for all $v_{\phi}'\in\mathcal{N}\left(v_{\phi}\right).$ Similarly
for any such $v_{\phi}'$ we have $v_{\phi}\in\mathcal{N}\left(v_{\phi}'\right)$
and hence $\left(v_{\phi'},v_{\phi}\right)$ also endowed with positive
weight. Since any $\tilde{\mathbf{b}}^{3}\in\mathbb{B}_{K}\left(\mathbf{d}^{3}\right)$
can be constructed by swapping the edges of a finite number of stable
star systems starting from any other $\mathbf{b}^{3}\in\mathbb{B}_{K}\left(\mathbf{d}^{3}\right)$,
there exists a directed path, with positive probability weight, between
any two vertices in $\Phi$. The state transition graph is strongly
connected and hence the Markov chain defined by Algorithm \ref{alg: markov_draw}
is irreducible. Due to the presence of self-loops, the state transition
graph is not bipartite and therefore aperiodic. Since the chain is
finite, irreducible, and aperiodic it is ergodic with a unique stationary
distribution. We can use Theorem \ref{thm: stationary_distribution}
to verify that this stationary distribution coincides with the claimed
(null) sampling distribution. 

To use the Theorem we associate $\pi_{l}=\delta\left(v_{\phi}\right),$
$\pi_{m}=\delta\left(v_{\phi}'\right)$, $p_{lm}=\frac{1-\rho}{\mathbf{m}_{2:K,N}}\min\left\{ 1,\delta\left(v_{\phi}'\right)/\delta\left(v_{\phi}\right)\right\} $
and $p_{ml}=\frac{1-\rho}{\mathbf{m}_{2:K,N}}\min\left\{ 1,\delta\left(v_{\phi}\right)/\delta\left(v_{\phi}'\right)\right\}$.
Condition (i) of Theorem \ref{thm: stationary_distribution} is satisfied
since $\ell^{css}\left(\cdot;\theta_{0}\right)$ is a probability mass function
with support $\mathbb{B}_{K}\left(\mathbf{d}^{3}\right)$ and hence
$\sum_{v_{\phi}\in\mathcal{V}_{\phi}}\delta\left(v_{\phi}\right)=1$.
Requirement (ii) is the so called \emph{detailed
balanced} condition.

To verify this condition let $v_{\phi}$ and $v_{\phi}'$ be two vertices
in $\mathcal{V}_{\phi}$ labeled, without loss of generality, such
that $\delta\left(v_{\phi}\right)\geq\delta\left(v_{\phi}'\right)$.
If $v_{\phi}'\notin\mathcal{N}\left(v_{\phi}\right)$ or $v_{\phi}'=v_{\phi}$,
then condition (ii) of Theorem \ref{thm: stationary_distribution}
holds trivially, otherwise for $v_{\phi}'\in\mathcal{N}\left(v_{\phi}\right)$
\begin{align*}
\pi_{l}p_{lm}= & \delta\left(v_{\phi}\right)\frac{1-\rho}{\mathbf{m}_{2:K,N}}\min\left\{ 1,\frac{\delta\left(v_{\phi}'\right)}{\delta\left(v_{\phi}\right)}\right\} \\
= & \delta\left(v_{\phi}\right)\frac{1-\rho}{\mathbf{m}_{2:K,N}}\frac{\delta\left(v_{\phi}'\right)}{\delta\left(v_{\phi}\right)}\\
= & \delta\left(v_{\phi}'\right)\frac{1-\rho}{\mathbf{m}_{2:K,N}}\\
= & \delta\left(v_{\phi}'\right)\frac{1-\rho}{\mathbf{m}_{2:K,N}}\min\left\{ 1,\frac{\delta\left(v_{\phi}\right)}{\delta\left(v_{\phi}'\right)}\right\} \\
= & \pi_{m}p_{ml}
\end{align*}
as required. The second equality follows from $\min\left\{ 1,\delta\left(v_{\phi}'\right)/\delta\left(v_{\phi}\right)\right\} =\delta\left(v_{\phi}'\right)/\delta\left(v_{\phi}\right)$
and the fourth from $\min\left\{ 1,\delta\left(v_{\phi}\right)/\delta\left(v_{\phi}'\right)\right\} =1$
(both implications of $\delta\left(v_{\phi}\right)\geq\delta\left(v_{\phi}'\right)$).
This gives a stationary distribution for our Markov chain of
\begin{align*}
\underline{\pi}= & \left(\delta\left(v_{\phi1}\right),\ldots,\delta\left(v_{\phi L}\right)\right)'\\
= & \left(\ell^{css}\left(\mathbf{b}_{1}^{3};\theta_{0}\right),\ldots,\ell^{css}\left(\mathbf{b}_{L}^{3};\theta_{0}\right)\right)'
\end{align*}
with $\mathbb{B}_{K}\left(\mathbf{d}^{3}\right)=\left\{ \mathbf{b}_{1}^{3},\ldots,\mathbf{b}_{L}^{3}\right\} $
and $\mathcal{V}_{\phi}=\left\{ v_{\phi1},\ldots,v_{\phi L}\right\}$  (conformably arranged).

\subsection{Proof of Theorem \ref{thm: consistency} (Consistency)}\label{app: consistency_proof}
Define:
\begin{align*}
    \tilde Q_N(\theta):&=\mathbb{E}\left[\left.\hat{Q}_N\left(\theta\right)\right|\mathbf{D}_0,\mathbf{D}_3,\mathbf{Z},\mathbf{A}\right]\\
    &=\frac{1}{\mathbb{E}\left[\sum_{p,\mathbf{i}}  Z_{p\mathbf{i}}  \right]} \sum_{p=2}^K \sum_{\mathbf{i}=\{i_1, i_2,\dots, i_p\}}Z_{p\mathbf{i}} \cdot  \mathbb{E}\left[ \left. \log\left(\frac{\exp(T_{p\mathbf{i}}'\theta)}{1+\exp(T_{p\mathbf{i}}'\theta)} \right) \right|\mathbf{D}_0,\mathbf{D}_3,\mathbf{Z},\mathbf{A}\right],
\end{align*}
and recall the definition of $\hat Q_N(\theta)$ given in the main text. The proof involves verifying that the conditions of Lemma \ref{lem: consistency} in Appendix \ref{app: preliminary-results} are satisfied whenever Assumptions \ref{ass: parameter_space} to \ref{ass: identifying_events} hold. This involves the following steps:
\begin{enumerate}
        \item Show uniform convergence to the target: $\sup_{\theta \in \Theta} 
  \big| \hat Q_N(\theta) - \tilde Q_N(\theta) \big| 
  \xrightarrow{p} 0.$ This involves two sub-steps:
  \begin{enumerate}
      \item Demonstrating point-wise convergence, for all $\theta \in \Theta$: $\widehat Q_N(\theta) - \tilde Q_N(\theta) \overset{p}{\rightarrow} 0$
      \item Inferring uniform convergence through a stochastic equicontinuity argument (following \cite{Newey1991}).
      
  \end{enumerate}
  \item Show that $\theta_0$ is a well separated maximizer of $\tilde Q_N$, by establishing a lower bound on the hessian of $\tilde Q_N$.
\end{enumerate}

\noindent \textbf{Step 1.a:}
Fix $\theta$; by conditional independence (Lemma \ref{lem: conditional-independence}) we have
\begin{multline*}
    \mathbb{V}(\left.\hat{Q}_N(\theta)\right|\mathbf{D}_0,\mathbf{D}_3,\mathbf{Z},\mathbf{A} ) \\
    = \frac{1}{\left(\mathbb{E}\left[\sum_{p,\mathbf{i}}  Z_{p\mathbf{i}}  \right]\right)^2}\sum_{p=2}^K \sum_{\mathbf{i}=\{i_1, i_2,\dots, i_p\}}Z_{p\mathbf{i}} \cdot\mathbb{V}\left(\left.\ln\left(\frac{\exp(T_{p\mathbf{i}}'\theta)}{1+\exp(T_{p\mathbf{i}}'\theta)}  \right) \right|\mathbf{D}_0,\mathbf{D}_3,\mathbf{Z},\mathbf{A} \right).
\end{multline*}
Next observe that for all $x \in \mathbb{R}^1$:
\begin{equation}\label{softmaxBound}
\min \{0, x \} - \ln(2) \leq \ln\left(\frac{\exp(x)}{1+\exp(x)}  \right)\leq \min \{0, x \},
\end{equation}
and hence
\begin{equation*}
\mathbb{V}\left(\left.\ln\left(\frac{\exp(T_{p\mathbf{i}}'\theta)}{1+\exp(T_{p\mathbf{i}}'\theta)}  \right) \right|\mathbf{D}_0,\mathbf{D}_3,\mathbf{Z},\mathbf{A} \right) \leq \mathbb{E}\left[(|T_{p\mathbf{i}}'\theta|+\ln(2))^2\right].
\end{equation*}
Therefore, by Assumptions \ref{ass: parameter_space} and  \ref{ass: boundedness}:
\begin{multline*}
    \mathbb{V}\left(\left.\hat{Q}_N(\theta)\right|\mathbf{D}_0,\mathbf{D}_3,\mathbf{Z},\mathbf{A} \right) \leq \frac{1}{\left(\mathbb{E}\left[\sum_{\mathbf{i},p}  Z_{\mathbf{i},p}  \right]\right)^2}\sum_{p=2}^K \sum_{\mathbf{i}=\{i_1, i_2,\dots, i_p\}} Z_{p\mathbf{i}} \cdot \mathbb{E}\left[(|T_{p\mathbf{i}}'\theta|+\ln(2))^2\right],
\end{multline*}
and consequently
\begin{equation*}
\mathbb{E} \left[ \mathbb{V}\left(\left.\hat{Q}_N(\theta)\right|\mathbf{D}_0,\mathbf{D}_3,\mathbf{Z},\mathbf{A} \right)\right] \leq \frac{C}{\mathbb{E}\left[\sum_{p,\mathbf{i}}  Z_{p\mathbf{i}}  \right]}
\end{equation*}
for some constant $C$ (by Assumption \ref{ass: parameter_space} and  \ref{ass: boundedness}). 

Assumption \ref{ass: identifying_events}, and non-negativity of variances, then yields $\mathbb{V}\left(\left.\hat{Q}_N(\theta)\right|\mathbf{D}_0,\mathbf{D}_3,\mathbf{Z},\mathbf{A} \right) \overset{p}{\rightarrow} 0$. From this we conclude that $\hat{Q}_N(\theta) - \mathbb{E}\left[\left.\hat{Q}_N(\theta)\right|\mathbf{D}_0,\mathbf{D}_3,\mathbf{Z},\mathbf{A} \right]\overset{p}{\rightarrow} 0$ (i.e., $\widehat Q_N(\theta) -\tilde Q_N (\theta) \overset{p}{\rightarrow} 0 $ for any $\theta$). This completes Step 1.a.\\

\textbf{Step 1.b:} We demonstrate uniform convergence by verifying the conditions of Corollary 2.2 in \cite{Newey1991} for $\Delta Q_N:=\widehat Q_N - \tilde Q_N$. Fix $\theta, \tilde \theta \in \Theta$, by Assumption \ref{ass: parameter_space} and the mean value theorem, there exists some $\bar \theta\in \left(\theta, \tilde \theta\right)$ such that
\begin{align*}
    \Delta Q_N(\theta)- \Delta Q_N(\tilde \theta) =& \frac{\partial \Delta Q_N}{\partial \theta}(\bar \theta) (\theta-\tilde \theta)\\
    =& \frac{1}{\mathbb{E}\left[\sum_{p,\mathbf{i}}  Z_{p\mathbf{i}} \right]} \\
    &\quad \times\sum_{p, \mathbf{i}} Z_{p\mathbf{i}} \cdot \left( (1-\Lambda(T_{p \mathbf{i}}'\bar \theta)) T_{p \mathbf{i}}' - \mathbb{E}\left[\left.(1-\Lambda(T_{p \mathbf{i}}'\bar \theta)) T_{ p\mathbf{i}}'\right|\mathbf{D}_0,\mathbf{D}_3,\mathbf{Z},\mathbf{A}\right]  \right)\\
    &\quad  \quad \times (\theta -\tilde \theta),
\end{align*}
recalling that $\Lambda\left(x\right)\overset{def}{\equiv}\frac{\exp(x)}{1+\exp(x)}.$ By the TI and CSI we have
\begin{align*}
    |\Delta Q_N(\theta)- \Delta Q_N(\tilde \theta)|&\leq \frac{1}{\mathbb{E}\left[\sum_{p,\mathbf{i}}  Z_{p\mathbf{i}}  \right]}\sum_{p, \mathbf{i}} Z_{p \mathbf{i}} \cdot \left(  \|T_{p\mathbf{i}}\| + \mathbb{E}\left[\left. \|T_{p\mathbf{i}}\| \right|\mathbf{D}_0,\mathbf{D}_3,\mathbf{Z},\mathbf{A}\right]  \right) \cdot \|\theta -\tilde \theta\|.
\end{align*}
To apply Corollary 2.2 of \cite{Newey1991}, it suffices to show that $$B_N:= \frac{1}{\mathbb{E}\left[\sum_{p,\mathbf{i}}  Z_{p\mathbf{i}}  \right]}\sum_{p, \mathbf{i}} Z_{p\mathbf{i}} \left(  \|T_{p\mathbf{i}}\|+ \mathbb{E}\left[\left.\|T_{p\mathbf{i}}\| \right | \mathbf{D}_0,\mathbf{D}_3,\mathbf{Z},\mathbf{A}\right]  \right) =O_p(1)$$
which is established by noting that $\mathbb{E}\left[B_N\right]$ is bounded uniformly in $N$ by Assumption \ref{ass: boundedness}.\footnote{Note that Assumption \ref{ass: boundedness} is expressed in terms of moments of $T_{pi}$ conditional on $(\mathbf{D}_0,\mathbf{D}_3,\mathbf{Z})$, not $(\mathbf{D}_0,\mathbf{D}_3,\mathbf{Z},\mathbf{A})$. A consequence of Lemma \ref{lem: conditional-independence} is that the distribution of $T_{p\mathbf{i}}$ conditional on $\mathbf{D}_0,\mathbf{D}_3,\mathbf{Z},\mathbf{A}$ is independent of $\mathbf{A}$. Hence, the moments of ${T}_{p\mathbf{i}}$ under both sets of conditioning variables are equal.} We therefore conclude that
\[
\sup_{\theta \in \Theta} 
  \big| \widehat Q_N(\theta) - \tilde Q_N(\theta) \big| 
  \xrightarrow{p} 0
\]
by applying Corollary 2.2 in \cite{Newey1991}, where we take Newey's $\hat{Q}_N$ to be $\Delta Q_N$ in this proof and Newey's $\bar Q_N$ to be the constant $0$ function. This completes Step 1.b.\\

\textbf{Step 2:} First, we show that $\tilde Q_N$ is maximized at $\theta_0$ for all $N$. The conditional likelihood of  $\mathbf{d}^3$ given $\mathbf{D}_0,\mathbf{D}_3,\mathbf{Z}_p,\mathbf{A}$ is\footnote{Let  $\mathbb{B}^+_{p}\left(\mathbf{d}^{3}\right)$ denote the set of network sequences reachable by swapping the edges of identifying star systems of \emph{size $p$ alone}. The information content of the event $\mathbf{D}^3 \in \mathbb{B}^+_{p}\left(\mathbf{d}^{3}\right)$, and $\mathbf{Z}_p$ taking a particular configuration, coincides.}
\[
\Pr(\left.\mathbf{d}^3\right|\mathbf{D}_0,\mathbf{D}_3,\mathbf{Z}_p,\mathbf{A})=\prod_{\mathbf{i}} \left(\frac{\exp(T_{p\mathbf{i}}'\theta_0)}{1+\exp(T_{p\mathbf{i}}'\theta_0)} \right)^{Z_{p\mathbf{i}}}
\]
therefore, for almost any $\mathbf{D}_0,\mathbf{D}_3,\mathbf{Z}_p,\mathbf{A}$, the function $\theta\rightarrow \mathbb{E}\left[\left.\ln \left( \prod_{\mathbf{i}} \left(\frac{\exp(T_{p\mathbf{i}}'\theta)}{1+\exp(T_{p\mathbf{i}}'\theta)} \right)^{Z_{p\mathbf{i}}}\right)\right| \mathbf{D}_0,\mathbf{D}_3,\mathbf{Z}_p,\mathbf{A} \right]$ is maximized at $\theta_0$.

Fix $\epsilon>0$. Concavity of $\tilde Q_N$ in $\theta$ and being maximized at $\theta_0$ for every $N$, implies that $\sup_{\|\theta - \theta_0\| \ge \varepsilon} \tilde Q_N(\theta)$ is achieved at some $\theta_{\epsilon}$ such that $||\theta_0-\theta_\epsilon||=\epsilon$. The multivariate mean value theorem (and by Assumption \ref{ass: parameter_space}) gives, for some $\bar{\theta}_\epsilon\in (\theta, \theta_\epsilon)$ and defining $\Delta \theta:= \theta_\epsilon-\theta_0$:
\begin{align*}
    \tilde Q_N(\theta_\epsilon)-\tilde Q_N(\theta_0) = &\frac{\partial \tilde Q_N }{\partial \theta} (\theta_0) \Delta \theta+\Delta \theta' \frac{\partial^2 \tilde Q_N }{\partial \theta^2} (\bar{\theta}_\epsilon) \Delta \theta\\
    =& 0 +\Delta \theta' \frac{\partial^2 \tilde Q_N }{\partial \theta^2} (\bar{\theta}_\epsilon) \Delta \theta\\
    =&-\frac{1}{\mathbb{E}\left(\sum_{p\mathbf{i}}Z_{p\mathbf{i}}  \right)} \\
    & \times\sum_{p,\mathbf{i}} Z_{p\mathbf{i}} \cdot  \Delta \theta'\mathbb{E}\left[\left. \Lambda(T_{p\mathbf{i}}' \bar{\theta}_\epsilon)(1-\Lambda(T_{p\mathbf{i} }'\bar{\theta}_\epsilon))T_{p\mathbf{i} }T_{p\mathbf{i} }'\right|\mathbf{D}_0,\mathbf{D}_3,\mathbf{Z},\mathbf{A} \right]\Delta \theta.
\end{align*}

\noindent Fix some constant $a>0$; for any  $\mathbf{i}$ and $p$ we have, using the MI and CSI,
\begin{align*}
\Delta\theta'\mathbb{E}\left[\left.\Lambda(T_{p\mathbf{i}}'\bar{\theta}_{\epsilon})(1-\Lambda(T_{p\mathbf{i}}'\bar{\theta}_{\epsilon}))T_{p\mathbf{i}}T_{p\mathbf{i}}'\right|\mathbf{D}_{0},\mathbf{D}_{3},\mathbf{Z},\mathbf{A}\right]\Delta\theta\geq & \Lambda(a)(1-\Lambda(a))\\
 & \times\mathbb{E}\left[\left.\Delta\theta'T_{p\mathbf{i}}T_{p\mathbf{i}}'\Delta\theta\mathbf{1}(\left|T_{p\mathbf{i}}'\bar{\theta}_{\epsilon}\right|\leq a)\right|\mathbf{D}_{0},\mathbf{D}_{3},\mathbf{Z},\mathbf{A}\right]\\
\geq & \Lambda(a)(1-\Lambda(a))\left\{ \mathbb{E}\left[\left.\Delta\theta'T_{p\mathbf{i}}T_{p\mathbf{i}}'\Delta\theta\right|\mathbf{D}_{0},\mathbf{D}_{3},\mathbf{Z},\mathbf{A}\right]\right.\\
 & \left.-\mathbb{E}\left[\left.\Delta\theta'T_{p\mathbf{i}}T_{p\mathbf{i}}'\Delta\theta\mathbf{1}(\left|T_{p\mathbf{i}}'\bar{\theta}_{\epsilon}\right|\geq a)\right|\mathbf{D}_{0},\mathbf{D}_{3},\mathbf{Z},\mathbf{A}\right]\right\} \\
\geq & \Lambda(a)(1-\Lambda(a))\left\{ \mathbb{E}\left[\left.\Delta\theta'T_{p\mathbf{i}}T_{p\mathbf{i}}'\Delta\theta\right|\mathbf{D}_{0},\mathbf{D}_{3},\mathbf{Z},\mathbf{A}\right]\right.\\
 & -\frac{1}{a}\left(\mathbb{E}\left[\left.(\Delta\theta'T_{p\mathbf{i}}T_{p\mathbf{i}}'\Delta\theta)^{2}\right|\mathbf{D}_{0},\mathbf{D}_{3},\mathbf{Z},\mathbf{A}\right]\right)^{1/2}\\
 & \left.\times\left(\mathbb{E}\left[\left.\bar{\theta}_{\epsilon}'T_{p\mathbf{i}}T_{p\mathbf{i}}'\bar{\theta}_{\epsilon}\right|\mathbf{D}_{0},\mathbf{D}_{3},\mathbf{Z},\mathbf{A}\right]\right)^{1/2}\right\}. 
\end{align*}
Combining the results above we get:
\begin{align}
\tilde{Q}_{N}(\theta_{0})-\tilde{Q}_{N}(\theta_{\epsilon})\geq & \frac{\Lambda(a)(1-\Lambda(a))}{\mathbb{E}\left[\sum_{p,\mathbf{i}}Z_{p\mathbf{i}}\right]}\sum_{p,\mathbf{i}}Z_{p\mathbf{i}}\cdot\left\{ \mathbb{E}\left[\left.\Delta\theta'T_{p\mathbf{i}}T_{p\mathbf{i}}'\Delta\theta\right|\mathbf{D}_{0},\mathbf{D}_{3},\mathbf{Z},\mathbf{A}\right]\right.\notag\\ 
 & -\frac{1}{a}\left(\mathbb{E}\left[\left.(\Delta\theta'T_{p\mathbf{i}}T_{p\mathbf{i}}'\Delta\theta)^{2}\right|\mathbf{D}_{0},\mathbf{D}_{3},\mathbf{Z},\mathbf{A}\right]\right)^{1/2}\notag\\
 & \left.\times\left(\mathbb{E}\left[\left.\bar{\theta}_{\epsilon}'T_{p\mathbf{i}}T_{p\mathbf{i}}'\bar{\theta}_{\epsilon}\right|\mathbf{D}_{0},\mathbf{D}_{3},\mathbf{Z},\mathbf{A}\right]\right)^{1/2}\right\}\notag \\
\geq & \Lambda(a)(1-\Lambda(a))\left\{ \Delta\theta'\mathbb{E}\left[\left.\frac{1}{\mathbb{E}\left[\sum_{p,\mathbf{i}}Z_{p\mathbf{i}}\right]}\sum_{p,\mathbf{i}}Z_{p\mathbf{i}}T_{p\mathbf{i}}T_{p\mathbf{i}}'\right|\mathbf{D}_{0},\mathbf{D}_{3},\mathbf{Z},\mathbf{A}\right]\Delta\theta\right.\notag\\
 & -\frac{1}{a}\left(\frac{1}{\mathbb{E}\left[\sum_{p,\mathbf{i}}Z_{p\mathbf{i}}\right]}\sum_{p,\mathbf{i}}Z_{p\mathbf{i}}\mathbb{E}\left[\left.(\Delta\theta'T_{p\mathbf{i}}T_{p\mathbf{i}}'\Delta\theta)^{2}\right|\mathbf{D}_{0},\mathbf{D}_{3},\mathbf{Z},\mathbf{A}\right]\right)^{1/2}\notag\\
 & \left.\times\left(\frac{1}{\mathbb{E}\left[\sum_{p,\mathbf{i}}Z_{p\mathbf{i}}\right]}\sum_{p,\mathbf{i}}Z_{p\mathbf{i}}\mathbb{E}\left[\left.\bar{\theta}_{\epsilon}'T_{p\mathbf{i}}T_{p\mathbf{i}}'\bar{\theta}_{\epsilon}\right|\mathbf{D}_{0},\mathbf{D}_{3},\mathbf{Z},\mathbf{A}\right]\right)^{1/2}\right\}\notag \\
 & \geq\Lambda(a)(1-\Lambda(a))\left\{ \epsilon^{2}C_{1}\right.\notag\\
 & -\frac{1}{a}\left(\frac{1}{\mathbb{E}\left[\sum_{p,\mathbf{i}}Z_{p\mathbf{i}}\right]}\sum_{p,\mathbf{i}}Z_{p\mathbf{i}}\|\Delta\theta\|^{4}\mathbb{E}\left[\left.\|T_{p\mathbf{i}}\|^{4}\right|\mathbf{D}_{0},\mathbf{D}_{3},\mathbf{Z},\mathbf{A}\right]\right)^{1/2}\notag\\
 & \left.\times\left(\sup_{\Theta}\|\theta\|^{2}\frac{1}{\mathbb{E}\left[\sum_{p,\mathbf{i}}Z_{p\mathbf{i}}\right]}\sum_{p,\mathbf{i}}Z_{p\mathbf{i}}\mathbb{E}\left[\left.\|T_{p\mathbf{i}}\|^{2}\right|\mathbf{D}_{0},\mathbf{D}_{3},\mathbf{Z},\mathbf{A}\right]\right)^{1/2}\right\}\notag \\
 & \geq\Lambda(a)(1-\Lambda(a))\bigg[\epsilon^{2}C_{1}-\frac{1}{a}\sqrt{\epsilon^{4}C_{2}}\sqrt{\sup_{\Theta}\|\theta\|^{2}\sqrt{C_{2}}}\bigg],\label{eq: LowerBound}
\end{align}
where the first relationship follows immediately from the calculations above, the second from Jensen's inequality, the third from Assumption \ref{ass: parameter_space}, Assumption \ref{ass: rank}, and the relationship $\mathrm{tr}\left(ba'ab'\right)^{1/2}=\left(b'ba'a\right)^{1/2}=\left\Vert a\right\Vert \left\Vert b\right\Vert$, and -- finally -- the fourth from Assumptions \ref{ass: rank} and \ref{ass: boundedness}.

Finally, if we choose an $a>0$ such that
\[
\delta(\epsilon):= \Lambda(a)(1-\Lambda(a)) \left[C_1 \epsilon^2 -\frac{1}{a}\epsilon^2 \sqrt{\sup_\Theta ||\theta|| ^2  }C_2^{\frac{3}{4}}\right]>0,
\]
we have
\begin{equation*}
    \big|\tilde Q_N(\theta_0)-\sup_{\|\theta - \theta_0\| \ge \varepsilon} \tilde Q_N(\theta)\big|\geq\delta(\epsilon)
\end{equation*}
as required. This completes step 2.
   
All the conditions of Lemma \ref{lem: consistency} are met, so we can conclude: $\hat{\theta}\rightarrow_p \theta_0$.

\subsection{Proof of Theorem \ref{thm: asymptotic-normality} (Asymptotic Normality)}\label{app: asymptotic-normality-proof}

\noindent Define 
$$L_N(\theta):=\frac{1}{\sum_{p,\mathbf{i}} \mathbb{E}\left[Z_{p\mathbf{i}}\right]}\sum_{p=2}^K \sum_{\mathbf{i}} l_{p\mathbf{i}} (\theta)$$
with $$ l_{p\mathbf{i}} (\theta):=Z_{p\mathbf{i}} \log\left(\frac{\exp(T_{p\mathbf{i}}'\theta)}{1+\exp(T_{p\mathbf{i}}'\theta)}. \right)  $$
A mean value expansion in $\theta$ around $\theta_0$, in conjunction with Assumption \ref{ass: parameter_space}, yields some $\bar \theta \in (\hat \theta, \theta_0)$ such that
\begin{align*}
   0=\frac{\partial L_N}{\partial \theta}(\hat \theta) &=\frac{\partial L_N}{\partial \theta}(\theta_0)+\frac{\partial^2 L_N}{\partial \theta^2}(\theta_0)(\hat{\theta}-\theta_0)+\sum_{k} (\hat \theta_k -\theta_{0,k}) \frac{\partial^3 L_N}{\partial \theta_k \partial \theta^2}(\bar \theta)(\hat \theta -\theta_{0})\\
   &=\frac{\partial L_N}{\partial \theta}(\theta_0)+\frac{\partial^2 L_N}{\partial \theta^2}(\theta_0)(\hat{\theta}-\theta_0)+o_p(\hat{\theta}-\theta_0),
\end{align*}
since
\begin{align*}
    \frac{\partial^3 L_N}{\partial \theta_j \partial \theta_k \partial \theta_l}(\bar \theta)=-\frac{1}{\sum_{p,\mathbf{i}} \mathbb{E}\left[Z_{p\mathbf{i}}\right]}\sum_{p,\mathbf{i}} Z_{p\mathbf{i}}\cdot\Lambda\left(T_{p\mathbf{i}}'\bar \theta\right)\left(1-\Lambda\left(T_{p\mathbf{i}}'\bar \theta\right)\right)T_{p\mathbf{i}, j} T_{p\mathbf{i}, k} T_{p\mathbf{i}, l},
\end{align*}
and by Assumption \ref{ass: boundedness}, the TI, the LIE, and JI (observing that $f\left(x\right)=x^{3/4}$ is concave):
\begin{align*}
    \mathbb{E}\left[\left|\frac{\partial^3 L_N}{\partial \theta_j \partial \theta_k \partial \theta_l}(\bar \theta)\right|\right] & \leq \frac{1}{\sum_{p,\mathbf{i}} \mathbb{E}\left[Z_{p\mathbf{i}}\right]}\sum_{p,\mathbf{i}}\mathbb{E}\left[\left|Z_{p\mathbf{i}} \cdot T_{p\mathbf{i}, j} T_{p\mathbf{i}, k} T_{p\mathbf{i}, l}\right|\right]\\
    & \leq \frac{1}{\sum_{p,\mathbf{i}} \mathbb{E}\left[Z_{p\mathbf{i}}\right]} \sum_{p,\mathbf{i}}\mathbb{E}\left[Z_{p\mathbf{i}}\cdot \|T_{p\mathbf{i}} \|^3\right]\\
    &\leq \frac{1}{\sum_{p, \mathbf{i}}  \mathbb{E}\left[Z_{p\mathbf{i}}\right]} \mathbb{E}\left[\left(\sum_{p, \mathbf{i}} Z_{p\mathbf{i}}\right) \mathbb{E}\left[\left.\frac{\sum_{p, \mathbf{i}} Z_{p\mathbf{i}} \|T_{p\mathbf{i}}\|^3}{\sum_{p, \mathbf{i}} Z_{p\mathbf{i}}}\right|\mathbf{D}_0,\mathbf{D}_3,\mathbf{Z}, \mathbf{A}\right]\right]\\
    &\leq \frac{1}{\sum_{p, \mathbf{i}}  \mathbb{E}\left[Z_{p\mathbf{i}}\right]} \mathbb{E}\left[\left(\sum_{p, \mathbf{i}} Z_{p\mathbf{i}}\right) \mathbb{E}\left[ \left.\frac{\sum_{p, \mathbf{i}} Z_{p\mathbf{i}} \|T_{p\mathbf{i}}\|^4}{\sum_{p, \mathbf{i}} Z_{p\mathbf{i}}}\right|\mathbf{D}_0,\mathbf{D}_3,\mathbf{Z}, \mathbf{A}\right]^{3/4}\right]\\
    & \leq  \frac{1}{\sum_{p, \mathbf{i}} \mathbb{E}\left[Z_{p\mathbf{i}}\right]} \mathbb{E}\left[\left(\sum_{p, \mathbf{i}} Z_{p\mathbf{i}}\right) \left( C_2\right)^{3/4}\right]\leq C_2^{3/4},
\end{align*}
where the last inequality is a direct implication of Assumption \ref{ass: boundedness}. We consequently have
\begin{align*}
    \hat \theta -\theta_0=-\left( \frac{\partial^2 L_N}{\partial \theta^2}(\theta_0)+o_p(1)\right)^{+}\frac{\partial L_N}{\partial \theta}(\theta_0)
\end{align*}
where a `$+$' superscript denotes the generalized inverse. To demonstrate asymptotic normality of $\hat \theta$ we therefore need to show that (i) $\frac{\partial L_N}{\partial \theta}\left(\theta_0\right)$, suitably normalized, asymptotically behaves like a normal random variable and (ii) $\frac{\partial^2 L_N}{\partial \theta^2}\left(\theta_0\right)$ converges in probability to a positive definite matrix.

\noindent \textbf{Asymptotic normality of $\frac{\partial L_N}{\partial \theta}(\theta_0)$}: We use the Wald Device. Fix some $\lambda\in \mathbb{R}^2$.
Note that, by conditional independence:
$$\mathbb{V}\left( \sum_{p, \mathbf{i}} \lambda ' \nabla_\theta l_{p\mathbf{i}}(\theta_0)\right)= \sum_{p, \mathbf{i}} \mathbb{V}\left(  \lambda '\nabla_\theta l_{p\mathbf{i}}(\theta_0)\right)= \sum_{p, \mathbf{i}} \mathbb{E} \left[ \left(  \lambda '\nabla_\theta  l_{p\mathbf{i}}(\theta_0)\right)^2\right].$$

Define the algebra: $\mathcal{G}_N:=\sigma(\mathbf{D}_0,\mathbf{D}_3,\mathbf{Z}, \mathbf{A})$, and enumerate the tuples $(p, \bold i)$ so that every tuple $(p, \bold i)$ corresponds to an integer $i \in \{1,\dots \}$. Define the filtration $\mathcal{F}_{N, i}:= \sigma(\mathcal{G}_N, {t_1}, t_2, \dots, t_i)$.

Define $W_i:=\frac{1}{\sqrt{\mathbb{V}\left(\sum_{p, \mathbf{i} }  \lambda' \nabla_{\theta}l_{p\mathbf{i}}(\theta_{0})\right)}} \lambda' \nabla_\theta l_{i}(\theta_0)$. We verify the conditions of Theorem 2.3 in  \cite{McLeish1974} for the martingale difference $W_i$:

\begin{enumerate}[label=(\alph*)]
    \item $\mathbb{E}\left[\max_i |W_i|^2\right]\leq \sum_i\mathbb{E} \left[|W_i|^2\right]=1$, so $\max_i |W_i|$ is uniformly bounded in $L_2$.

    \item Fix $\epsilon>0$. Note that $\Pr(\max_i |W_i|>\epsilon)\leq \Pr\left(\sum_i W_i^2 1(|W_i|>\epsilon)> \epsilon^2\right)$.\footnote{This follows because occurrence of the event on the right implies that the one on the left also occurs.} Moreover:
    $$ \sum_i W_i^2 \cdot 1(|W_i|>\epsilon) \leq  \frac{1}{\epsilon^2}\sum_i W_i^{4}=\frac{1}{\epsilon^2}\frac{\sum_{p, \mathbf{i}} \left( \lambda' \nabla_{\theta}  l_{p\mathbf{i}}(\theta_{0}) \right)^4 }{\left(\sum_{p, \mathbf{i}} \mathbb{E}\left[(\lambda'\nabla_{\theta}  l_{p\mathbf{i}}(\theta_{0}))^2 \right]\right)^2}.$$
   Next observe that, using Assumption \ref{ass: boundedness}:
 \begin{align*}
    \mathbb{E}\left[\sum_{p, \mathbf{i}} \left( \lambda' \nabla_{\theta}  l_{p\mathbf{i}}(\theta_{0}) \right)^4 \right]&=\sum_{p,\mathbf{i}} \mathbb{E}\left[ \left( \lambda' \nabla_{\theta}  l_{p\mathbf{i}}(\theta_{0}) \right)^4 \right] \\
    &= \sum_{p, \mathbf{i}} \mathbb{E}\left[(1-\Lambda(T_{p\mathbf{i}}'\theta_0 ))^4 \cdot Z_{p\mathbf{i}} \cdot (\lambda'T_{p\mathbf{i}})^4\right]\\
    &\leq \|\lambda\|^4 \sum_{p, \mathbf{i}} \mathbb{E}\left[\|T_{p\mathbf{i}}\|^4 Z_{p\mathbf{i}} \right]\\
    &\leq \|\lambda\|^4 \cdot\mathbb{E}\left[\left(\sum_{p, \mathbf{i}} Z_{p\mathbf{i}}\right) \mathbb{E}\left[\left.\frac{\sum_{p, \mathbf{i}} \|T_{p\mathbf{i}}\|^4 \cdot Z_{p\mathbf{i}}}{\sum_{p, \mathbf{i}} Z_{p\mathbf{i}}} \right| \mathbf{D}_0,\mathbf{D}_3,\mathbf{Z}, \mathbf{A}\right]\right]\\
    &\leq C_2 \cdot \|\lambda\|^4 \cdot \mathbb{E}\left[\sum_{p, \mathbf{i}} Z_{p \mathbf{i}}\right]. 
\end{align*}

Moreover, via an argument similar to that used to derive \eqref{eq: LowerBound} above, there exists a constant $C_{1}$ such that
\begin{equation*}
    \sum_{p, \mathbf{i}} \mathbb{E}\left[(\lambda'\nabla_{\theta}  l_{p \mathbf{i}}(\theta_{0}))^2 \right]\geq C_1 \cdot \|\lambda\|^2 \cdot \mathbb{E}\left[\sum_{p,  \mathbf{i}} Z_{p \mathbf{i}}\right], 
\end{equation*}
and hence we have that:
\begin{align*}
    \mathbb{E}\left[\sum_i W_i^2 \cdot 1(|W_i|>\epsilon) \right]&\leq \frac{1}{\epsilon^2}\frac{C_2 \cdot\|\lambda\|^4 \cdot \sum_{p,  \mathbf{i}} \mathbb{E}\left[Z_{p \mathbf{i}}\right]}{C_1^2 \cdot\|\lambda\|^4 \cdot \left(\sum_{p,  \mathbf{i}} \mathbb{E}\left[Z_{p \mathbf{i}}\right]\right)^2},
\end{align*}
which implies, invoking Assumption \ref{ass: identifying_events}, that: \begin{equation}\label{eq: LindebergFeller}
    \mathbb{E}\left[\sum_i W_i^2 \cdot 1(|W_i|>\epsilon) \right]\rightarrow 0,
\end{equation} 
and hence that $\max_i |W_i| \overset{p}{\rightarrow} 0$ as required.

\item This requirement is satisfied by Theorem 2.23 in \cite{HallHeyde1980}. Condition 2.24 in that reference follows from observing that:
$$\sum_{i=1}^n \mathbb{E}\left[W_i^2|\mathcal{F}_{n, i-1}\right]=\frac{\mathbb{V}\left(\left.\sum_{p, \mathbf{i}} \lambda' \nabla_{\theta} l_{p\mathbf{i}}(\theta_{0})\right|\mathcal{G}_N\right)}{\mathbb{V}\left(\sum_{p, \mathbf{i}} \lambda' \nabla_{\theta} l_{p\mathbf{i}}(\theta_{0})\right)},$$

then, by the MI, for any $M>0$
\begin{align*}
    \Pr\left( \sum_{i=1}^n \mathbb{E}\left[W_i^2|\mathcal{F}_{n, i-1}\right]>M\right)&\leq \frac{\mathbb{E}\left(\mathbb{V}\left(\left.\sum_{p, \mathbf{i}}  \lambda' \nabla_{\theta}l_{p\mathbf{i}}(\theta_{0})\right|\mathcal{G}_N\right)\right)}{M\times \mathbb{V}\left(\sum_{p, \mathbf{i}} \lambda' \nabla_{\theta} l_{p\mathbf{i}}(\theta_{0})\right)} \\
    &\leq \frac{1}{M}
\end{align*}
where the second inequality follows by the law of total variance. So that:
$$\sup_N  \Pr\left( \sum_{i=1}^n \mathbb{E}\left[\left.W_i^2\right|\mathcal{F}_{n, i-1}\right]>M\right) \rightarrow 0; \mbox{ as } M\rightarrow +\infty$$
as desired and required. Condition 2.25 of \cite{HallHeyde1980} (the conditional Lindeberg-Feller condition) follows from equation \eqref{eq: LindebergFeller} and an application of the MI. Fix $M>0$ and $\epsilon>0$:
\begin{align*}
   \Pr \left(\sum_i \mathbb{E}\left[\left.W_i^2 \cdot 1 \left(|W_i|>\epsilon\right)\right|\mathcal{F}_{N,i-1}\right]\geq M  \right)&\leq \frac{\mathbb{E}\left[\sum_i \mathbb{E}\left[\left.W_i^2 \cdot 1(|W_i|>\epsilon)\right|\mathcal{F}_{N,i-1}\right] \right]}{M}\\
   &=\frac{\mathbb{E}\left[\sum_i W_i^2 \cdot 1(|W_i|>\epsilon) \right] }{M}\\
   &\rightarrow_{N\rightarrow+\infty}0
\end{align*}
Then Equation 2.26 in \cite{HallHeyde1980} shows that $\sum W_i^2 \rightarrow_p 1$.
\end{enumerate}
Steps (a), (b) and (c) above allow us to conclude that:
\begin{equation*}
    \frac{1}{\sqrt{\mathbb{V}\left(\sum_{p, \mathbf{i}} \lambda' \nabla_{\theta}l_{p\mathbf{i}}(\theta_{0})\right)}} \sum_i \lambda' \nabla_\theta l_{i}(\theta_0)\rightarrow_d \mathcal{N}(0,1)
\end{equation*}
By Assumption \ref{ass: hessian}, Slutsky's lemma, as well as the information matrix inequality, we then get:
\begin{equation*}
    \frac{1}{\sqrt{\lambda' \mathcal{I}\left(\theta_0\right) \lambda}} \cdot \frac{1}{\sqrt{\mathbb{E}\left[\sum_{p,\mathbf{i}} Z_{p\mathbf{i} }\right]}} \cdot \sum_i \lambda' \nabla_\theta l_{i}(\theta_0)\rightarrow_d \mathcal{N}\left(0,1\right).
\end{equation*}
Since $\lambda$ is arbitrary, we have:
\begin{equation}\label{eq: CLT}
     \sqrt{\mathbb{E}\left[\sum_{p, \mathbf{i}} Z_{p\mathbf{i}}\right]} \cdot \frac{\partial L_N (\theta_0)}{\partial \theta} \overset{d}{\rightarrow} \mathcal{N}\left(0,\mathcal{I}\left(\theta_0\right)\right).
\end{equation}
\noindent \textbf{Convergence of $\frac{\partial^2 L_N}{\partial \theta^2}(\theta_0)$}: Notice that, for any $\lambda$, by the information matrix equality, almost surely for every $p, \mathbf{i}$:
$$ \mathbb{E}\left[\left.(\lambda'\nabla_\theta l_{p\mathbf{i}}(\theta_{0}))^2+\lambda' \nabla_{\theta\theta} l_{p\mathbf{i}}(\theta_0) \lambda\right|\mathbf{D}_0,\mathbf{D}_3,\mathbf{Z}, \mathbf{A}\right]=0,$$
and hence we get, using the conditional mean zero property of the score vector and conditional independence (i.e., Lemma \ref{lem: conditional-independence}),
$$\mathbb{V}\left( \frac{\sum_{p,\mathbf{i}} (\lambda' \nabla_\theta l_{p\mathbf{i}}(\theta_{0}))^2+\lambda' \nabla_{\theta\theta} l_{p\mathbf{i}}(\theta_0) \lambda }{\mathbb{V}\left(\sum_{p, \mathbf{i}} \lambda' \nabla_{\theta} l_{p\mathbf{i}}(\theta_{0})\right)}\right) = \frac{\sum_{p, \mathbf{i}} \mathbb{E}\left[\left((\lambda'\nabla_\theta l_{p\mathbf{i}}(\theta_{0}))^2+\lambda' \nabla_{\theta\theta} l_{p\mathbf{i}}(\theta_0) \lambda\right)^2\right] }{\left(\sum_{p,\mathbf{i}} \mathbb{E}\left[(\lambda' \nabla_{\theta} l_{p\mathbf{i}}(\theta_{0}))^2\right]\right)^2}.$$ 
By arguments similar to those used to derive  equation \eqref{eq: LowerBound}, we can show that $\sum_{p, \mathbf{i}} \mathbb{E}\left[(\lambda' \nabla_{\theta}l_{p\mathbf{i}}(\theta_{0}))^2\right]$ is lower bounded by a constant times $ \mathbb{E}\left[\sum_{p, \mathbf{i} } Z_{p\mathbf{i}}\right]$ and that $\sum_{p, \mathbf{i}} \mathbb{E}\left[\left((\lambda'\nabla_\theta l_{p\mathbf{i} }(\theta_{0}))^2+\lambda' \nabla_{\theta\theta} l_{p\mathbf{i}}(\theta_0) \lambda\right)^2\right]$ is upper bounded by another constant times $ \mathbb{E}\left[\sum_{p, \mathbf{i} } Z_{p\mathbf{i}}\right]$, then by Assumption \ref{ass: identifying_events}:
$$\mathbb{V}\left( \frac{\sum_{p,\mathbf{i}} (\lambda' \nabla_\theta l_{p\mathbf{i}}(\theta_{0}))^2+\lambda' \nabla_{\theta\theta} l_{p\mathbf{i}}(\theta_0) \lambda }{\mathbb{V}\left(\sum_{p, \mathbf{i}} \lambda' \nabla_{\theta} l_{p\mathbf{i}}(\theta_{0})\right)}\right)\rightarrow 0,$$

By our demonstration of pointwise convergence of the conditional star system log-likelihood in Step 1.a of Theorem \ref{thm: consistency} above, $\frac{\sum_{p, \mathbf{i}} (\lambda' \nabla_\theta l_{p\mathbf{i}}(\theta_{0}))^2}{\mathbb{V}\left(\sum_{p, \mathbf{i}} \lambda' \nabla_{\theta} l_{ p\mathbf{i} }(\theta_{0})\right)} \overset{p}{\rightarrow} 1$ and hence, again invoking the information matrix equality,
\begin{equation*}
    \frac{-1}{{\mathbb{V}\left(\sum_{p, \mathbf{i}} \lambda' \nabla_{\theta} l_{p\mathbf{i}}(\theta_{0})\right)}} \sum_{p, \mathbf{i}} {\lambda' \nabla_{\theta\theta} l_{p\mathbf{i}}(\theta_0) \lambda} \overset{p}{\rightarrow} 1.
\end{equation*}
Assumption \ref{ass: hessian} then implies:
$$ \frac{1}{\mathbb{E}\left[\sum_{p,\mathbf{i}} Z_{p\mathbf{i}}\right]} \sum_{p, \mathbf{i}} {\lambda' \nabla_{\theta\theta} l_{p\mathbf{i}}(\theta_0) \lambda} \overset{p}{\rightarrow} -\lambda'\mathcal{I}\left(\theta_0\right)\lambda . $$

Since convergence happens for any $\lambda$, we therefore have:\footnote{ We use the identity: $ x' H y
= \frac{1}{4} \Big[ (x+y)' H (x+y) - (x-y)' H (x-y) \Big] $ for any symmetric $H$ and any vectors $x$ and $y$.}

\begin{equation}\label{eq: HessianConvergence}
    \frac{1}{\mathbb{E}\left[\sum_{p, \mathbf{i}} Z_{p\mathbf{i}}\right]} \frac{\partial^2 L_N(\theta_0)}{\partial \theta^2} \overset{p}{\rightarrow} -\mathcal{I}\left(\theta_0\right).     
\end{equation}

Equation \eqref{eq: CLT} and \eqref{eq: HessianConvergence} allow to conclude that:

$$\sqrt{\mathbb{E}\left[\sum_{p, \mathbf{i} } Z_{p\mathbf{i}}\right]} \cdot (\hat \theta- \theta_0) \overset{D}{\rightarrow} \mathcal{N}\left(0, \mathcal{I}^{-1}\left(\theta_0\right)\right),$$
as claimed.

\subsection{Proof of Proposition \ref{prop:ER}}\label{proof:ER}

We consider the two cases $\beta_0\geq 0$ and $\beta_0< 0$ separately.\\
\textbf{Case 1: $\beta_0\geq 0$}\\ We show that for any dyad $i\neq j$, the probability of swapping an edge between periods 1 and 2 decays exponentially with $N$. To start, we present and prove a Lemma which controls the tails of $R_{ij0}$ and $R_{ij1}$.
     \begin{lem}\label{lem: RTail}
        Assume $\beta_0\geq 0$. Define $c_0:= \frac{1}{2}p^2$ and $c_1:=\frac{1}{2}\Lambda(-|\alpha_0|-\bar{A})^2$ 
        , there exist $\gamma_0>0$ and  $\gamma_1>0$ independent of $N$ such that:
        $$\mathbb{P}\left(\exists i,j:\; R_{ij0}\leq c_0 N\right)\leq N^2 \exp(-\gamma_0 N)\mbox{ and } \mathbb{P}\left(\exists i,j:\; R_{ij1}\leq c_1 N\right)\leq N^2 \exp(-\gamma_1 N).$$
    \end{lem}
    \begin{proof}[Proof of Lemma \ref{lem: RTail}]
We start by observing that $R_{ij0}=\sum_{k}D_{ik0}D_{jk0}$ is a sum of $N-2$ independent Bernoulli random variables each with success probability $p^2$ and hence $\mathbb{E}[R_{ij0}]=(N-2)p^2$. Hoeffding's Inequality (HI) then gives, for $N$ large enough:
\begin{align*}
\mathbb{P}(\exists i,j:\quad R_{ij0}\leq c_0 N)&\leq \sum_{i<j} \mathbb{P}( R_{ij0}\leq c_0 N)\\
& \leq \sum_{i<j} \mathbb{P}\left(R_{ij0}-\mathbb{E}[R_{ij0}]\leq c_0 N -\mathbb{E}[R_{ij,0}]\right)\\
&\leq \tbinom{N}{2} \times \exp\left(-\frac{2[(N-2)p^2-c_0N]^2}{N-2} \right)\\
&\leq N^2 \exp\left(-\frac{Np^4}{6}\right).
\end{align*}
\noindent So the first hypothesis holds with $\gamma_0=\frac{p^4}{6}$. Next consider the second hypothesis. By the law of total probability:
\begin{align*}
\mathbb{P}(\exists i,j:\quad R_{ij1}\leq c_1 N)&\leq \sum_{i<j} \mathbb{P}( R_{ij1}\leq c_1 N)\\
& \leq \sum_{i<j} \mathbb{E}\left[\mathbb{P}\left(R_{ij1}-\mathbb{E}[R_{ij1}|\mathbf{D}_0, \mathbf{A}]\leq c_1 N-\mathbb{E}[R_{ij1}|\mathbf{D}_0, \mathbf{A}]\big | \mathbf{D}_0, \mathbf{A}\right)\right].
\end{align*}
\noindent Recalling that $A_{ij} \in [-\bar{A},\bar{A}]$ we get:
\begin{align*}
    \mathbb{E}[R_{ij1}|\mathbf{D}_{0},\mathbf{A}]&=\sum_{k\neq i,j}\Lambda(\alpha_{0}D_{ik0}+\beta_{0}R_{ik0}+A_{ik})\Lambda(\alpha_{0}D_{jk0}+\beta_{0}R_{jk0}+A_{jk})\\&\geq\sum_{k\neq i,j}\Lambda({-|\alpha_0|+}\beta_{0}R_{ik0}-\bar{A})\Lambda(-|\alpha_0|+\beta_{0}R_{jk0}-\bar{A})\\&\geq\left(N-2\right)\Lambda({-|\alpha_0|}-\bar{A})^{2}.
\end{align*}
\noindent Invoking the HI for a second time yields, for $N$ large enough:
\begin{align*}
    \mathbb{P}\left(R_{ij1}-\mathbb{E}[R_{ij1}|\mathbf{D}_0, \mathbf{A}]\leq c_1 N-\mathbb{E}[R_{ij1}|\mathbf{D}_0, A]\big | \mathbf{D}_0, \mathbf{A}\right)&\leq  \exp\left(-\frac{2\left(\mathbb{E}[R_{ij1}|\mathbf{D}_0, \mathbf{A}]-c_1 N \right)
    ^2}{N-2} \right)\\
    & \leq\exp\left(-\frac{2\left(\left[\frac{N-4}{2}\right]\Lambda\left({-|\alpha_0|}-\bar{A}\right)^{2}\right)^{2}}{N-2}\right) \\
    & \leq \exp\left(-\frac{N\Lambda\left({-|\alpha_0|}-\bar{A}\right)^{4}}{6}\right). 
\end{align*}
\noindent Then
\begin{align*}
\mathbb{P}(\exists i,j:\quad R_{ij1}\leq c_1 N| \mathbf{D}_0, \mathbf{A})&\leq N^2  \exp\left(-\frac{N\Lambda\left({-|\alpha_0|}-\bar{A}\right)^{4}}{6}\right) 
\end{align*}
So the second hypothesis holds with $\gamma_1=\frac{\Lambda{({-|\alpha_0|}-\bar{A}})^4}{6}$.

\vspace{10mm}    
\noindent Returning to the proof of Proposition \ref{prop:ER}. Let $$\mathcal{E}_0 = \{\forall k,l:\; R_{kl0} > c_0 N\}\; \mbox{ and }  \mathcal{E}_1 = \{\forall k,l:\; R_{kl1} > c_1 N\}.$$ By Lemma \ref{lem: RTail}, $\mathbb{P}(\mathcal{E}_0^c) \leq N^2 \exp(-\gamma_0 N)$ and $\mathbb{P}(\mathcal{E}_1^c) \leq N^2 \exp(-\gamma_1 N)$.

Fix a pair $i,j$. Conditional on $\mathbf{D}^1$ and $\mathbf{A}$, on the event $\mathcal{E}_1$, we have $R_{kl1} > c_1 N$ for all $k,l$. Hence:
\begin{align*}
\mathbb{P}(D_{ij1}\neq D_{ij2}| \mathbf{D}^1, \mathbf{A}) 1(\mathcal{E}_1) &\leq 1(D_{ij1}=0) + \mathbb{P}(D_{ij2}=0|D_{ij1}=1, \mathbf{D}^1, \mathbf{A}) 1(\mathcal{E}_1) \\
&\leq 1(D_{ij1}=0) + (1-\Lambda(\alpha_0+\beta_0 c_1 N - \bar{A})) \\
&\leq 1(D_{ij1}=0) + e^{-\alpha_0+\bar{A}}e^{-\beta_0 c_1 N}.
\end{align*}

The probability that an arbitrary $K$-system $i:=\{i_1, \dots, i_K\}$ is identifying, conditional on $\mathbf{D}^1, \mathbf{A}$ and restricted to $\mathcal{E}_1$, is bounded above by:
\begin{align*}
\mathbb{P}(Z_i=1 | \mathbf{D}^1, \mathbf{A}) 1(\mathcal{E}_1) &\leq \prod_{l\neq 1} \mathbb{P}(D_{i_1 i_l 1} \neq D_{i_1 i_l 2} | \mathbf{D}^1, \mathbf{A}) 1(\mathcal{E}_1) \\
&\leq \prod_{l\neq 1} \left(1(D_{i_1 i_l 1}=0) + e^{-\alpha_0+\bar{A}}e^{-\beta_0 c_1 N}\right).
\end{align*}

Since the period-1 links form independently conditional on $\mathbf{D}^0, \mathbf{A}$, taking the expectation over $\mathbf{D}^1$ gives:
$$\mathbb{E}[\mathbb{P}(Z_i=1 | \mathbf{D}^1, \mathbf{A}) 1(\mathcal{E}_1) | \mathbf{D}^0, \mathbf{A}] \leq \prod_{l\neq 1} \left(\mathbb{P}(D_{i_1 i_l 1}=0 | \mathbf{D}^0, \mathbf{A}) + e^{-\alpha_0+\bar{A}}e^{-\beta_0 c_1 N}\right).$$

On the event $\mathcal{E}_0$, $R_{kl0} > c_0 N$ for all $k,l$, so $\mathbb{P}(D_{i_1 i_l 1}=0 | \mathbf{D}^0, \mathbf{A}) \leq e^{|\alpha_0|+\bar{A}}e^{-\beta_0 c_0 N}$. Thus:
\begin{align*}
\mathbb{P}(Z_i=1, \mathcal{E}_1 | \mathbf{D}^0, \mathbf{A}) 1(\mathcal{E}_0) &\leq \prod_{l\neq 1} \left(e^{|\alpha_0|+\bar{A}}e^{-\beta_0 c_0 N} + e^{-\alpha_0+\bar{A}}e^{-\beta_0 c_1 N}\right) \\
&\leq C^{K-1} \exp(-\gamma_3 (K-1) N),
\end{align*}
for some constants $C>0$ and $\gamma_3>0$ independent of $N, K$. Consequently, taking the unconditional expectation ensures:
$$\mathbb{P}(Z_i=1, \mathcal{E}_0 \cap \mathcal{E}_1) \leq C^{K-1} \exp(-\gamma_3 (K-1) N).$$

To bound the expected number of identifying stars, we partition the expectation. By Corollary 1, any dyad belongs to at most one identifying star, thus the total number of identifying stars is deterministically bounded by $N/2$. We then obtain:
\begin{align*}
\mathbb{E}\left[\sum_{K} m_{K,N}\right] &= \mathbb{E}\left[\sum_{K} m_{K,N} 1(\mathcal{E}_0 \cap \mathcal{E}_1)\right] + \mathbb{E}\left[\sum_{K} m_{K,N} 1((\mathcal{E}_0 \cap \mathcal{E}_1)^c)\right] \\
&\leq  N \sum_{K=2}^{N-1} \binom{N-1}{K-1} \cdot \mathbb{P}(Z_i=1, \mathcal{E}_0 \cap \mathcal{E}_1) + \frac{N}{2} \mathbb{P}(\mathcal{E}_0^c \cup \mathcal{E}_1^c) \\
&\leq N \sum_{K=2}^{N-1} \binom{N-1}{K-1} \left(C \exp(-\gamma_3 N)\right)^{K-1} + \frac{N}{2} \left(N^2 \exp(-\gamma_0 N) + N^2 \exp(-\gamma_1 N)\right) \\
&= N \left( \left[1 + C \exp(-\gamma_3 N)\right]^{N-1} - 1- \left[C \exp(-\gamma_3 N)\right]^{N-1}  \right)\\
& \quad + \frac{N^3}{2} \left(\exp(-\gamma_0 N) + \exp(-\gamma_1 N)\right).
\end{align*}

Where the second inequality follows by the separation lemma and Corollary \ref{cor: separation-dyads-corro}. Since $(1+x)^m - 1 \sim mx$ for small $x$, the first term behaves as $N(N-1) C \exp(-\gamma_3 N)$, which converges to $0$. The second term is a polynomial multiplied by an exponentially decaying factor, which also clearly converges to $0$. Thus $\mathbb{E}\left[\sum_{K} \bold m_{K,N}\right] \xrightarrow{N\to\infty} 0$ as claimed.

\noindent \textbf{Case 2: $\beta_0<0$} \\
For this case we divide the proof into three steps.

\medskip
\noindent
\textbf{Step 1: the period-1 graph is empty with overwhelming probability.}

Let
\[
\mathcal E_N:=\{D_{ij1}=0\ \forall i<j\}
\]
denote the event that the period-1 graph is empty.

Fix a dyad $ij$. Since $\Lambda(x)\le e^x$ for all $x\in\mathbb R$, we have
\[
\Pr(D_{ij1}=1\mid \mathbf{D}_0, \mathbf{A})
=
\Lambda(\alpha_0 D_{ij0}+\beta_0 R_{ij0}+A_{ij})
\le
\exp\big(\alpha_+ + \beta_0 R_{ij0} + \bar A\big),
\]
where $\alpha_+ := \max\{\alpha_0,0\}$. Taking expectations and using $\beta_0<0$,
\[
\Pr(D_{ij1}=1)
\le
e^{\alpha_+ + \bar A}\,\mathbb E[e^{\beta_0 R_{ij0}}].
\]
Under $\mathbf{D}_0\sim G(N,p)$,
\[
R_{ij,0}\sim \mathrm{Bin}(N-2,p^2),
\]
hence
\[
\mathbb E[e^{\beta_0 R_{ij,0}}]
=
\big(1-p^2+p^2 e^{\beta_0}\big)^{N-2}
\le
\exp\big(-(1-e^{\beta_0})(N-2)p^2\big).
\]
Therefore there exist constants $C_1,C_2>0$ such that, for all $N$ large enough,
\[
\Pr(D_{ij1}=1)
\le
C_1 e^{-C_2 N p^2}.
\]
By a union bound,
\[
\Pr(\mathcal E_N^c)
\le
\sum_{i<j}\Pr(D_{ij1}=1)
\le
C_1 N^2 e^{-C_2 N p^2}.
\]
Since $N p^2/\log N\to\infty$, it follows that
\[
\Pr(\mathcal E_N^c)=o(N^{-m})
\qquad\text{for every fixed }m>0.
\]

\medskip
\noindent
\textbf{Step 2: on $\mathcal E_N$, a given candidate star is exponentially unlikely to be identifying.}

Fix $K\in\{2,\dots,N-1\}$ and consider a candidate $K$-star system
\[
\mathcal{S}=\{i_1i_2,\dots,i_1i_K\}
\]
with center $i_1$ and leaves $i_2,\dots,i_K$.

On $\mathcal E_N$, we have $D_{uv1}=0$ for every dyad $uv$, and hence also
\[
R_{uv1}=0
\qquad\text{for every }u\neq v.
\]
Therefore, conditional on $(\mathcal E_N, \mathbf{A})$, the period-2 link variables
$\{D_{uv2}\}_{u<v}$ are independent across dyads and satisfy
\[
\Pr(D_{uv2}=1\mid \mathcal E_N,\mathbf{A})=\Lambda(A_{uv}),
\qquad
\Pr(D_{uv2}=0\mid \mathcal E_N,\mathbf{A})=1-\Lambda(A_{uv})=\Lambda(-A_{uv}).
\]
Since $|A_{uv}|\le \bar A$, both probabilities are bounded above by
\[
c_*:=\Lambda(\bar A)<1.
\]

Now suppose that $\mathcal{S}$ is identifying. Because $\mathbf{D}_{1}=0$, Definitions \ref{def: stable-neighborhood} and \ref{def: identifying-star-system} imply the
following period-2 requirements:

\begin{enumerate}
\item each of the $K-1$ spoke dyads in $\mathcal{S}$ must satisfy
\[
D_{i_1i_m2}=1,\qquad m=2,\dots,K;
\]
\item each of the $\binom{K-1}{2}$ leaf-leaf dyads must satisfy
\[
D_{i_m i_l2}=0,\qquad 2\le m<l\le K;
\]
\item each dyad joining one of the $K$ star vertices to one of the $N-K$ outside vertices
must satisfy
\[
D_{uv2}=0.
\]
There are exactly $K(N-K)$ such dyads.
\end{enumerate}

Hence
\[
\Pr(\mathcal{S}\text{ is identifying}\mid \mathcal E_N,\mathbf{A})
\le
c_*^{\,L_{N,K}}.
\]
with \[
L_{N,K}
:=
(K-1)+\binom{K-1}{2}+K(N-K)
=
K(N-K)+\binom{K}{2}
= K\left(N-\frac{K+1}{2}\right).\]

\medskip
\noindent
\textbf{Step 3: summing over all sizes and all systems.}

Define:
\[
Y_{ij}:=1\{\text{dyad }ij\text{ is a spoke dyad of some identifying star system}\}.
\]
By Dyad Separation (see Corollary \ref{cor: separation-dyads-corro}), any dyad belongs to at most one identifying star system, and every
identifying star system contains at least one spoke dyad. Therefore
\[
\sum_K \bold m _{K,N}\le \sum_{i<j} Y_{ij}.
\]

By symmetry,
\[
\mathbb E\left[\sum_K \bold m _{K,N} 1(\mathcal E_N)\right]
\le
\binom{N}{2}\Pr(Y_{12}=1,\mathcal E_N)
=
\binom{N}{2}\,
\mathbb E\!\left[
1(\mathcal E_N)\Pr(Y_{12}=1\mid \mathcal E_N,\mathbf{A})
\right].
\]
Fix $K$. The number of candidate $K$-stars for which dyad $12$ is a spoke is
\[
2\binom{N-2}{K-2},
\]
since either node $1$ is the center and node $2$ is one of the leaves, or vice versa,
and then the remaining $K-2$ leaves are chosen from the other $N-2$ nodes. Hence, by a
union bound and the estimate from Step 2,
\[
\Pr(Y_{12}=1\mid \mathcal E_N,A)
\le
\sum_{K=2}^{N-1}
2\binom{N-2}{K-2} c_*^{\,L_{N,K}}.
\]

We now bound this sum.

\[\mathbb{E}\left[\sum_K \bold m _{K,N}| \mathcal E_N,A\right] \leq N\sum_{K=2}^{N-1} \binom{N-1}{K-1} c_*^{\,L_{N,K}}\]

For $2\le K\le N/2$, we have
\[
L_{N,K}
=
K\Big(N-\frac{K+1}{2}\Big)
\ge
\frac{K N}{2},
\]
so
\[
\binom{N-1}{K-1} c_*^{\,L_{N,K}}
\le
N^{K-1} c_*^{K N/2}
=
c_*^{N/2} \big(N c_*^{N/2}\big)^{K-1}.
\]
Since $c_*<1$, we have $N c_*^{N/2}\to 0$, and therefore, for all large $N$,
\[
\sum_{K=2}^{\lfloor N/2\rfloor}
 N \binom{N-1}{K-1}  c_*^{\,L_{N,K}}
\le
\sum_{K=2}^{\lfloor N/2\rfloor}
 N \binom{N-1}{K-1} (N c_*^{N/2}\big)^{K}\leq 2 N^3 c_*^N.
\]
where the last inequality follows by upper bounding the sum from $K=2$ to $\lfloor N/2\rfloor$ by the sum from $K=2$ to $N-1$ then using the binomial identity.

For $K>N/2$, the function $K\mapsto L_{N,K}$ is bounded below by a quadratic term; in
particular, for all large $N$,
\[
L_{N,K}\ge \frac{N^2}{4}.
\]
Hence
\[
\sum_{K=\lfloor N/2\rfloor+1}^{N-1}
N \binom{N-1}{K-1}  c_*^{\,L_{N,K}}
\le
N 2^N c_*^{N^2/4}.
\]

Combining the two ranges,
$$\mathbb{E}\left[\sum_K \bold m _{K,N}|A,\mathcal E_N\right]\leq N 2^N c_*^{N^2/4}+2N^3 c_*^N $$

Note that 
\begin{align*}
    \mathbb{E}\left[\sum_K \bold m _{K,N}|\mathbf{A}\right]&=\mathbb{E}\left[\sum_K \bold m _{K,N}| \mathbf{A}, \mathcal E_N\right]\Pr (\mathcal E_N)+\mathbb{E}\left[\sum_K \bold m _{K,N} 1(\mathcal E_N^c)|\mathbf{A}\right]\\
    &\leq N 2^N c_*^{N^2/4}+2N^3 c_*^N+ N \times \Pr(\mathcal E_N^c|\mathbf{A})
\end{align*}
where the inequality follows from the observation that $\sum_K \bold m _{K,N}\leq N$ almost surely,  by the separation lemma. Taking expectations, and recalling that $c_*<1$, we get:
$$ \mathbb{E}\left[\sum_K \bold m _{K,N}\right]\leq N 2^N c_*^{N^2/4}+2N^3 c_*^N+ \Pr(\mathcal E_N^c) \rightarrow 0,$$
as claimed.

\subsection{Proof of Proposition \ref{prop:RGG}}\label{proof:RGG}

For any  dyad $ij$,
\[
R_{ijt}\le \#\{k\neq i,j:\|z_k-z_i\|\le r,\ \|z_k-z_j\|\le r\}.
\]
The set on the right is contained in the intersection of two radius-$r$ balls, hence in a region of area at most $\pi r^2$. Since the square has area $N$, that dominating count is $\mathrm{Binomial}(N-2,p)$ with $p\le \pi r^2/N$, so it is tight, indeed $O_p(1)$. Thus, in sparse geometric graphs, the raw common-neighbor count is already locally bounded in probability.

We begin with the upper bound. If $\mathcal{S}_\bold i$ is identifying, then every spoke $i_1i_m$ must be feasible, because Definition \ref{def: identifying-star-system} requires
\[
D_{i_1i_m1}\neq D_{i_1i_m2},
\]
and that is impossible when $A_{i_1i_m}=-\infty$. Hence
\[
\Pr(Z_{\bold i,K}=1)
\le
\Pr\big(\|z_{i_1}-z_{i_m}\|\le r,\ m=2,\dots,K\big).
\]
Let $B(x,r)$ be the closed ball in $\mathbb{R}^2$ centered at $x$ with radius $r$. Conditioning on $z_{i_1}=x$,
\[
\Pr(\|z_{i_1}-z_{i_m}\|\le r \ \forall m\mid z_{i_1}=x)
=
\left(\frac{|B(x,r)\cap \left[0, \sqrt{N}\right]^2|}{N}\right)^{K-1}
\le
\left(\frac{\pi r^2}{N}\right)^{K-1}.
\]
Therefore
\[
\Pr(Z_{\bold i,K}=1)\le \left(\frac{\pi r^2}{N}\right)^{K-1}.
\]

We next turn to the lower bound. Let
\[
M:=\binom K2,
\]
the number of internal dyads among the $K$ star nodes. Consider the event $E_\bold i$ defined by:

\begin{enumerate}
\item the center is well inside the square:
\[
z_{i_1}\in [3r/2,\sqrt N-3r/2]^2;
\]

\item each leaf lies within distance $r/2$ of the center:
\[
z_{i_m}\in B(z_{i_1},r/2),\qquad m=2,\dots,K;
\]

\item no other agent lies in $B(z_{i_1},3r/2)$;

\item among the $M$ feasible internal dyads:
\begin{itemize}
\item at $t=0$, all are absent;
\item at $t=1$, all possible edges are absent;
\item at $t=2$, exactly the $K-1$ ``spoke-edges" are present and the remaining $M-(K-1)$ ``leaf-leaf edges" are absent.
\end{itemize}
\end{enumerate}

On $E_\bold i$, the candidate $K$-star is identifying:

\begin{itemize}
\item all spokes switch $1\to 2$, so Definition \ref{def: identifying-star-system}(ii) holds;
\item because all leaves are within $r/2$ of the center, any outside node that could connect to any star node would have to lie within $3r/2$ of the center, but event 3 rules that out;
\item hence the star has no outside feasible neighbors at either $t=1$ or $t=2$;
\item the period-1 neighborhood is empty, so Definition \ref{def: stable-neighborhood}(ii) is vacuous;
\item all non-spoke internal dyads are stable (absent) in $t=1$ and $t=2$, so Definition \ref{def: stable-neighborhood}(iii) holds.
\end{itemize}

Thus
\[
E_\bold i\subseteq \{Z_{\bold i,K}=1\}.
\]

We now compute the probability of $E_\bold i$. First, for large $N$,
\[
\Pr\!\left(z_{i_1}\in [3r/2,\sqrt N-3r/2]^2\right)
=
\frac{(\sqrt N-3r)^2}{N}
\ge \frac12.
\]

Second, conditional on that interior event and on $z_{i_1}$,
\[
\Pr\big(z_{i_m}\in B(z_{i_1},r/2),\ m=2,\dots,K\mid z_{i_1}\big)
=
\left(\frac{\pi r^2}{4N}\right)^{K-1}.
\]

Third, again conditional on the center being interior and the leaves lying in $B(z_{i_1},r/2)$,
\[
\Pr\big(\text{no other agent in }B(z_{i_1},3r/2)\big)
=
\left(1-\frac{9\pi r^2}{4N}\right)^{N-K},
\]
which is bounded below by a positive constant for all large $N$.

So the geometry part contributes
\[
\Pr(\text{items 1--3}) \ge c_{1,K}\,N^{-(K-1)}
\]
for some $c_{1,K}>0$.

Finally, conditional on items 1--3, all feasible dyads are internal to the $K$ nodes. On the event $E_{\bold i}$, all internal dyads are zero at $t=0$ and $t=1$, so every relevant common-neighbor count is  $
R_{ij0}=R_{ij 1}=0$.

The probability of item 4, conditional on items 1--3, is
\[
(1-\Lambda(a))^{M}
\cdot
(1-\Lambda(a))^{M}
\cdot
\Lambda(a)^{K-1}(1-\Lambda(a))^{M-(K-1)},
\]
which is a strictly positive constant depending only on $K$ and $a$. (Recall that $A_{ij}=a$ when nodes $i$ and $j$ are distance $r$ or less apart.) Therefore,
\[
\Pr(Z_{\bold i,K}=1)\ge \Pr(E_\bold i)\ge c_{1,K}(1-\Lambda(a))^{M}
\cdot
(1-\Lambda(a))^{M}
\cdot
\Lambda(a)^{K-1}(1-\Lambda(a))^{M-(K-1)}\,N^{-(K-1)}.
\]
This proves
\[
\Pr(Z_{\bold i,K}=1)\asymp N^{-(K-1)}.
\]

To conclude, note that the number of candidate $K$-stars is
\[
N\binom{N-1}{K-1},
\]
because the center is distinguished and the $K-1$ leaves are chosen from the remaining $N-1$ agents. Therefore
\[
\mathbb E[\bold{m}_{K,N}]
=
N\binom{N-1}{K-1}\Pr(Z_{\bold i,K}=1)
\asymp
N\cdot N^{K-1}\cdot N^{-(K-1)}
\asymp N.
\]
Hence the expected number of identifying $K$-systems is of order $N$. Summing over $K=2,\dots, \bar K$ yields
\[
\mathbb E\!\left[\sum_{K=2}^{ \bar K} \bold{m}_{K,N}\right]\asymp N
\]
for every fixed $\bar K$.
\end{proof}

\end{document}